\documentclass[11pt]{article} 

\usepackage[top=1in,bottom=1in,left=1in,right=1in]{geometry} 
\RequirePackage{amsthm,amsmath,amsfonts,amssymb}
\RequirePackage[numbers]{natbib}
\RequirePackage[colorlinks,allcolors=blue,urlcolor=blue]{hyperref}
\RequirePackage{graphicx}

\usepackage{color, amsmath, amssymb, amsbsy, amsthm, graphicx, bbm, adjustbox, amsfonts, bm, dsfont, courier, subfig}
\usepackage[toc,page]{appendix}
\usepackage{booktabs}
\usepackage[flushleft]{threeparttable}
\usepackage{latexsym}
\usepackage[dvipsnames]{xcolor}
\usepackage{comment}
\usepackage{wrapfig}
\usepackage{graphicx}
\usepackage{graphics}
\usepackage{algorithm}
\usepackage{algorithmic}
\usepackage{setspace}
\usepackage{enumitem}
\usepackage{mathabx}

\usepackage{soul}
\usepackage{threeparttable,booktabs}
\usepackage{etoolbox}
\appto\TPTnoteSettings{\footnotesize}
\newtheoremstyle{theoremstyle}
{\topsep} % Space above
{\topsep} % Space below
{\itshape} % Body font
{} % Indent amount
{} % Theorem head font
{} % Punctuation after theorem head
{.5em} % Space after theorem head
{\color{black}\ifthenelse{\equal{#3}{}}{{\bfseries #1 #2}}{{\bfseries #1 #2 (#3)}}}
\newtheoremstyle{theoremstylealt}
{\topsep} % Space above
{\topsep} % Space below
{\itshape} % Body font
{} % Indent amount
{} % Theorem head font
{} % Punctuation after theorem head
{.5em} % Space after theorem head
{\color{black}\ifthenelse{\equal{#3}{}}{{\bfseries #1 #2$^\prime$}}{{\bfseries #1 #2$^\prime$ (#3)}}}
\newtheoremstyle{examplestyle}
{\topsep} % Space above
{\topsep} % Space below
{} % Body font
{} % Indent amount
{} % Theorem head font
{} % Punctuation after theorem head
{.5em} % Space after theorem head
{\color{black}\ifthenelse{\equal{#3}{}}{{\bfseries #1 #2}}{{\bfseries #1 #2 (#3)}}}
\theoremstyle{theoremstyle}\newtheorem{thm}{Theorem}
\theoremstyle{theoremstylealt}
\theoremstyle{theoremstyle}     
\theoremstyle{theoremstyle}\newtheorem{lem}{Lemma}  
\theoremstyle{theoremstyle}\newtheorem{coro}{Corollary}        
\theoremstyle{theoremstyle}
\theoremstyle{theoremstyle}\newtheorem{assumption}{Assumption}
\theoremstyle{theoremstylealt}
\theoremstyle{theoremstyle}

\theoremstyle{theoremstyle}

\theoremstyle{theoremstyle}

\theoremstyle{examplestyle}\newtheorem{example}{Example}
\theoremstyle{examplestyle}\newtheorem{remark}{Remark}
\theoremstyle{examplestyle}

\usepackage{xr}

\renewcommand{\epsilon}{\varepsilon}

\newcommand{\br}{\mathbf{r}}

\def \hat{\widehat}

\newcommand{\bbeta}{\boldsymbol{\beta}}

\newcommand{\bgamma}{\boldsymbol{\gamma}}

\usepackage{titletoc}
\usepackage{enumitem}
\usepackage{subfig}

\usepackage{array}
\newcolumntype{H}{>{\setbox0=\hbox\bgroup}c<{\egroup}@{}}

\usepackage{minitoc}

\def \w{ \bm{w} }
\def \x{ \bm{x} }
\def \v{ \bm{v} }
\def \M{ \bm{M} }
\def \z{ \bm{z} }
\def \v{ \bm{v} }
\def \bl{ b_{\texttt{L}} }
\def \br{ b_{\texttt{R}} }
\def\check{\widecheck}

\def\bbeta{\bm{\beta}}
\def\bgamma{\bm{\gamma}}

\def\hat{\widehat}
\def\tilde{\widetilde}
\begin{document}

%%%%%%% for TOC
\doparttoc % Tell to minitoc to generate a toc for the parts
\faketableofcontents 
%%%%%%%%%%%%%%%%%%%%%%

\title{Inference on the Best Policies with Many Covariates
}
\author{
Waverly Wei\thanks{Division of Biostatistics, University of California, Berkeley.} \thanks{Co-first authors in alphabetical order} \and
Yuqing Zhou\thanks{Department of Statistics, University of Michigan} \footnotemark[2] \and Zeyu Zheng \thanks{Department of Industrial Engineering and Operations Research, University of California, Berkeley}\and   Jingshen Wang \footnotemark[1]   \thanks{Correspondence: jingshenwang@berkeley.edu} 
}

\date{}

\maketitle

\begin{abstract}
	Understanding the impact of the most effective policies or treatments on a response variable of interest is desirable in many empirical works in economics, statistics and other disciplines. Due to the widespread winner's curse phenomenon, conventional statistical inference assuming that the top policies are chosen independent of the random sample may lead to overly optimistic evaluations of the best policies.  In recent years, given the increased availability of large datasets, such an issue can be further complicated when researchers include many covariates to estimate the policy or treatment effects in an attempt to control for potential confounders.  In this manuscript, to simultaneously address the above-mentioned issues, we propose a resampling-based procedure that not only lifts the winner's curse in evaluating the best policies observed in a random sample, but also is robust to the presence of many covariates. The proposed inference procedure yields accurate point estimates and valid frequentist confidence intervals that achieve the exact nominal level as the sample size goes to infinity for multiple best policy effect sizes. We illustrate the finite-sample performance of our approach through Monte Carlo experiments and two empirical studies, evaluating the most effective policies in charitable giving and the most beneficial group of workers in the National Supported Work program.
		\end{abstract}
		
		\vskip 2cm
		\begin{center}\bfseries
% 		Preliminary and incomplete. 
% 		Please do not cite, circulate, or post on the Internet.
		\end{center}

\clearpage

\doublespacing

\section{Introduction}
 
\subsection{Motivation and our contribution} 

Many empirical work requires an understanding of the impact of the most effective policies or treatments on a relevant response variable of interest. For instance, in randomized (factorial) experiments with multiple treatments, researchers may be interested in the most effective policies (combinations). In online platforms, decision makers may be interested in the top five advertising strategies. In financial portfolio management, managers might want to learn about the best-performing strategies among many alternatives.  In practice, after different policy effect sizes are estimated from a random sample, researchers may naturally look into those policies with the largest effect sizes. Accurately measuring the performance of top policies allows policy makers to deliver better-informed decisions for forecasting the effects of future policy implementations.

Nevertheless, given the well-recognized ``winner's curse" phenomenon, there can be considerable uncertainties concerning if the top policies with large estimated effect sizes are indeed effective in the population (see Section \ref{Sec:intro-literature} for a literature review). 
In fact, due to the winner's curse phenomenon, literature documents that the estimated effect sizes of the best-performing policies without additional adjustments tend to be overly optimistic, rendering under-covered confidence intervals \citep*{lee2018winner, andrews2019inference}.  In this manuscript, we refer to the optimistic bias introduced by the winner's curse phenomenon as the winner's curse bias. To mitigate this bias issue, we focus on the problem of constructing accurate point estimates and valid confidence intervals for the true effect sizes of the (observed) best policies. By the best policies, we refer to a user-supplied number of policies that have the largest (estimated) effects among a set of candidate policies (see Section \ref{Sec:setup} for a 
concrete problem setup), as we would expect that in practice researchers might want to focus on a few top policies of interest.

Other than the winner's curse phenomenon discussed above, an additional consideration gains prominence in the evaluation of the most effective policies. Since policy (or intervention) variables are often not exogenous, researchers may adopt observational methods to estimate their effects. In recent years, given the increased availability of large datasets with rich covariate information, one commonly adopted approach in empirical works is to assume that the policy variables are exogenous after controlling for a sufficiently large set of factors or covariates. Such a consideration demandingly requires empirical researchers to estimate the policy effects in the presence of many covariates.

To simultaneously address the above-mentioned issues, in this article, we propose a procedure that not only is robust to the presence of many covariates, but also provides accurate point estimates and valid frequentist confidence intervals for multiple best policy effect sizes. By many covariates, we allow the number of covariates $q_n$ to diverge with the sample size $n$ as long as $\limsup_{n\rightarrow \infty} q_n/n <1$. Note that this does not rule out the cases where $q_n$ is fixed or $q_n = o(n)$. In other words, our inferential method remains valid when the dimension of the covariates $q_n$ is fixed or $q_n = o(n)$.
Our proposed confidence intervals are built upon resampling methods, and we demonstrate that they achieve exact nominal coverage as the sample size goes to infinity under fairly moderate assumptions. Our empirical evidence shows that conventional estimates ignoring the winner's curse issue are substantially upward biased, while our corrections reduce the winner's curse bias and increase coverage. As far as we know, valid statistical inferential tools on multiple best policies that lift the winner's curse while incorporating possibly many covariates have been lacking, and the contribution of our work is to bridge this gap and help policy makers deliver well-informed decisions in practice.

We illustrate our method with two empirical applications. In the first case study, we use the charitable giving data from \cite{karlan2007does} to evaluate the best pricing policies that motivate donors to give. Our results suggest that simple methods without adjusting for the winner's curse bias could be potentially overly optimistic in identifying the most effective polices. After accounting for the winner's curse bias, we do not find sufficient evidence to support that the second best pricing policy--asking the donor to give 25\% more than his/her highest historical donation--is effective, implying that asking for a more ``expensive" donation may not encourage donors to give. We nevertheless note that given our calibration only marginally reduces the effect size of the second best policy, the above conclusion might not warrant a different economic interpretation. 
In the second case study, we evaluate the effectiveness of the national supported work (NSW) program in different groups of workers. The NSW program is a job training program designed to prepare disadvantaged workers for employment, and it has been investigated in various studies \citep{dehejia2002propensity,lalonde1986evaluating}. We apply the proposed approach to evaluate the performance of the NSW program on the most-affected subgroups of workers observed in the dataset. Our study results potentially suggest that married black workers might benefit from the NSW program with an average increase of \$4,410 for their annual income.

\subsection{Connection to the existing literature}\label{Sec:intro-literature}

One fundamental trend that drives the motivation of the methodology developed in this manuscript is the increasing availability of massive datasets and the associated increasing dimensionality. Such a trend brings scientists opportunities to deliver better-informed policies but, at the same time, presents challenges in developing econometric and statistical tools; see \cite{fan2011sparse}, \cite{belloni2014inference}, \cite{fan2014challenges}, \cite{cattaneo2018inference} for example. A recent book \citep{fan2020statistical} provides a thorough discussion of analytical methods that aim to address such challenges. Specifically, the increasing data availability brings challenges and also opportunities to better understand various policies whose effects can be inferred from data. Along this line, our manuscript aims at providing understating for policies that are estimated and selected to be the most effective from a pool of policies.

The winner's curse phenomenon and its related issues have been widely recognized in economics, statistics, and data science at large. Seminal works by \cite{fan2016guarding, fan2018discoveries} point out that spurious discoveries can easily arise when target parameters are selected through data mining and statistical machine learning algorithms. Recent work by \cite{andrews2019inference} considers performing conditional and unconditional inference on observed best policy \textcolor{black}{and \cite{686054} extends the work to more general ranking problems, which is still different from our goal in conducting unconditional inference on multiple top policies. Moreover, while the conditional approaches in \cite{andrews2019inference} and \cite{686054} produce optimal confidence intervals for the observed policy effects, their point estimates and confidence intervals can be conservative when they are applied unconditionally.} 
\cite{efron2011tweedie} considers a method to handle the winner's curse bias with Tweedie's formula concerning the empirical Bayes theory. \cite{lee2018winner} consider a plug-in correction of the winner's curse bias and propose to construct confidence interval based on bootstrapping in the context of A/B testing, but the proposed method lacks theoretical justifications. 
In clinical trials for evaluating the largest observed treatment effect in multiple subpopulations, \cite{guo2020inference} propose a bootstrap-based confidence interval that achieves the exact nominal level as the sample size goes to infinity, though generalizing their method to make inference on several top policies might not be straightforward, especially in the presence of many covariates. %\cite{hall2010bootstrap, dezeure2017high} propose valid inference procedures based on simultaneous controls so the resulting inference procedures can be overall conservative. 

Our manuscript builds upon the literature on linear regression models with many or high dimensional covariates; see \cite{huber1973robust}, \cite{mammen1993bootstrap}, \cite{long2000using}, \cite{anatolyev2012inference}, \cite{el2013robust}, \cite{cattaneo2018alternative}, \cite*{cattaneo2019two}, \cite{jochmans2020heteroscedasticity} and the reference therein. In particular, \cite{mammen1993bootstrap} has established the asymptotic normality results for any contrasts of the ordinary least squares (OLS) coefficient vector estimator, when the dimension of the covariates divided by the sample size vanishes asymptotically. More recently, \cite{cattaneo2018inference} have shown that a small subset of the OLS estimators for the regression coefficients are asymptotically normal without restricting the dimension of the covariates to be  a vanishing fraction of the sample size. Moreover, \cite{cattaneo2018inference} have proposed a robust covariance matrix estimator for the subset of the the OLS estimator under fairly general conditions. \cite{jochmans2020heteroscedasticity} has proposed an alternative covariance matrix estimator that can deal with designs with even large number of covariates under additional assumptions (Assumption \ref{assumption:variance_estimation} in the current manuscript). 

Making inference on the best-performing policies is related to the literature on constructing confidence intervals for extrema parameters with bootstrap; see \cite{andrews2000inconsistency}, \cite{fan2007many}, \cite{xie2009confidence},   \cite{chernozhukov2013gaussian},  \cite{claggett2014meta} and the reference therein. Given the asymptotic distributions of extrema parameter estimators are often not normal, bootstrap-based methods can face serious difficulties when used to replicate the distribution of extrema of parameter estimators \citep{mammen1993bootstrap, mammen2012does}. While subsampling could overcome this issue faced by the classical bootstrap, it can exhibit very poor finite-sample performance because of the noise introduced by the vanishing subsample size. Different from our goal in constructing confidence intervals that achieve the exact nominal level, \citet{hall2010bootstrap} and \cite{chernozhukov2013gaussian} propose to construct conservative bootstrap confidence intervals for extrema of parameters. In our current problem setup with many covariates, the problem becomes even more acute as \cite{el2018can} show through a mix of simulation and theoretical analyses that the bootstrap is fraught with problems in moderate high dimensions. In the context of meta-analyses, \cite{claggett2014meta} propose an approach to make inference on ordered fixed study-specific parameters when different parameters are estimated independently from multiple studies.

Our method also contributes to the rapidly growing literature on program evaluations; see \cite{farrell2015robust}, \cite{belloni2014inference}, \cite{kueck2017estimation}, \cite{athey2017state}, \cite{abadie2018econometric}, \cite{chernozhukov2018double},  \cite{fan2021optimal},  \cite{vazquez2021identification} among many others. Under our asymptotic regime where the number of covariates $q_n$ grows with the sample size $n$, the Neyman orthogonalization based approaches often need to work with models with sparse regression coefficients \citep{belloni2014inference}. Rather than imposing such a sparsity assumption, our approach estimates the policy effects with regression adjustments without requiring the regression coefficients to be sparse. 
Because our approach only requires a consistent covariance matrix estimation for different policy effect estimators, we expect that the proposed framework on evaluating the best policies can be generalized when different policy effects are estimated with other off-shelf methods and we relegate such extensions for future work.

%\subsection{Organization and notations}

%The remainder of the manuscript is organized as follows. Section \ref{Sec:method-overview} introduces our problem setup, demonstrates the consequence of the winner's curse phenomenon, and describes the proposed procedure with some heuristics. We then discuss our theoretical investigation for the proposed procedure in Section \ref{Sec:theoretical-investigation}, followed by its simulation studies and practical implementation in Section \ref{Sec:simulation}. In Section \ref{Sec:real-data}, we discuss two empirical studies in depth. Finally, we summarize our method and discuss its potential impacts in Section \ref{Sec:discussion}.

\noindent \textit{Notation. } We work with triangular array data $\{ \omega_{i,n}: i=1, \ldots, n; n=1, 2,\ldots \}$ where for each $n$, $\{   \omega_{i,n}: i=1, \ldots, n \}$ is defined on the probability space $( \Omega, \mathcal{S}, P_n)$. All parameters that characterize the distribution of  $\{   \omega_{i,n}: i=1, \ldots, n \}$ are implicitly indexed by $P_n$ and thus by $n$. We write vectors and matrices in bold font, and use regular font for univariate variables and constants.

\section{Model setup and methodology}\label{Sec:method-overview}

\subsection{Problem setup and a revisit to the winner's curse phenomenon}\label{Sec:setup}
Suppose we have a random sample $\{ (y_{i,n}, \x_{i,n}', \w_{i,n}')' \}_{i=1}^n$, we pose the problem in the framework of a linear regression model under heteroscedasticity
\begin{align}\label{eq:linear-model}
y_{i,n} = \x_{i,n}'\bbeta +\w_{i,n}'\bgamma_n + u_{i,n}, \quad i=1, \ldots, n,
\end{align}
where $y_{i,n}$ is the outcome variable, $\x_{i,n}\in\mathbb{R}^{d}$ are the treatment or policy variables of interest, $\w_{i,n}\in \mathbb{R}^{q_n}$ contains the confounding factors, $ u_{i,n}$ is an unobserved error term, and the coefficient vector $\bbeta = (\beta_{1}, \ldots, \beta_{d})'$ contains the treatment effect of $\x_{i,n}$ on the outcome $y_{i,n}$. We allow the linear model \eqref{eq:linear-model} to hold approximately by allowing $\mathbb{E}[u_{i,n}| \{ \x_{i,n}\}_{i=1}^n, \{\w_{i,n}\}_{i=1}^n ]\neq 0$. We are also in a scenario where $\w_{i,n}$ is high-dimensional, in the sense that $q_n$ can be a vanishing fraction of the sample size $n$ as long as $\limsup_{n\rightarrow\infty} q_n/n <1$.
To simplify notations, we drop subscript $n$ in univariate random variables in the rest of the manuscript. That is, for example, we denote $\gamma_{j,n}$ as $\gamma_{j}$. 

We write the ordered values of $ \beta_{1}, \ldots, \beta_{d} $ as $ \beta_{(1)}\geq \ldots \geq \beta_{(d)}$. We adopt the ordinary least-squares (OLS) estimator $\hat{\bbeta}$ (see Remark \ref{remark:other-estimates} for other possible estimates) to estimate $\bbeta$ and write the order statistics of $\hat{\bbeta}$ as $ \hat{\beta}_{(1)}\geq \ldots \geq \hat{\beta}_{(d)}$. 
Because researchers in practice might hope to focus on a few top policies, given that $d_0$ is a user-supplied positive integer, our goal is to construct accurate point estimates and valid confidence intervals for two sets of quantities: 
\begin{enumerate}
    \item[(1)] the best policy effect sizes in the population: $\beta_{(1)}, \ldots,  \beta_{(d_0)}$, 
    \item[(2)] the observed best policy effect sizes: $\beta_{\hat{j}}$, where $  \hat{j} = \sum_{k=1}^{d} k\cdot \mathds{1}( \hat{\beta}_k = 	\hat{\beta}_{(j)} )$, for $j=1, \ldots, d_0$. 
\end{enumerate}
The first set of quantities characterizes the effects of the top $d_0$ policies in the population and are thus fixed parameters. The second set of quantities describes the true effect sizes of the best performing policies observed in the random sample, and these quantities are thus ``data-dependent parameters." Both sets of quantities can be of interest in different empirical applications \citep*{dawid1994selection, chernozhukov2013intersection, rai2018statistical}, and our proposed procedure can be used to deliver valid statistical inference on both quantities (Theorem \ref{Theorem:resampling-consistency} and Corollary \ref{corollary:observed-policy}). 

\begin{remark}[Ties in the estimated policy effects]
The second set of parameters is well defined if the observed policies do not have exact ties in the sense that $ \hat{\beta}_{(1)}> \ldots > \hat{\beta}_{(d)}$. When the policies effect estimators solve to the interior points of the feasible parameter space, it is likely that no exact ties appear in the random sample. On the other hand, there can exist scenarios where, for example, $d_0$ is set as 2 but there are multiple policy effect sizes that tie at rank 2. In this case, one may choose instead a data-dependent $\hat{d}_0 = \max \{k:\hat{\beta}_{(k)}> \hat{\beta}_{(2)} - C_1\cdot n^{-0.25}\}$. This new random $d_0$ will asymptotically be able to incorporate all the effect sizes that actually are equal to the true effect size associated with $\hat{\beta}_{(2)}$. In this way, the limiting value of $\hat{d}_0$ will not necessarily be 2, but can be a larger number than 2 to incorporate ``very close" effect sizes with the rank-2 effect size. We also provide some related discussions in Remark \ref{remark:data-dependent-choice-d}.
\end{remark}

\begin{remark}[Other possible estimators of $\bbeta$]\label{remark:other-estimates}
In the presence of many covariates when $q_n$ is potentially large ($\limsup_{n\rightarrow \infty}q_n/n\rightarrow 1$ in our asymptotic regime) without assuming the coefficient $\mathbf{\gamma}_n$ to be sparse, we adopt the OLS estimator to estimate $\bbeta$, because the OLS estimator has been thoroughly studied in the existing literature and enjoys favorable theoretical guarantees. 
In high dimensions when $q_n \gg n$, other estimators of $\bbeta$ that incorporate model selection procedures can be adopted as well. Our procedure can produce valid statistical inference as long as the covariance matrix of $\hat{\bbeta}$ can be consistently estimated. For example, under the sparsity assumption on $\bm{\gamma}_n$ documented in the literature \citep{fan2020statistical}, we may adopt the covariance matrix estimator from the de-sparsified Lasso procedure \citep{van2014asymptotically, zhang2014confidence}.  
\end{remark}

To fully realize the challenges on delivering valid statistical inference on these two sets of parameters in our current problem setup, we revisit the winner's curse phenomenon. When first discussed in common-value auctions, the winner's curse refers to the bidding behavior where bidders systematically overbid, resulting in an expected loss \citep{charness2009origin}. In our context of policy evaluations, the winner's curse refers to the issue that the observed best policies have the tendency to over-estimate the best policies in the population. We would thus often expect that neither $\mathbb{E}[ \hat{\beta}_{(j)} - \beta_{(j)} ] $ nor $\mathbb{E}[ \hat{\beta}_{(j)} - \beta_{\hat{j}} ]$ is close to zero, and the resulting confidence interval may fail to reach the nominal level. Such an issue becomes even more acute as we have many covariates $\w_{i,n}$ entering the inferential process. 

We next illustrate the winner's curse issue through Example \ref{Example:winner-curse} with a simple simulation study, where we observe substantial winner's curse bias and under-covered confidence intervals for the top polices. In particular, Figure \ref{Example:winner-curse}(b) demonstrates that coverage probabilities are worsened when a larger number of covariates are incorporated for estimating $\bbeta$. It is worth pointing out that when $d=3$, $\hat\beta_{(2)}$ is the median policy effect. Thus, the estimation bias is around 0, and the true standard deviation is much smaller than the estimated standard deviation, resulting in a confidence interval with close to 100\% coverage. When $d$ increases, the coverage probability gradually drops due to a larger estimation bias and inaccurately estimated standard deviation.

\begin{example}[A simulation study demonstrating the winner's curse phenomenon with many covariates]\label{Example:winner-curse}
 We generate 1000 Monte Carlo samples following the setup in Model \eqref{eq:linear-model}. We generate $\x_{i,n}\sim \mathcal{N}(0, \bm{\Sigma}) $ with $\Sigma_{jk} = 0.5^{|j-k|}$ for $j,k=1, \ldots, d$, $\w_{i,n} = \mathds{1}( \tilde{\w}_{i,n}\geq \Phi^{-1}(0.98) )$ with $\tilde{\w}_{i,n} \sim \mathcal{N}(0, \bm{I}_{q_n})$, where $\bm{I}_{q_n}$ is a $q_n$-dimensional identity matrix. We consider the case where no policy is effective (so that $\bbeta = 0$, $\beta_{\hat{1}} = \beta_{\hat{2}} = 0$) and $\gamma_{j} = 1/j$, for $j=1, \ldots, q_n$. We report the asymptotic bias of the conventional estimator (i.e., $\sqrt{n}\cdot \mathbb{E}[ \hat{\beta}_{(j)} - \beta_{(j)} ] $) as well as the coverage probability of confidence intervals constructed based on normal approximation with the Eicker-White \citep{eicker1963asymptotic, white1980heteroskedasticity} covariance matrix estimator defined in Eq \eqref{Eq:eicker-white}. 
\end{example}

\begin{figure}[!t]
	\centering  
	\subfloat[Winner's curse bias ($q_n=141$)]{
		\label{subfig:bias}
		\includegraphics[height=0.4\linewidth]{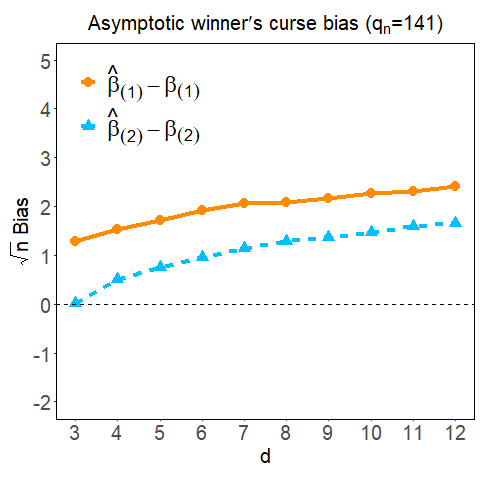}}
	\subfloat[Coverage probability]{
		\label{subfig:cov1}
		\includegraphics[height=0.4\linewidth]{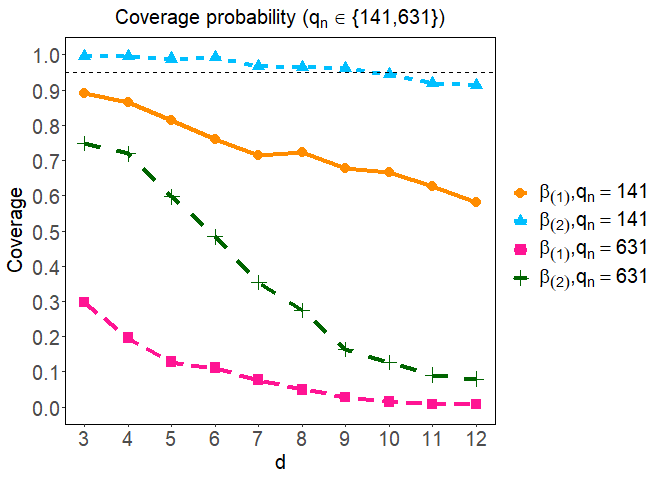}}
	\caption{Demonstration of the winner's curse phenomenon following the simulation setup in Example \ref{Example:winner-curse}. \textcolor{black}{The maximum Monte Carlo standard error for the asymptotic bias is 0.88.} Panel (a) captures the asymptotic winner's curse bias when $q_n = 141$; Panel (b) captures the coverage probability when $q_n \in \{ 141, 631 \}$ and the nominal level is 0.95.  }
	\label{fig:winner curse}
\end{figure}

\subsection{Methodology}\label{Sec:Method}

Our method starts with the ordinary least-squares (OLS) estimator of $\bbeta$, that is
\begin{align*}
\hat{\bbeta} = \Big( \sum_{i=1}^n \hat{\v}_{i,n}\hat{\v}_{i,n}' \Big)^{-1} \Big( \sum_{i=1}^n \hat{\v}_{i,n} y_{i,n} \Big), 
\end{align*}
where $\hat{\v}_{i,n} = \sum_{j=1}^n (\M_{n})_{i,j} \x_{j,n}$, and $ (\M_{n})_{i,j} \triangleq \mathds{1}( i=j ) -  \w'_{i,n}\Big( \sum_{k=1}^n \w_{k,n}\w_{k,n}' \Big)^{-1} \w_{j,n}$. 
As we focus on the case when $q_n$ can be a non-vanishing fraction of $n$, $n\rightarrow\infty$, we adopt the robust covariance matrix estimator proposed in \cite{jochmans2020heteroscedasticity}. We try to follow the author's notation as closely as possible: 
\begin{align*}
    \hat{\bm{\Omega}}^{\texttt{KJ}}_n \triangleq \hat{\bm{\Gamma}}_n^{-1}\hat{\Sigma}^{\texttt{KJ}}_n \hat{\bm{\Gamma}}_n^{-1},
\end{align*}
where 
\begin{align*}
\hat{\bm{\Gamma}}_n=\frac{1}{n}\sum_{i=1}^n\hat{\v}_{i,n}\hat{\v}_{i,n}',\quad \hat{\Sigma}^{\texttt{KJ}}_n \triangleq \frac{1}{n}\sum_{i=1}^n \hat{\v}_{i,n}\hat{\v}_{i,n}'  y_{i,n} \acute{u}_{i,n}, 
\end{align*}
where $\acute{u}_{i,n} = \frac{ \hat{u}_{i,n}  }{(\M_n)_{i,i} }$, $\hat{u}_{i,n} = \sum_{j=1}^n (\M_n)_{i,j}(y_{j,n} -\x_{j,n}'\hat{\bbeta} )$, for $i=1, \ldots, n$. Such an estimator is well-defined as long as $\min_{i}(\M_n)_{i,i}  >0 $. If $\min_{i}(\M_n)_{i,i}=0$, it means that the auxiliary regression produces a perfect prediction. So the observation does not carry information on $\bbeta$ and can be ignored.

As shown in Example \ref{Example:winner-curse}, the estimated top policy effect sizes with $\hat{\beta}_{(1)},\ldots, \hat{\beta}_{(d_0)}$ are often biased upward for our target parameters due to the winner's curse phenomenon. Inspired by the procedure proposed by \cite{claggett2014meta}\footnote{Note that there is a typo in \cite{claggett2014meta} for the definition of the near tie set. Although the near tie $\mathcal{H}_{(j)}$ in their manuscript was originally defined as $\mathcal{H}_{(j)} = \big\{ k: | \beta_k  - \beta_{(j)} | = O(n^{-\frac{1}{2}}),\ k=1, \ldots, d \big\}$, their proof goes through when the near tie set is defined with $  \mathcal{H}_{(j)} = \big\{ k: | \beta_k  - \beta_{(j)} | = o(n^{-\frac{1}{2}}),\ k=1, \ldots, d \big\}$. } for meta-analyses,  
% of $\beta_{(j)}$ in a data adaptive fashion (see Section \ref{Sec:heuristics-example} for detailed discussion): 
% \begin{align*}
%     \mathcal{H}_{(j)} = \big\{ k: | \beta_k  - \beta_{(j)} | = o(n^{-\frac{1}{2}}),\ k=1, \ldots, d \big\}.
% \end{align*}
%To do so,
we generate replicates of $\hat{\bbeta}$ from a multivariate normal distribution
\begin{align}\label{eq:resample-step}
\hat{\bm{\bbeta}}^*   \big| \{ (y_{i,n}, \x_{i,n}', \w_{i,n}')'\}_{i=1}^n \sim \mathcal{N}( \hat{ \bbeta}, \hat{\bm{\Omega}}^{\texttt{KJ}}_n /n ), \quad\text{ where } \hat{\bm{\bbeta}}^*  = ( \hat{\beta}^*_1, \ldots, \hat{\beta}^*_d )',
\end{align}
and we denote the ordered values of the vector $\hat{\bm{\bbeta}}^* $ as $  \hat{\beta}^*_{(1)}\geq  \ldots\geq \hat{\beta}^*_{(d)} $. Note that the above description of $\hat{\bm{\bbeta}}^* $ differs from some previous work on bootstrapping insofar we have suppressed the role of ``multiplier variables," and we have defined $\hat{\bm{\bbeta}}^*$ as a sample from $\mathcal{N}( \hat{ \bbeta}, \hat{\bm{\Omega}}^{\texttt{KJ}}_n /n )$. Different from \cite{claggett2014meta} that requires different estimators to be estimated from independent studies with non-overlapping random samples, our approach relaxes such an requirement and allows $\hat{\beta}_1, \ldots, \hat{\beta}_d$ to be correlated. 

Next, given properly chosen $\bl$ and $\br$ so that $\bl -\br = O(n^{-\delta})$ with $\delta\in (0,\frac{1}{2})$ (see Supplementary Materials Section C.1 for their data-adaptive choices, and robustness to different choices of tuning parameters in Supplementary Materials, Section C.2), we estimate a ``near tie" set that captures policies that have similar effect sizes to the $j$-th largest policy:
\begin{align*}
 \widehat{\mathcal{H}}_{(j)} = \big\{k: \hat{\beta}^*_{ {{(j)}}} - \bl \leq \hat{\beta}^*_{k} \leq \hat{\beta}^*_{ {{(j)}}}+\br  ,\ k=1, \ldots, d \big\}. 
\end{align*}
We then record the averages of $\hat{\beta}^*_1, \ldots, \hat{\beta}^*_{d} $ and of $\hat{\beta}_1, \ldots, \hat{\beta}_{d}$ in the estimated tie set $ \widehat{\mathcal{H}}_{(j)}$ as 
\begin{align}\label{eq:weighted-average}
\textcolor{black}{\tilde{\beta}_{(j)}^* = \frac{ \sum_{k\in  \widehat{\mathcal{H}}_{(j)}}\hat{\beta}^*_k }{ | \widehat{\mathcal{H}}_{(j)} |}}, \text{ and } \tilde{\beta}_{(j)} = \frac{ \sum_{k\in  \widehat{\mathcal{H}}_{(j)}} \hat{\beta}_k }{ | \widehat{\mathcal{H}}_{(j)} |},
\end{align}
where $| \widehat{\mathcal{H}}_{(j)} |$ denotes the cardinality of the set $\widehat{\mathcal{H}}_{(j)}$.

Finally, we apply the above resampling procedure to construct point estimates and confidence intervals for $\beta_{(j)}$ as well as $\beta_{\hat{j}}$ (as defined in Section \ref{Sec:setup} and in Eq \eqref{eq:observed-best-policy}), $j=1, \ldots, d_0$. Specifically, for confidence interval construction, we generate $B$ independent samples of $\tilde{\beta}^*_{(j)}$ as in Eq \eqref{eq:weighted-average}, and then define $\hat{q}_{(j)}(\alpha/2)$ to be the empirical $\alpha/2$-quantile of the $B\geq 1$ samples (and similarly for $\hat{q}_{(j)}(1-\alpha/2)$), leading to a level-$\alpha$ confidence interval for $\beta_{(j)}$ with 
\begin{align*}
\big[ \hat{q}_{(j)}(\alpha/2),\ \hat{q}_{(j)}(1-\alpha/2) \big], \quad j=1, \ldots, d_0.
\end{align*}
Corollary \ref{corollary:observed-policy} demonstrates that the above confidence interval also serves as an asymptotically exact level-$\alpha$ prediction interval for $\beta_{\hat{j}}$. For point estimates, we may either use $\tilde{\beta}_{(j)}$ or the averaged resampled statistics $\tilde{\beta}_{(j)}^*$ to estimate $\beta_{(j)}$ and $\beta_{\hat{j}}$.

\section{Theoretical investigation}\label{Sec:theoretical-investigation}

\subsection{Notations and assumptions}

Before discussing the theoretical results in detail, we revisit and introduce some notations and assumptions adopted in the manuscript.
We denote the sample $\{ (y_{i,n}, \x_{i,n}', \w_{i,n}')'\}_{i=1}^n$ as $\{ \z_{i,n}\}_{i=1}^n$. Recall $u_{i,n}$ is the random error in the considered linear model \eqref{eq:linear-model}, we define 
\begin{align}\label{eq:varepsion-v-definition}
\varepsilon_{i,n} = u_{i,n} - \mathbb{E}[ u_{i,n} |  \{ \w_{i,n} \}_{i=1}^n, \{ \x_{i,n} \}_{i=1}^n],\quad \v_{i,n} = \x_{i,n} - \mathbb{E}[\x_{i,n} |  \{ \w_{i,n} \}_{i=1}^n], \quad i=1, \ldots, n.  
\end{align}
Let $e_{i,n} = \mathbb{E}[ u_{i,n} |  \{ \w_{i,n} \}_{i=1}^n, \{ \x_{i,n} \}_{i=1}^n]$, we further denote 
\begin{align*}
&  \sigma_{i,n}^2 = \mathbb{E}[\varepsilon_{i,n}^2 |  \{ \w_{i,n} \}_{i=1}^n, \{ \x_{i,n} \}_{i=1}^n ], \quad  \tilde{\v}_{i,n} = \sum_{j=1}^n (\bm{M}_n)_{i,j}\v_{j,n}, \\ 
&  \rho_{n}^{1} = \frac{ \sum_{i=1}^n\mathbb{E}\big[ e_{i,n}^2  \big]  }{n},\quad  \rho_{n}^{2} = \frac{ \sum_{i=1}^n\mathbb{E}\Big[ \mathbb{E}\big(e_{i,n} |  \{ \w_{i,n} \}_{i=1}^n \big)^2  \Big]  }{n}, \\
&\bm{Q}_{i,n} = \mathbb{E}\Big[ \x_{i,n} - \big(\sum_{j=1}^n \mathbb{E}[ \w_{j,n}\w_{j,n}' ]  \big)^{-1} \sum_{j=1}^n \mathbb{E}[ \w_{j,n}\x'_{j,n} ]   \Big|  \{ \w_{i,n} \}_{i=1}^n  \Big], \\
& \tilde{\bm{Q}}_{i,n} = \sum_{j=1}^n ( \bm{M}_n)_{i,j}\bm{Q}_{i,n} .
\end{align*}
For a policy $j$, we define the near tie set in the population as: 
\begin{align*}
\mathcal{H}_{(j)} = \big\{ k: | \beta_k  - \beta_{(j)} | = o(n^{-\frac{1}{2}}),\ k=1, \ldots, d \big\}.
\end{align*}
% This suggests that $\forall k \in \mathcal{H}_{(j)} $, there exist a sequence $\delta_{n} \rightarrow 0$ as $n\rightarrow 0$, such that 
% \begin{align*}
% \beta_k = \beta_{(j)} + n^{-\frac{1}{2}} \cdot \delta_{n}, \quad \forall k \in \mathcal{H}_{(j)} .
% \end{align*}
Next, 
let $\hat{\bm{e}}_j $ denote a $d$-dimensional (sparse) vector with
\begin{align*}
\hat{\bm{e}}_j = (\hat{e}_{j,1}, \ldots, \hat{e}_{j, d}),\quad  \hat{e}_{jk} = \frac{\mathds{1}( k \in  \widehat{\mathcal{H}}_{(j)} )}{ |\hat{\mathcal{H}}_{(j)} | }, \quad k=1, \ldots, d.
\end{align*}
We will use the notation $\mathbb{P}(\cdot | \{ \z_{i,n}\}_{i=1}^n)$ to refer to the probability that is conditional on the random variables $\{ \z_{i,n}\}_{i=1}^n$.

We make following assumptions throughout this section. Note that Assumptions \ref{Assumption:Sampling}-\ref{assumption:variance_estimation} listed below largely follow the assumptions in \cite{cattaneo2018inference} and \cite{jochmans2020heteroscedasticity}, we list these assumptions along with their interpretations to present a full picture for our readers. 

\begin{assumption}[Sampling] \label{Assumption:Sampling}
The errors $\varepsilon_{i,n}$ are uncorrelated across $i$ conditional on $\{ \x_{i,n} \}_{i=1}^n$ and $\{ \w_{i,n} \}_{i=1}^n$. Let $\{ N_1, \ldots, N_{G_n} \}$ represents a partition of $\{1,\ldots, n\}$ with $\underset{ 1\leq g\leq G_n}{\max}{ |N_g| } = O(1)$ such that $ \{ (\varepsilon_{i,n}, \v_{i,n} ), i\in N_g \} $ (defined in \eqref{eq:varepsion-v-definition}) are independent across $g$ conditional on $\{ \w_{i,n} \}_{i=1}^n$. 
\end{assumption}

Assumption \ref{Assumption:Sampling}  generalizes the classical independent and identically distributed (i.i.d.) setting to allow for repeated measurements or group structures in the observed data. For example, Assumption \ref{Assumption:Sampling} allows the observed data to form clusters of finite sample sizes, and within-cluster dependency is allowed as long as the observations between clusters are independent. 

\begin{assumption}[Design] \label{Assumption:Design}
The dimension of the covariates $\w_{i,n}$ satisfies that $\limsup_{n\rightarrow \infty} q_n/n < 1$. The minimum eigenvalue of the matrix $\sum_{i=1}^n \w_{i,n}\w_{i,n}'$ is bounded away from 0 with probability approaching one, that is 
\begin{align*}
 \lim_{n\rightarrow \infty}\mathbb{P}\Big( \lambda_{\min} \big( \sum_{i=1}^n \w_{i,n}\w_{i,n}'\big) >0 \Big) = 1. 
\end{align*}
Lastly, 
\begin{align*}
& \underset{1\leq i\leq n}{\max}\Big\{  \mathbb{E}[\varepsilon_{i,n}^4 |  \{ \w_{i,n} \}_{i=1}^n, \{ \x_{i,n} \}_{i=1}^n ],\ \frac{1}{\sigma_{i,n}^2},\\
& \quad \quad\quad \quad  \mathbb{E}[\v_{i,n}^4 | \{ \w_{i,n} \}_{i=1}^n ], \ 1/\lambda_{\min} \big( \frac{\sum_{i=1}^n \mathbb{E}\big[  \tilde{\v}_{i,n}\tilde{\v}_{i,n}'|  \{ \w_{i,n} \}_{i=1}^n \big] }{n} \big) \Big\} = O_p(1).
\end{align*}
\end{assumption}

Assumption \ref{Assumption:Design} contains three conditions. The first condition allows the dimension of the covariates $\w_{i,n}$ to grow at the sample rate as the sample size $n$. 
The second condition requires the matrix $\sum_{i=1}^n \w_{i,n}\w_{i,n}'$ to be full rank, which is necessary otherwise the OLS estimator would not be able to calculate the matrix $\bm{M}_n$. Furthermore, as noted in \cite{cattaneo2018inference}, such an assumption can be imposed by dropping any covariates in $\w_{i,n}$ that are collinear. 
The third condition contains conventional moment conditions for the covariates and heteroscedasticity. 

\begin{assumption}[Linear model approximation] \label{Assumption:Approximation}
 $ \sum_{i=1}^n \mathbb{E}[|| \bm{Q}_{i,n}||^2 ] /n = O(1)$, $\rho_{n}^{1} + n( \rho_{n}^{1} - \rho_{n}^{2} ) +  \rho_{n}^{1} \cdot \sum_{i=1}^n \mathbb{E}[|| \bm{Q}_{i,n}||^2 ]  = o(1)$, and $\underset{1\leq i\leq n}{\max} || \hat{\v}_{i,n}||/\sqrt{n} = o_p(1)$, $n\rho_n^{1} = O(1)$. 
\end{assumption}

Assumption \ref{Assumption:Approximation} mainly characterizes the difference between the mean squares of the conditional errors $\rho_n^1$ and the projection  $\rho_n^2$ into the covariate space $\{\w_{i,n}\}$'s. The characterization of this difference involves  $\sum_{i=1}^n \mathbb{E}[|| \bm{Q}_{i,n}||^2 ]$ where $\bm{Q}_{i,n}$ describes the deviation of $\x_{i,n}$ from its population linear projection. Residuals of this linear projection, represented by $\hat{\v}_{i,n}$'s, are assumed to satisfy a negligibility condition after a maximization over all $i$'s. This negligibility condition regularizes the distributional connection between $\x_{i,n}$'s  and $\w_{i,n}$'s. We note that if the mean squares of $\x_{i,n}$'s are bounded and that an exogeneity condition $e_{i,n}=0$ holds for all $i$ and $n$, then the linear model approximation assumption naturally holds. Otherwise, if the exogeneity condition does not hold, Assumption \ref{Assumption:Approximation} requires a small-bias condition $n\rho_n^{1} = O(1)$. 

\begin{assumption}[Variance estimation] $\lim_{n\rightarrow \infty}\mathbb{P}( \min_i  \big( \bm{M}_n \big)_{i,i} >0  ) = 1$, 
\begin{align*}
\mathbb{P}\Big( \min_i \big( \bm{M}_n \big)_{i,i}>0 \Big) = O_p(1), \quad \frac{ \sum_{i=1}^n || \tilde{ \bm{Q} }_{i,n}||^4 }{n} = O_p(1),
\end{align*}
and $\max_i ||\mu_{i,n}||/\sqrt{n} = o_p(1)$ with $\mu_{i,n} = \mathbb{E}\big[ y_{i,n} | \{\x_{i,n} \}_{i=1}^n, \{\w_{i,n}\}_{i=1}^n \big]$. 
\label{assumption:variance_estimation}
\end{assumption}
Assumption \ref{assumption:variance_estimation} has two major parts. The first part regularizes the diagonal elements $ (\bm{M}_n \big)_{i,i}$'s, essentially requiring the smallest diagonal element to be consistently bounded away from zero when $n$ tends to infinity. Even though it is difficult to provide broadly general primitives to validate this assumption, Assumption 2 of \cite{cattaneo2018inference}, Assumption 4 of \cite{jochmans2020heteroscedasticity}, and the discussions therein provide sufficient conditions for this assumption to hold. The second part regularizes $\mu_{i,n}$'s and $\tilde{\mathbf{Q}}_{i,n}$'s in order to control the variance of $y_{i,n}$'s and the variance of $ \mathbb{E} (\v_{i,n}|\{\w_{i,n}\}_{i=1}^n)$'s.

\begin{assumption}[Policy effect sizes]\label{Assumption:seperation}
For $\delta \in (0,\frac{1}{2})$, the asymptotic distance between the effects of policy $k\not\in \mathcal{H}_{(j)} $ and $j\in  \mathcal{H}_{(j)}$ diverges as $n\rightarrow \infty$: 
\begin{align*}
n^{\delta} \cdot \min_{k\not\in \mathcal{H}_{(j)}}\big| \beta_{(j)} - \beta_k \big| \rightarrow \infty, \text{ as }n\rightarrow \infty, \quad j=1, \ldots, d. 
\end{align*}
\end{assumption}

Assumption \ref{Assumption:seperation} requires that any policies outside the near tie set $\mathcal{H}_{(j)}$ have effect sizes sufficiently different from the ones in  $\mathcal{H}_{(j)}$. In fixed dimensions when $q_n$ does not grow with $n$, the underlying policy effect sizes $\beta_1, \ldots, \beta_d$ are constant with respect to the sample size $n$. The near tie set reduces to a ``precise" tie set $\mathcal{H}_{(j)} = \big\{ k:  \beta_k  = \beta_{(j)} ,\ k=1, \ldots, d \big\}$, suggesting that $\min_{k\not\in \mathcal{H}_{(j)}}\big| \beta_{(j)} - \beta_k \big|$ is a positive constant bounded away from zero. In such a case, Assumption \ref{Assumption:seperation} is automatically satisfied.

\subsection{Properties of the proposed estimator}\label{Sec:theoretical-results}

For the proposed estimator, we show that the following theorem holds: 

\begin{thm}\label{Theorem:resampling-consistency}
Under Assumptions \ref{Assumption:Sampling}-\ref{Assumption:seperation}, for any $t\in \mathbb{R}$, for the resampled statistics, the following holds 
\begin{align*}
\lim_{n\rightarrow \infty}\mathbb{P}\left(  \frac{ \sqrt{n}\big(\tilde{\beta}_{(j)}^* - \tilde{\beta}_{(j)}\big) }{ (\hat{\bm{e}}_j'\hat{\bm{\Omega}}^{\texttt{KJ}}_n \hat{\bm{e}}_j  )^{\frac{1}{2}} } \leq t\Big | \{ (y_{i,n}, \x_{i,n}', \w_{i,n}')' \}_{i=1}^n \right) = \Phi(t).
\end{align*}
For the original statistics, it holds that 
\begin{align*}
 \lim_{n\rightarrow \infty} \mathbb{P}\left(  \frac{ \sqrt{n}\big(\tilde{\beta}_{(j)} - {\beta}_{(j)}\big) }{ (\hat{\bm{e}}_j'\hat{\bm{\Omega}}^{\texttt{KJ}}_n \hat{\bm{e}}_j  )^{\frac{1}{2}} } \leq t\right) = \Phi(t). 
\end{align*}
Furthermore, we have that $\lim_{n\rightarrow \infty }\mathbb{P}\Big( 	\mathbb{P}\big( 	\tilde{\beta}_{(j)}^* \leq  \beta_{(j)}   |\{ \z_{i,n}\}_{i=1}^n   \big) \leq s \Big)  = s$. 
\end{thm}

Theorem \ref{Theorem:resampling-consistency} confirms that our proposed confidence interval for $\beta_{(j)}$ achieves exact $1-\alpha$ coverage probability as the sample size goes to infinity when $B$ is sufficiently large, which distinguishes the proposed inference procedure from simultaneous methods. Furthermore, Theorem \ref{Theorem:resampling-consistency} says that $\tilde{\beta}_{(j)}$ is a root-$n$ consistent estimator of $\beta_{(j)}$, in the sense that $\forall \varepsilon >0$, there exists $M > 0$ such that $\mathbb{P}\big(|\sqrt{n}( \tilde{\beta}_{(j)} - {\beta}_{(j)} )| >M \big) \leq \varepsilon$, for $n\geq 1$. 
% , that is,  
% \begin{align*}
%  \lim_{n\rightarrow\infty} \mathbb{P} \Big( \beta_{(j)}\in [ \widehat{F}_{(j)}^{-1}(\alpha/2) ,\ \widehat{F}_{(j)}^{-1}(1-\alpha/2) ]\Big) = 1-\alpha, \quad j=1, \ldots, d_0.
% \end{align*}

As for the observed best policies, recall that we denote the observed $j$-th largest policy as 
\begin{align}\label{eq:observed-best-policy}
        \hat{j} = \sum_{k=1}^d k\cdot \mathds{1}( \hat{\beta}_k = 	\hat{\beta}_{(j)} ).
\end{align}
The following corollary suggests that the proposed confidence interval for $\beta_{(j)}$ can also serve as an exact prediction interval for $\beta_{\hat{j}}$. Therefore, the proposed procedure in Section \ref{Sec:Method} can also be used to make inference on the observed top policies in a random sample: 
\begin{coro}\label{corollary:observed-policy}
Under Assumptions \ref{Assumption:Sampling}-\ref{Assumption:seperation}, we have that  $\lim_{n\rightarrow \infty }\mathbb{P}\Big( 	\mathbb{P}\big( 	\tilde{\beta}_{(j)}^* \leq  \beta_{\hat{j}}   |\{ \z_{i,n}\}_{i=1}^n   \big) \leq s \Big)  = s$. Furthermore, $\tilde{\beta}_{(j)}$ is a ``root-$n$ consistent" estimator of the data-dependent parameter $\beta_{\hat{j}}$ in the sense that $\forall \varepsilon >0$, there exists $M > 0$ such that $\mathbb{P}\big(|\sqrt{n}( \tilde{\beta}_{(j)} - {\beta}_{\hat{j}} )| >M \big) \leq \varepsilon$, for $n\geq 1$. 
\end{coro}

\begin{remark}[Regression models with fixed effects]
The proposed resampling-based approach can be used to calibrate multiple best policies when fixed effects are introduced in linear regression models (see \cite{verdier2020estimation} for comprehensive discussion). This suggests that our approach not only applies to independently sampled data, but also remains valid when there are repeated-measurements present in the data. These may include short panel data, and datasets in which, for example, two individuals have sampled from each household. To conserve space in the main manuscript, we have leave the detailed discussion in the Supplementary Materials (Section D). 
\end{remark}

\begin{remark}[Data dependent choice of $d_0$]\label{remark:data-dependent-choice-d}
In addition to a deterministic choice of $d_0$, another practically relevant scenario is a data dependent choice of $d_0$. An example of such a data dependent choice is $\hat{d}_0 = \max \{k:\hat{\beta}_{(k)}> C\}$, where $C$ is a user-specified threshold for the effect size. A relatively complicated situation is that $C$ coincides with some of the policy sizes in $\beta_1,\beta_2,\cdots,\beta_d$. In this situation, it is possible that no matter how large $n$ is, $\hat{d}_0$ does not converge to a deterministic value but instead to a non-degenerate random variable. For the purpose of separation, we may adjust  $\hat{d}_0 = \max \{k:\hat{\beta}_{(k)}> C\}$ to be  $\hat{d}'_0 = \max \{k:\hat{\beta}_{(k)}> C + C_1\cdot n^{-0.25}\}$, where $C_1$ is a constant that does not depend on $n$. The choice of $-0.25$ is tunable and may be of independent interest. By this new choice of $\hat{d}'_0$, the policy effects that exactly equal $C$ will be eliminated almost surely when $n$ tends to infinity. This elimination exactly matches the target to select all the policy sizes that are larger than $C.$ In the limit of $n$ tending to infinity, $\max \{k:\hat{\beta}_{(k)}> C + C_1\cdot n^{-0.25}\}$ will converge almost surely to a set that contains all effect sizes larger than $C$. Therefore, the large-sample theory results for a pre-specified deterministic integer would still hold by plugging in $\hat{d}'_0 $. 
\end{remark}

\section{Simulation studies}\label{Sec:simulation}

\subsection{Simulation design}\label{subsec:sim-design}

 We generate i.i.d. Monte Carlo samples of $\{( y_{i,n}, \x_{i,n}',\w_{i,n}') \}_{i=1}^n$ from the model
\begin{align*}
y_{i,n} = \x_{i,n}'\bbeta +\w_{i,n}'\gamma_n + \varepsilon_{i,n}, \quad i=1, \ldots, n. 
\end{align*}
\textcolor{black}{We consider various data generating processes (DGP) for different choices of the policy variable $\x_{i,n}$, the covariates $\w_{i,n}$ and the random noise $\varepsilon_{i,n}$. } The first DGP follows a similar setup taking from \cite{jochmans2020heteroscedasticity} and \cite{cattaneo2018inference}, where we generate many (sparse) dummy variables entering the estimation of $\bbeta$. We generate $\x_{i,n}\sim \mathcal{N}(0, \Sigma) $ with $\Sigma_{jk} = 0.5^{|j-k|}$ for $j,k=1, \ldots, d$, $\w_{i,n} = \mathds{1}( \tilde{\w}_{i,n}\geq \Phi^{-1}(0.98) )$ with $\tilde{\w}_{i,n} \sim \mathcal{N}(0, \bm{I}_{q_n})$ and $\bm{I}_{q_n}$ is a $q_n$-dimensional identity matrix, and $\varepsilon_{i,n}\sim \mathcal{N}(0,1)$. The second DGP considers a case with dummy policy random variables and normal covariates, where we generate $\x_{i,n} = \mathds{1}(\tilde{\x}_{i,n} > 0) $ with $ \tilde{\x}_{i,n}\sim  \mathcal{N}(0, \Sigma)  $, $\w_{i,n} \sim \mathcal{N}(0, \bm{I}_{q_n})$ and $\varepsilon_{i,n}\sim \mathcal{N}(0,1)$.  In the Supplementary Materials, we have further included DGPs with more realistic error terms beyond normal distribution, including error terms with asymmetric and bimodal distributions. For most of the DGPs, we investigate both homoscedastic as well as heteroscedastic models. See Supplementary Materials Section C for detailed description and simulation results. 

As for the coefficients, we consider {three DGPs that vary in $\bbeta$ and $\bm{\gamma}_n$.} The first DGP considers the case in which no policy is effective (meaning that $\bbeta = 0$), and the coefficient $\gamma_{j} = 1/j$, for $j=1, \ldots, q_n$. We refer to this case as the ``homogeneity" case since $\beta_j$'s take the same value zero. The second and the third DGPs consider cases where policy effects are generated from $\beta_{j} = \Phi^{-1}\big( \frac{j}{d+1} \big)$ for $j=1, \ldots, d$, and the coefficients are either $\bm{\gamma}_{n} = 0$ or $\gamma_{j} = 1/j$, for $j=1, \ldots, q_n$. We refer to this case as the {``heterogeneity(1)" case and ``heterogeneity(2)" case}, respectively, since different policies have heterogeneous effects. 

We set the sample size $n \in\{700,2000\}$ to mimic the sample size in our case studies, the number of policies $d \in \{5,10\}$, and the dimension of the covariates $q_n$ from $ q_n\in \{ 1,  141,  281,  421,561, 631 \}$. All statistics reported below are computed based on over 1,000 Monte Carlo replications. To avoid redundancy, we present the results for $n=700$ and $d=5$ in the main manuscript, and rests are provided in the Supplementary Materials (Section C).

To demonstrate the robustness of the adopted covariance matrix estimator, we compare our proposal with three alternative covariance matrix estimators. The first one we compare with is the covariance matrix estimator proposed by \cite{cattaneo2018inference}: $ \hat{\bm{\Omega}}^{\texttt{HCK}}_n= \hat{\bm{\Gamma}}_n^{-1}\hat{\Sigma}^{\texttt{HCK}}_n \hat{\bm{\Gamma}}_n^{-1}$, 
where $ \hat{\Sigma}^{\texttt{HCK}}_n \triangleq \frac{1}{n}\sum_{i=1}^n\sum_{j=1}^n\kappa_{ij,n}^{\texttt{HCK}} \hat{\v}_{i,n}\hat{\v}_{i,n}'  \hat{u}_{j,n}^2$,  $\hat{u}_{j,n} = \sum_{k=1}^n (\M_n)_{j,k}(y_{k,n} -\x_{k,n}'\hat{\bbeta} )$, and 
\begin{align*}
    \kappa_n^{\texttt{HCK}}=
    \begin{pmatrix}
        M_{11,n}^2 & \cdots & M_{1n,n}^2 \\
        \vdots & \ddots & \vdots \\
        M_{n1,n}^2 & \cdots & M_{nn,n}^2  \\
    \end{pmatrix}^{-1}
    & =(\M_n\odot\M_n)^{-1},
\end{align*} 
with $\odot$ denoting the Hadamard product. The estimator $\hat{\Sigma}^{\texttt{HCK}}_n$ is well-defined whenever $(\M_n\odot\M_n)$ is invertible. We use the acronym ``HCK" to denote this estimator in the following parts. The second one we compare with is the classical Eicker-White covariance matrix estimator \citep{eicker1963asymptotic, white1980heteroskedasticity} of the form: 
\begin{align}\label{Eq:eicker-white}
    \hat{\bm{\Omega}}^{\texttt{EW}}_n= \hat{\bm{\Gamma}}_n^{-1}\hat{\Sigma}^{\texttt{EW}}_n \hat{\bm{\Gamma}}_n^{-1},
\end{align}
where $ \hat{\Sigma}^{\texttt{EW}}_n \triangleq \frac{1}{n}\sum_{i=1}^n \hat{\v}_{i,n}\hat{\v}_{i,n}'  \hat{u}_{i,n}^2$ and  $\hat{u}_{i,n} = \sum_{j=1}^n (\M_n)_{i,j}(y_{j,n} -\x_{j,n}'\hat{\bbeta} )$. 
% \begin{align*}
%     \hat{\Sigma}^{\texttt{EW}}_n \triangleq \frac{1}{n}\sum_{i=1}^n \hat{\v}_{i,n}\hat{\v}_{i,n}'  \hat{u}_{i,n}^2, \quad \hat{u}_{i,n} = \sum_{j=1}^n (\M_n)_{i,j}(y_{j,n} -\x_{j,n}'\hat{\bbeta} ). 
% \end{align*}
We use the acronym ``EW" to denote this estimator in our simulation results section.  Huber-Eicker-White standard error is also known as the HC0 standard error, where HC stands for ``heteroskedasticity robust." The last covariance matrix estimator we adopted is a variant of the HC0 estimator: 
\begin{align}\label{Eq:HC3}
    \hat{\bm{\Omega}}^{\texttt{HC3}}_n= \hat{\bm{\Gamma}}_n^{-1}\hat{\Sigma}^{\texttt{HC3}}_n \hat{\bm{\Gamma}}_n^{-1},\quad\text{ where } \hat{\Sigma}^{\texttt{HC3}}_n \triangleq \frac{1}{n}\sum_{i=1}^n \hat{\v}_{i,n}\hat{\v}_{i,n}'  \frac{ \hat{u}_{i,n}^2}{ (\M_n)_{i,j}^2 }.
\end{align}
The above estimator upward reweights regression residuals, and we use the acronym ``HC3" to denote this estimator in our simulation results section.

%For the candidate region of $c_{\texttt{L}}$ and $c_{\texttt{R}}$, we first use the region proposed in Supplementary Material Section 3.3.1 for 30 Monte Carlo samples and get the values of $L_j( c_{\texttt{L}}, c_{\texttt{R}})$. According to the 30 outcomes, we carefully delete the $c_{\texttt{L}}$ and $c_{\texttt{R}}$ that seems always lead to large losses $L_j( c_{\texttt{L}}, c_{\texttt{R}})$ and add more choices around $c_L$, $c_R$ that has smallest $L_j( c_{\texttt{L}}, c_{\texttt{R}})$. Then we fix the candidate region on the whole 1000 Monte Carlo samples to get the estimation.

\subsection{Simulation results}

We summarize our main takeaways from the simulation results presented in Table \ref{table:sim-d5-hetero-1st}-\ref{table:sim-d5-hetero-2nd}, where we have compared our proposed approach (``Proposed + KJ") in Section \ref{Sec:Method} with four other methods. ``Proposed + EW", ``Proposed + HC3", and ``Proposed + HCK" refer to methods adjusting for the winner's curse bias but use $\hat{\bm{\Omega}}_n^{\texttt{EW}}$, $\hat{\bm{\Omega}}_n^{\texttt{HC3}}$,  and $\hat{\bm{\Omega}}_n^{\texttt{HCK}}$, respectively, to estimate the covariance matrix of $\bbeta$.  ``No adjustment+KJ" refers to the approach with no adjustment for the winner's curse bias and adopts the robust covariance matrix estimator proposed by \cite{jochmans2020heteroscedasticity} to make inference on the best policies. 
We present the coverage probabilities and $\sqrt{n}$-scaled biases for the top two policies in the population, i.e., $\beta_{(1)}$ and $\beta_{(2)}$. As the simulation results are rather similar for the observed top two policies in the random sample, i.e., $\beta_{\hat{1}}$, $\beta_{\hat{2}}$, we present these results in the Supplementary Materials (Section C.4). 

Our simulation results confirm our theoretical results presented in Theorem \ref{Theorem:resampling-consistency}. When no policy is effective, our proposed method not only successfully suppresses the winner's curse bias for the top two policies but also attains near nominal coverage (Table \ref{table:sim-d5-homo-beta1}).  Similar pattern can also be observed when top policies are effective (i.e., $\beta_j$'s are heterogeneous, and Table \ref{table:sim-d5-hetero-1st} in particular). 
In nearly all designs and for a range of considered values of $q_n$, our proposal yields close to nominal confidence interval, though some under coverage is observed for large values of $q_n$. The method with no adjustment is obviously biased upward due to the winner's curse phenomenon, thus it provides under-covered confidence intervals and point estimates with rather large biases.
In all considered cases, both the EW-based method {and the HC3-based method} tend to lose coverage when $q_n\geq141$, and the HCK-based method tends to produce under-covered confidence interval whenever $q_n\geq561$. 
In moderately high dimensions so that $q_n/n$ is approximately one half, the proposed method with the HCK variance estimator has comparable performances with our approach. 

%\textcolor{blue}{In the Supplementary Materials, we further investigate the performance of our proposed method for making inference on $\beta_{(5)}$ and  $\beta_{(10)}$ when $d=10$. The results are summarized in Supplementary Materials Table 2--4. The results demonstrate that our proposed method still reaches nominal level coverage and suppresses the winner's curse regardless of the ranking of $\beta_{j}$. }

%our method also makes valid inference on the observed best (or second-best) policy's effect size, $\beta_{\hat{1}}$ ($\beta_{\hat{2}}$). In Table \ref{table:sim-d5-hetero-1st-hat}, while the EW-based method and the HCK-based method lose coverage for $\beta_{\hat{1}}$ quickly as $q_n$ grows, our method retains nominal-level coverage at $q_n = 561$. The robustness of our proposal remains for estimating $\beta_{\hat{2}}$, as seen in Table \ref{table:sim-d5-hetero-2nd-hat}. 

\begin{table}[h!]
%\begin{center}
 %\centering
\caption{ Simulation results ($d=5,\text{heterogeneity}, \beta_{(1)}$)}\label{table:sim-d5-hetero-1st}

%\resizebox{\columnwidth}{!}{%
\begin{adjustbox}{width=.8\textwidth,center}
\begin{tabular}{ccccccc}
     \hline\hline
           & &\multicolumn{4}{c}{  $\beta_{j} = \Phi^{-1}\big( \frac{j}{d+1} \big), \quad \bm{\gamma}_n=0$, \quad j = 1,\ldots, d} \\[0.15cm]
           \cline{3-7}
           &&\multicolumn{4}{c}{ $ \x_{i,n}\sim \mathcal{N}(0, \Sigma), \quad \w_{i,n} = \mathds{1}( \tilde{\w}_{i,n}\geq \Phi^{-1}(0.98) ) $}   \\[0.15cm]
            \cline{3-7} 
             & &  Proposed+KJ  &Proposed+HCK  &   Proposed+HC3 & Proposed+EW & No adjustment+KJ\\[0.15cm]
       \cline{3-7} 
    $q_n=1$ &Cover  & 0.97(0.01) & 0.96(0.01) & 0.96(0.01)  & 0.96(0.01)& 0.97(0.01)
    \\
    
 & $\sqrt{n}$Bias & -0.04(0.06)
     &-0.03(0.04) &-0.03(0.04) & -0.04(0.05) & 0.05(0.06)
    \\[0.15cm]

   $q_n=141$ &Cover & 0.96(0.01) & 0.96(0.01) & 0.94(0.01) & 0.95(0.01) & 0.95(0.01)\\

     & $\sqrt{n}$Bias & -0.04(0.05) & -0.04(0.04) & -0.04(0.04) & 0.06(0.06) & 0.06(0.07)\\[0.15cm]
     
       $q_n=281$ &Cover  &0.96(0.01) & 0.95(0.01) & 0.82(0.02)  &  0.80(0.01)& 0.94(0.01)
    \\  
    
    & $\sqrt{n}$Bias & -0.05(0.06)
     & -0.06(0.05)& -0.06(0.03) & -0.10(0.07) & -0.08(0.08)
    \\[0.15cm]

   $q_n=421$ &Cover & 0.95(0.02) & 0.94(0.01) &0.79(0.01) &  0.76(0.01)& 0.78(0.01)\\

     & $\sqrt{n}$Bias & -0.05(0.05) & -0.06(0.05)&-0.07(0.05) &  -0.12(0.09)& 0.11(0.09)\\[0.15cm]
     
      $q_n=561$ &Cover &0.95(0.01) & 0.92(0.01) & 0.65(0.02)& 0.63(0.01) & 0.68(0.01) \\
    
     & $\sqrt{n}$Bias & -0.07(0.07) & -0.09(0.07) &-0.17(0.10) & -0.20(0.12) & 0.15(0.13)\\[0.15cm]
     
      $q_n=631^*$ &Cover &0.93(0.01) &0.91(0.01) &0.51(0.02) & 0.48(0.01) & 0.55(0.01) \\
    
     & $\sqrt{n}$Bias & -0.17(0.08) & -0.19(0.10) &-0.28(0.11) & -0.35(0.22) & -0.26(0.13)\\[0.15cm]
      \cline{3-7}

     & & \multicolumn{4}{c}{$\x_{i,n} = \mathds{1}(\tilde{\x}_{i,n} > 0), \quad \w_{i,n} \sim \mathcal{N}(0, I) $}    \\[0.15cm]
      \cline{3-7} 
      & &  Proposed+KJ & Proposed+HCK  &  Proposed+HC3  &  Proposed + EW &No adjustment+KJ\\[0.15cm]
       \cline{3-7} 
     $q_n=1$ &Cover  &0.97(0.01) &0.97(0.01) & 0.95(0.01) & 0.96(0.01) & 0.97(0.01)
    \\

       & $\sqrt{n}$Bias & -0.02(0.07) & -0.02(0.04)& -0.01(0.03) & -0.07(0.11) & -0.05(0.09)
       \\[0.15cm]
       
    $q_n=141$ &Cover &0.96(0.01) & 0.95(0.01)& 0.94(0.01)&  0.94(0.01) &0.96(0.01)\\

        & $\sqrt{n}$Bias &-0.02(0.03) & -0.02(0.02)& -0.03(0.02)& 0.11(0.12) & -0.06(0.12)\\[0.15cm]
        
          $q_n=281$ &Cover  & 0.95(0.01) & 0.94(0.01) & 0.87(0.01) & 0.85(0.01)& 0.95(0.01)
    \\

    & $\sqrt{n}$Bias & -0.03(0.04)
     & -0.03(0.03) & -0.04(0.02) & 0.14(0.12) & -0.08(0.13)
    \\[0.15cm]

   $q_n=421$ &Cover & 0.95(0.01) & 0.94(0.01) & 0.78(0.01) & 0.76(0.01) & 0.75(0.01)\\

     & $\sqrt{n}$Bias & -0.03(0.03) &-0.05(0.04) & -0.08(0.05)& -0.19(0.17)& 0.19(0.14)\\[0.15cm]
     
      $q_n=561$ &Cover & 0.95(0.01) &0.92(0.01) &0.63(0.02) & 0.61(0.01) & 0.63(0.01)\\

     & $\sqrt{n}$Bias & -0.04(0.04) & -0.08(0.06)&-0.19(0.10) &  -0.30(0.26)& -0.24(0.22)\\[0.15cm]
     
     $q_n=631$ &Cover &0.94(0.01) & 0.91(0.01) & 0.49(0.02) &  0.45(0.01)& 0.68(0.01)\\
    
     & $\sqrt{n}$Bias &-0.06(0.07) & -0.11(0.09)& -0.23(0.13)& -0.42(0.20) & 0.39(0.29)\\[0.15cm]
      \cline{2-5}

     \hline\hline
     \end{tabular}
     \end{adjustbox}
      \begin{tablenotes}\footnotesize
  \item Note: ``Cover" is the empirical coverage of the 95\% confidence interval for $\beta_{(1)}$ and `` $\sqrt{n}$Bias " captures the root-$n$ scaled Monte Carlo bias for estimating $\beta_{(1)}$. `` * " indicates that $\hat{\bm{\Omega}}_n^{\texttt{KJ}}$ is not positive semi-definite in some Monte Carlo samples.
     \end{tablenotes}
%\end{center}
 \end{table}

  \begin{table}[h!]
%\begin{center}
 %\centering
\caption{ Simulation results ($d=5, \text{homogeneity},  \beta_{(1)}$)}\label{table:sim-d5-homo-beta1}

%\resizebox{\columnwidth}{!}{%
\begin{adjustbox}{width=.8\textwidth,center}
\begin{tabular}{ccccccc}
     \hline\hline
           & &\multicolumn{4}{c}{  $\bm{\beta}=0, \quad \bm{\gamma}_j=1/j$} \\[0.15cm]
           \cline{3-7}
           &&\multicolumn{4}{c}{ $ \x_{i,n}\sim \mathcal{N}(0, \Sigma), \quad \w_{i,n} = \mathds{1}( \tilde{\w}_{i,n}\geq \Phi^{-1}(0.98) ) $}   \\[0.15cm]
            \cline{3-7} 
             & &  Proposed+KJ & Proposed+HCK  & Proposed+HC3 & Proposed+EW  & No adjustment+KJ\\[0.15cm]
       \cline{3-7} 
    $q_n=1$ &Cover  & 0.97(0.01) & 0.96(0.01) & 0.96(0.01)  & 0.93(0.02) &0.90(0.01)
    \\
    
 & $\sqrt{n}$Bias & 0.02(0.03)
     & 0.02(0.03) &0.03(0.03) & 0.03(0.03) &1.64(0.04)
    \\[0.15cm]

   $q_n=141$ &Cover & 0.96(0.01) & 0.96(0.01) & 0.96(0.01) & 0.89(0.02) &0.88(0.01)\\

     & $\sqrt{n}$Bias & 0.03(0.04) & 0.04(0.04) & 0.03(0.04) & 0.12(0.04) & 1.78(0.04)\\[0.15cm]
     
       $q_n=281$ &Cover  &0.96(0.01) & 0.94(0.01) & 0.90(0.02)  & 0.85(0.01) &0.83(0.01)
    \\
    
    & $\sqrt{n}$Bias & 0.03(0.04)
     & 0.04(0.04)& 0.05(0.04) &  0.22(0.03) &2.03(0.05)
    \\[0.15cm]

   $q_n=421$ &Cover & 0.95(0.01) & 0.93(0.01) &0.82(0.02) & 0.79(0.01) &0.74(0.02) \\

     & $\sqrt{n}$Bias & 0.05(0.05) & 0.18(0.05)&0.24(0.06) & 0.36(0.03)  &2.63(0.06) \\[0.15cm]
     
      $q_n=561$ &Cover &0.95(0.01) & 0.93(0.01) & 0.67(0.02)& 0.73(0.01) &0.63(0.02)\\
    
     & $\sqrt{n}$Bias & 0.08(0.09) & 0.51(0.05) &0.74(0.06) & 0.44(0.04) &3.74(0.09)\\[0.15cm]
     
      $q_n=631^*$ &Cover &0.93(0.01) &0.89(0.01) &0.53(0.02) & 0.50(0.01) &0.45(0.02)\\
    
     & $\sqrt{n}$Bias & 0.18(0.09) & 1.21(0.09) &1.84(0.11) & 2.42(0.06) &5.10(0.12)\\[0.15cm]
      \cline{3-7}

     & & \multicolumn{4}{c}{$\x_{i,n} = \mathds{1}(\tilde{\x}_{i,n} > 0), \quad \w_{i,n} \sim \mathcal{N}(0, I) $}    \\[0.15cm]
      \cline{3-7} 
      & &  Proposed+KJ &Proposed+HCK  & Proposed+HC3 & Proposed+EW & No adjustment+KJ\\[0.15cm]
       \cline{3-7} 
     $q_n=1$ &Cover  &0.96(0.01) &0.96(0.01) & 0.96(0.01) & 0.94(0.01) &0.90(0.01)
    \\

       & $\sqrt{n}$Bias & 0.03(0.04) & 0.04(0.05)& 0.05(0.05) & 0.07(0.07) &2.75(0.06)
       \\[0.15cm]
       
    $q_n=141$ &Cover &0.96(0.01) & 0.96(0.01)& 0.93(0.01)& 0.90(0.01) &0.83(0.01) \\

        & $\sqrt{n}$Bias &0.05(0.05) & 0.05(0.06)& 0.17(0.07)& 0.31(0.07) &3.29(0.08) \\[0.15cm]
        
          $q_n=281$ &Cover  & 0.95(0.01) & 0.95(0.01) & 0.90(0.01) & 0.88(0.01) &0.75(0.02)
    \\

    & $\sqrt{n}$Bias & 0.07(0.08) & 0.07(0.07) &0.31(0.08) & 0.54(0.05) &3.59(0.08)
    \\[0.15cm]

   $q_n=421$ &Cover &0.95(0.01) & 0.95(0.01) & 0.85(0.01) & 0.86(0.01) &0.65(0.02)\\

     & $\sqrt{n}$Bias &0.04(0.04) &0.10(0.11) & 0.73(0.13)& 0.64(0.06) &4.58(0.11)\\[0.15cm]
     
     $q_n=561$ &Cover & 0.93(0.01) &0.90(0.02) &0.59(0.02) & 0.61(0.01) &0.53(0.02)\\

     & $\sqrt{n}$Bias &0.13(0.07) & 0.19(0.13)&2.00(0.13) & 1.73(0.08) &5.90(0.13)\\[0.15cm]
     
     $q_n=631$ &Cover &0.90(0.01) & 0.78(0.02) & 0.38(0.02) & 0.30(0.01) &0.33(0.02)\\
    
     & $\sqrt{n}$Bias &0.50(0.12) & 2.47(0.16)& 5.16(0.19)& 7.68(0.18) &6.51(0.19)\\[0.15cm]
      \cline{2-5}

     \hline\hline
     \end{tabular}
     \end{adjustbox}
      \begin{tablenotes}\footnotesize
  \item Note: ``Cover" is the empirical coverage of the 95\% confidence interval for $\beta_{(1)}$ and `` $\sqrt{n}$Bias " captures the root-$n$ scaled Monte Carlo bias for estimating $\beta_{(1)}$. `` * " indicates that $\hat{\bm{\Omega}}_n^{\texttt{KJ}}$ is not positive semi-definite in some Monte Carlo samples.
     \end{tablenotes}
%\end{center}
 \end{table}

 %%%%% beta(2) 

 \begin{table}[h!]
%\begin{center}
 %\centering
\caption{ Simulation results ($d=5,\text{heterogeneity}, \beta_{(2)}$)}\label{table:sim-d5-hetero-2nd}

%\resizebox{\columnwidth}{!}{%
\begin{adjustbox}{width=.8\textwidth,center}
\begin{tabular}{ccccccc}
     \hline\hline
           & &\multicolumn{4}{c}{  $\beta_{j} = \Phi^{-1}\big( \frac{j}{d+1} \big), \quad \bm{\gamma}_n=0$} \\[0.15cm]
           \cline{3-7}
           &&\multicolumn{4}{c}{ $ \x_{i,n}\sim \mathcal{N}(0, \Sigma), \quad \w_{i,n} = \mathds{1}( \tilde{\w}_{i,n}\geq \Phi^{-1}(0.98) ) $}   \\[0.15cm]
            \cline{3-7} 
             & &  Proposed+KJ & Proposed+HCK & Proposed+HC3 & Proposed+EW & No adjustment+KJ \\[0.15cm]
       \cline{3-7} 
       
  $q_n=1$ &Cover  & 0.97(0.01) & 0.96(0.01) & 0.96(0.01)  & 0.96(0.01)& 0.96(0.01)
    \\
    
 & $\sqrt{n}$Bias & 0.02(0.06) &0.02(0.06) & 0.01(0.06) & -0.05(0.07) &-0.04(0.07)
    \\[0.15cm]

   $q_n=141$ &Cover & 0.97(0.01) & 0.96(0.01) & 0.93(0.01) & 0.94(0.01)& 0.94(0.01)\\

     & $\sqrt{n}$Bias & -0.02(0.04) & -0.02(0.04) & -0.04(0.03) & -0.06(0.06) &-0.07(0.08) \\[0.15cm]
     
       $q_n=281$ &Cover  &0.95(0.01) & 0.94(0.01) & 0.88(0.02)  & 0.89(0.01)& 0.85(0.01)
    \\
    
    & $\sqrt{n}$Bias & -0.03(0.04)& -0.03(0.03)& -0.07(0.03) & -0.11(0.10) & 0.19(0.11)
    \\[0.15cm]

   $q_n=421$ &Cover & 0.95(0.01) & 0.94(0.02) &0.81(0.02) &  0.80(0.01) &0.77(0.01) \\

     & $\sqrt{n}$Bias & -0.03(0.03) & -0.03(0.04)&-0.10(0.06) &  -0.15(0.12)& -0.23(0.15) \\[0.15cm]
     
      $q_n=561$ &Cover &0.95(0.01) & 0.94(0.02) & 0.67(0.02)& 0.65(0.01)& 0.63(0.01)\\
    
     & $\sqrt{n}$Bias & 0.03(0.03) & 0.05(0.06) &0.12(0.08) & -0.17(0.13) &-0.26(0.18) \\[0.15cm]
     
      $q_n=631^*$ &Cover &0.94(0.01) &0.93(0.01) &0.56(0.02) & 0.53(0.01) & 0.50(0.01) \\
    
     & $\sqrt{n}$Bias & -0.07(0.07) & -0.18(0.06) &-0.23(0.08) & -0.26(0.17) & 0.47(0.22) \\[0.15cm]
      \cline{3-7}

     & & \multicolumn{4}{c}{$\x_{i,n} = \mathds{1}(\tilde{\x}_{i,n} > 0), \quad \w_{i,n} \sim \mathcal{N}(0, I) $}    \\[0.15cm]
      \cline{3-7} 
      & &  Proposed+KJ & Proposed+HCK  &  Proposed+HC3 &Proposed+EW &No adjustment+KJ \\[0.15cm]
       \cline{3-7} 
      $q_n=1$ &Cover  & 0.98(0.01) &0.98(0.01) & 0.98(0.01) & 0.96(0.01)& 0.98(0.01)
    \\

       & $\sqrt{n}$Bias & -0.08(0.12) & -0.09(0.12)& -0.10(0.13) & 0.09(0.12) & -0.07(0.10)
       \\[0.15cm]
       
    $q_n=141$ &Cover &0.97(0.01) & 0.97(0.01)& 0.95(0.01)& 0.95(0.01)& 0.97(0.01) \\

        & $\sqrt{n}$Bias &-0.09(0.12) & -0.10(0.13)& -0.12(0.13)& 0.09(0.13)  & -0.08(0.10)\\[0.15cm]
        
          $q_n=281$ &Cover  & 0.97(0.01) & 0.97(0.01) & 0.90(0.01) & 0.87(0.01) & 0.96(0.01)
    \\

    & $\sqrt{n}$Bias & -0.11(0.14) & -0.10(0.14) &-0.15(0.14) & -0.18(0.14) & -0.10(0.11)
    \\[0.15cm]

   $q_n=421$ &Cover &0.96(0.01) & 0.95(0.02) & 0.80(0.02) &  0.75(0.01) &0.94(0.01) \\

     & $\sqrt{n}$Bias &0.14(0.16) &-0.16(0.18) & -0.20(0.18)& -0.22(0.17) & -0.15(0.15) \\[0.15cm]
     
      $q_n=561$ &Cover & 0.96(0.01) &0.94(0.02) &0.63(0.02) & 0.60(0.01)& 0.92(0.01) \\

     & $\sqrt{n}$Bias &0.14(0.18) & 0.20(0.23)&-0.24(0.22) &  -0.30(0.23)& 0.19(0.20)\\[0.15cm]
     
     $q_n=631$ &Cover &0.94(0.01) & 0.93(0.02) & 0.58(0.02) & 0.55(0.01)& 0.52(0.01) \\
    
     & $\sqrt{n}$Bias&0.15(0.20)& 0.24(0.26)&0.28(0.25)& -0.35(0.13) & 0.56(0.24)\\[0.15cm]

      \cline{2-5}

     \hline\hline
     \end{tabular}
     \end{adjustbox}
      \begin{tablenotes}\footnotesize
  \item Note: ``Cover" is the empirical coverage of the 95\% confidence interval for $\beta_{(2)}$ and `` $\sqrt{n}$Bias " captures the root-$n$ scaled Monte Carlo bias for estimating $\beta_{(2)}$. `` * " indicates that $\hat{\bm{\Omega}}_n^{\texttt{KJ}}$ is not positive semi-definite in some Monte Carlo samples.
     \end{tablenotes}
%\end{center}
 \end{table}

\section{Case studies}\label{Sec:real-data}

\subsection{Case study I: Charitable giving}

In the past half century, charitable giving by individuals in the United States has grown and it has contributed to more than two percent of the annual GDP since 1998 \citep{list2011market}.
Charitable giving is often driven by altruism, while as suggested by many field experiments, improper policies adopted by the demand side--fundraisers--may impair the supply side's (individual donors) motivation of giving \citep{andreoni2002giving}. Therefore, to effectively attract resources from individual donors, fundraisers need to properly design donation incentives. One of the donation incentives is matching grant which means that a matching donor pledges to match any donation from other donors with certain ratio and up to some threshold. As the price elasticity of matching donation may differ from other donation incentives, we hope to carefully investigate different pricing policies in a matching donation and study if the observed top two performing policies are indeed effective. 

We work with the charitable giving data in \cite{karlan2007does}. \cite{karlan2007does} conduct a field experiment that explores the price elasticity in a matching donation. The field experiment involves $50,083$ previous donors to a political charity. Individuals are randomly assigned to two groups: treatment ($n=33,396$) and control ($n=16,687$). In the control group, individuals receive a standard letter with no matching details. In the treatment group, each potential donor receives a letter with three strategies: (1) match ratio, (2) match size, and (3) ask amount. Within each strategy, individuals are randomly assigned to a sub-policy detailed below. 

For the match ratio strategy, there are three sub-policies: (1) 1:1 (the matching donor contributes the same amount as the individual donor), (2) 2:1 (the matching donor contributes twice as many as the individual donor), (3) 3:1 (the matching donor contributes three times as many as the individual donor). For the match size strategy, there are four sub-policies with different pledge amounts: (1) \$25,000, (2) \$50,000, (3) \$100,000, and (4) unstated amount. For the ask amount strategy, individual donors are asked to give same amount, 25\% more or 50\% more than their largest past donation. 

In our study, we focus on the treatment ``ask amount" with three pricing policies, and we study the subpopulation ($n=7,938$) of unmarried males living in red counties or red states. Red county (state) refers to a county (state) in which residents predominantly vote for the Republican Party. The outcome of interest is the donation amount. We have adjusted $q_n = 1,049$ covariates including the donors' demographic information (26 variables), census information (27 variables), and their two-way interaction terms. Our results are summarized in Table \ref{table:ask-amount}.

\begin{table}[h!]
    \centering
 \begin{tabular}{ccccc}
    \hline
     \hline\\[-2ex] 
       Method & Policies(Ask amount)  & Est (95\% CI)  & $p$-value  &   \\
  \\[-2ex] 
    \hline
      \\[-2ex] 
 Uncalibrated & Same & 0.67 (0.09, 1.25) & 0.023*  &  \\
    \\[-2ex] 
       %\hline
  & 25\% more & 0.66 (0.01, 1.31)  & 0.046*  &\\
  \\[-2ex] 
  & 50\% more & 0.33 (-0.21, 0.86)  & 0.235 &  \\
  \\[-2ex]
  \hline
  \\[-2ex]
Calibrated  &  Same & 0.63 (0.10, 1.20) & 0.025*   &  \\
    \\[-2ex] 
       %\hline
  & 25\% more & 0.56(-0.01, 1.07)  & 0.052 &  \\
  \\[-2ex] 
     \hline\hline
    \end{tabular}
    \caption{\small Estimated treatment effects (Est), 95\% confidence intervals (95\% CI), and two-sided $p$-values for the three  ``ask amount" policies. ``Uncalibrated" refers to the study results obtained without any adjustment, and the confidence intervals are constructed based on normal approximation with the estimated covariance matrix $ \hat{\bm{\Omega}}^{\texttt{KJ}}_n$.   ``Calibrated" refers to our proposed methodology. The computational time is $741$ seconds on a Lenovo NeXtScale nx360m5 node (24 cores per node) equipped with Intel Xeon Haswell processor.   \label{table:ask-amount}}
\end{table}

Results in Table \ref{table:ask-amount} suggest that, without any calibration, asking the donor either to give the same amount or to give 25\% more than their highest past donation seems to be the best policies that significantly increase the donation amount. Specifically, our results from running a simple linear regression model suggest that asking the individual donor to give the same amount of their largest past donation appears to be the most effective pricing policy, and it on average raises \$0.67 (95\% CI = (0.09, 1.25), $p$-value = 0.023) per donor. Asking the individual donor to give 25\% more than their largest past donation is the second most effective policy, with an increased donation by \$0.66 (95\% CI = (0.01, 1.31), $p$-value = 0.046) per donor. 

Because we pick the most effective policies from a random sample, these estimates are potentially subject to the winner's curse bias. We thus apply the proposed method to carefully examine these seemly effective policies. After calibrating for the winner's curse bias, we confirm that the asking for the same amount policy remains as the most effective policy, though with a slightly smaller estimated effect size (Est = \$0.63, 95\% CI = (0.10, 1.20), $p$-value = 0.025). This result is moderately aligned with the analysis in \cite{karlan2007does}, whose results suggest that donors from red states or red counties are more willing to contribute, partially because the collaborating charity is politically oriented. However, for the effect of the second best policy--asking to donate 25\% more than past donation--is shifted downward, and it no longer has significant impact in promoting the donation amount (Est = \$0.56, 95\% CI = (-0.01, 1.07), $p$-value = 0.052). This result might be partially explained by the observation that donors are more motivated by a lower ``price" of donation \citep{warwick2003testing}. In sum, our analyses suggest that the best pricing policy of charitable giving for unmarried males living in the Republican Party dominated voting regions \textcolor{black}{could be} asking for the same amount as their highest previous donation, and asking for more donations may not incentivize the donors to give. Though given the obtained p-value before and after calibration for the second best policy is rather close to the 5 percent threshold, we note that such a conclusion should also be viewed with caution. 

%\textcolor{blue}{We further investigate the effect size of each policy under a much smaller model (with main effects only). The results are summarized in Supplementary Materials Table 12. The additional analysis demonstrates that adopting a smaller model does not change the conclusions substantively. }

\subsection{Case study II: National supported work (NSW) program}

In this case study, we revisit a dataset from the National Supported Work (NSW) program. The NSW program is a labor training program implemented in 1970's that provides work experience to disadvantaged workers. Our proposed method can also be used to evaluate if the job training program is indeed beneficial for certain groups of workers. To do so, the structural component $\x_{i,n}$ in the model \eqref{eq:linear-model} would include variables representing the interactions between the treatment variable (the job training program) and different subgroup indicator variables of interest. 

We use the field experiment dataset adopted in \cite{dehejia2002propensity} ($n=455$), in which 185 workers are in the treatment group and 260 workers are in the control group. This dataset consists of a treatment indicator variable, an outcome variable measured by the participant post-treatment earnings in 1978, and eight baseline variables ( including age, years of education, an indicator for high school degree, indicators for Black and Hispanic, marital status, and pre-treatment earnings in 1974 and 1975). We further add three sets of additional covariates following the setup in \cite{farrell2015robust}: (1) $\mathds{1}$(1974 earnings =0) and $\mathds{1}$(1975 earnings =0); (2) all first-order interactions; (3) all polynomials up to the 2nd-order. The final dataset includes $51$ covariates. We aim to investigate the effectiveness of the NSW program in four groups of workers:  (1) married Black workers, (2) unmarried Black workers, (3) married Non-Black workers, and (4) unmarried Non-Black workers. The summarized results are shown in Table \ref{table:nsw}.

\begin{table}[h!]
    \centering
 \begin{tabular}{ccccc}
    \hline
     \hline\\[-2ex] 
       Method & Subgroups  & Est (95\% CI) (\$$10^3$)  & $p$-value  &   \\
  \\[-2ex] 
    \hline
      \\[-2ex] 
 Uncalibrated & Black, married & 4.35 (0.89, 7.81) & 0.014* &  \\
    \\[-2ex] 
       %\hline
  & Black, unmarried & 1.10(-0.55, 2.75) & 0.190  & \\
  \\[-2ex] 
  & Non-Black, married &  1.33(-6.63, 9.29) & 0.743  &  \\
  \\[-2ex]
   & Non-Black, unmarried & 1.40(-2.61, 5.40)  &  0.494 &   \\
  \\[-2ex]
  \hline
  \\[-2ex]
Calibrated  & Black, married & 4.41(1.74, 8.50)  & 0.009* &  \\
    \\[-2ex] 
     \hline\hline
    \end{tabular}
    \caption{\small Estimated treatment effects (Est), 95\% confidence intervals (CI), in units \$$10^3$/year, and two-sided $p$-values for the four  subgroups in the NSW study ($n=445$, $q_n=51$). ``Uncalibrated" refers to the study results obtained without any adjustment, and the confidence intervals are constructed based on normal approximation with the estimated covariance matrix $ \hat{\bm{\Omega}}^{\texttt{KJ}}_n$.   ``Calibrated" refers to our proposed methodology. The computational time is $122$ seconds on a Lenovo NeXtScale nx360m5 node (24 cores per node) equipped with Intel Xeon Haswell processor. \label{table:nsw}}
\end{table}

Table \ref{table:nsw} demonstrates that without adjusting for the winner's curse bias, married Black workers (estimated treatment effect =  4.35, 95\% CI = (0.89, 7.81), $p$-value = 0.014, in units $\$10^3$) seem to benefit from the program the most. After accounting for the winner's curse bias issue, our approach \textcolor{black}{potentially} confirms that the treatment effect of the NSW program for the married Black workers is still significant, and the calibrated treatment effect remains roughly the same ( Est= 4.41, 95\% CI = (1.74, 8.50), $p$-value = 0.009, in units $\$10^3$/year).

The dataset collected from the NSW program has been frequently analyzed in the past decade, and our results are largely in-line with current understandings gathered in past studies. For example, although not focusing on the same groups of workers,  \cite{imai2013estimating} suggest that married and unemployed Black workers with some college education have increased their post-treatment earnings for about 38\%. \cite{dehejia2002propensity} show that the job training program yields positive treatment effect on the overall Black participants. In this case study, our approach {may} help to confirm the seemly effective subgroup observed in a random sample while providing a statistically justified estimate accounting for the winner's curse bias.

\section{Concluding remarks}\label{Sec:discussion}

In this article, we have introduced an approach to evaluate multiple best policies based on resampling in the context of a linear model with many covariates. While our approach is numerically reliable and theoretically grounded, it is worthwhile to generalize our framework so that the policy effects can be estimated with other off-shelf methods that are, for example, robust to the high-dimensional confounders or to the presence of interference and noncompliance. Our current theoretical analysis suggests that our proposed approach can be readily extended as long as the covariance matrix between different policies can be consistently estimated. It is thus desirable for us to provide a general framework to broaden future applications for other disciplines in general. 

\bigskip

\noindent\textbf{Acknowledgement.} The research of Jingshen Wang is supported in part by the National Science Foundation (DMS 2015325) and the National Institute of Health (R01MH125746).

\clearpage

\bibliographystyle{jasa} 
\bibliography{reference}

\clearpage

%%%%%%%%%%% appendix command %%%%%%%%%%
\appendix
\addcontentsline{toc}{section}{Appendix} 
\part{Appendix} 
\parttoc 

\setcounter{page}{1}
%%%%%%%%%%%%%%%%%%%%%%%%%%%%%%%%%%%%%%%

\setcounter{table}{0}
\renewcommand\thetable{\Alph{section}.\arabic{table}}

\numberwithin{equation}{section}

\section{Theorem 1}

\subsection{Review of notations and assumptions}

We denote the sample $\{ (y_{i,n}, \x_{i,n}', \w_{i,n}')'\}_{i=1}^n$ as $\{ \z_{i,n}\}_{i=1}^n$. Recall $u_{i,n}$ is the random error in the considered linear model:
\begin{align}\label{supp:eq:linear-model}
y_{i,n} = \x_{i,n}'\bbeta +\w_{i,n}'\bgamma_n + u_{i,n}, \quad i=1, \ldots, n,
\end{align}
we define 
\begin{align}\label{supp:eq:varepsion-v-definition}
\varepsilon_{i,n} = u_{i,n} - \mathbb{E}[ u_{i,n} |  \{ \w_{i,n} \}_{i=1}^n, \{ \x_{i,n} \}_{i=1}^n],\quad \v_{i,n} = \x_{i,n} - \mathbb{E}[\x_{i,n} |  \{ \w_{i,n} \}_{i=1}^n], \quad i=1, \ldots, n.  
\end{align}
Let $e_{i,n} = \mathbb{E}[ u_{i,n} |  \{ \w_{i,n} \}_{i=1}^n, \{ \x_{i,n} \}_{i=1}^n]$, we further denote 
\begin{align*}
& \hat{u}_i = \sum_{j=1}^n (\M_n)_{i,j}(y_{j,n} -\x_{j,n}'\hat{\bbeta} ),\quad \hat{\v}_{i,n} = \sum_{i=1}^n (\M_{n})_{i,j} \x_{j,n}, \\
&  (\M_{n})_{i,j} = \mathds{1}( i=j ) -  \w'_{i,n}\Big( \sum_{k=1}^n \w_{k,n}\w_{k,n}' \Big)^{-1} \w_{j,n},\\
&  \sigma_{i,n}^2 = \mathbb{E}[\varepsilon_{i,n}^2 |  \{ \w_{i,n} \}_{i=1}^n, \{ \x_{i,n} \}_{i=1}^n ], \quad  \tilde{\v}_{i,n} = \sum_{j=1}^n (\bm{M}_n)_{i,j}\v_{j,n}, \\ 
&  \rho_{n}^{1} = \frac{ \sum_{i=1}^n\mathbb{E}\big[ e_{i,n}^2  \big]  }{n},\quad  \rho_{n}^{2} = \frac{ \sum_{i=1}^n\mathbb{E}\Big[ \mathbb{E}\big(e_{i,n} |  \{ \w_{i,n} \}_{i=1}^n \big)^2  \Big]  }{n}, \\
& \bm{Q}_{i,n} = \mathbb{E}\Big[ \x_{i,n} - \big(\sum_{j=1}^n \mathbb{E}[ \w_{j,n}\w_{j,n}' ]  \big)^{-1} \sum_{j=1}^n \mathbb{E}[ \w_{j,n}\x'_{j,n} ]   \Big|  \{ \w_{i,n} \}_{i=1}^n  \Big].
\end{align*}
We will use the notation $\mathbb{P}(\cdot | \{ \z_{i,n}\}_{i=1}^n)$ to refer to the probability that is conditional on the random variables $\{ \z_{i,n}\}_{i=1}^n$. 

For a policy $j$, recall our definition of the near tie set: 
\begin{align*}
\mathcal{H}_{(j)} = \big\{ k: | \beta_k  - \beta_{(j)} | = o(n^{-\frac{1}{2}}),\ k=1, \ldots, d \big\}.
\end{align*}
This suggests that $\forall k \in \mathcal{H}_{(j)} $, there exist a sequence $\delta_{n} \rightarrow 0$ as $n\rightarrow 0$, such that 
\begin{align*}
\beta_k = \beta_{(j)} + n^{-\frac{1}{2}} \cdot \delta_{n}, \quad \forall k \in \mathcal{H}_{(j)} .
\end{align*}
Next, let $\hat{\bm{e}}_j $ denote a $d$-dimensional (sparse) vector based on the estimated tie set $\widehat{\mathcal{H}}_{(j)} $ with
\begin{align*}
\hat{\bm{e}}_j = (\hat{e}_{j,1}, \ldots, \hat{e}_{j, d}),\quad  \hat{e}_{jk} = \frac{\mathds{1}( k \in  \widehat{\mathcal{H}}_{(j)} )}{ |\hat{\mathcal{H}}_{(j)} | }, \quad k=1, \ldots, d,
\end{align*}
and define a $d-$dimensional sparse index vector  based on the true near-tie set ${\mathcal{H}}_{(j)} $ as 
	\begin{align}\label{eq:e-j-definition}
    {\bm{e}}_j = ({e}_{j,1}, \ldots, {e}_{j, d}),\quad  {e}_{jk} = \frac{\mathds{1}( k \in  {\mathcal{H}}_{(j)} )}{ |{\mathcal{H}}_{(j)} | }, \quad k=1, \ldots, d.
	\end{align}
	
We make following assumptions throughout this section: 

\begin{assumption}[Sampling] \label{supp:Assumption:Sampling}
The errors $\varepsilon_{i,n}$ are uncorrelated across $i$ conditional on $\{ \x_{i,n} \}_{i=1}^n$ and $\{ \w_{i,n} \}_{i=1}^n$. Let $\{ N_1, \ldots, N_{G_n} \}$ represents a partition of $\{1,\ldots, n\}$ with $\underset{ 1\leq g\leq G_n}{\max}{ |N_g| } = O(1)$ such that $ \{ (\varepsilon_{i,n}, \v_{i,n} ), i\in N_g \} $ (defined in \eqref{supp:eq:varepsion-v-definition}) are independent across $g$ conditional on $\{ \w_{i,n} \}_{i=1}^n$. 
\end{assumption}

\begin{assumption}[Design] \label{supp:Assumption:Design}
The dimension of the covariates $\w_{i,n}$ satisfies that $\limsup_{n\rightarrow \infty} q_n/n < 1$. The minimum eigenvalue of of the matrix $\sum_{i=1}^n \w_{i,n}\w_{i,n}'$ is bounded away from 0 with probability approaching one, that is  
\begin{align*}
 \lim_{n\rightarrow \infty}\mathbb{P}\Big( \lambda_{\min} \big( \sum_{i=1}^n \w_{i,n}\w_{i,n}'\big) >0 \big)\Big) = 1. 
\end{align*}
Lastly, 
\begin{align*}
& \underset{1\leq i\leq n}{\max}\Big\{  \mathbb{E}[\varepsilon_{i,n}^4 |  \{ \w_{i,n} \}_{i=1}^n, \{ \x_{i,n} \}_{i=1}^n ],\ \frac{1}{\sigma_{i,n}^2},\\
& \quad \quad\quad \quad  \mathbb{E}[\v_{i,n}^4 | \{ \w_{i,n} \}_{i=1}^n ], \ 1/\lambda_{\min} \big( \frac{\sum_{i=1}^n \mathbb{E}\big[  \tilde{\v}_{i,n}\tilde{\v}_{i,n}'|  \{ \w_{i,n} \}_{i=1}^n \big] }{n} \big) \Big\} = O_p(1).
\end{align*}
\end{assumption}

\begin{assumption}[Linear model approximation] \label{supp:Assumption:Approximation}
 $ \sum_{i=1}^n \mathbb{E}[|| Q_{i,n}||^2 ] /n = O(1)$, $\rho_{n}^{1} + n( \rho_{n}^{1} - \rho_{n}^{2} ) +  \rho_{n}^{2} \cdot \sum_{i=1}^n \mathbb{E}[|| Q_{i,n}||^2 ]  = o(1)$, and $\underset{1\leq i\leq n}{\max} || \hat{\v}_{i,n}||/\sqrt{n} = o_p(1)$, $n\rho_n^{1} = O(1)$. 
\end{assumption}

\begin{assumption}[Variance estimation] $\lim_{n\rightarrow \infty}\mathbb{P}( \min_i  \big( \bm{M}_n \big)_{i,i} >0  ) = 1$, 
\begin{align*}
\mathbb{P}\Big( \min_i \big( \bm{M}_n \big)_{i,i}>0 \Big) = O_p(1), \quad \frac{ \sum_{i=1}^n || \tilde{ \bm{Q} }_{i,n}||^4 }{n} = O_p(1),
\end{align*}
and $\max_i ||\mu_{i,n}||/\sqrt{n} = o_p(1)$ with $\mu_{i,n} = \mathbb{E}\big[ y_{i,n} | \{\mathbf{x}_{i,n} \}_{i=1}^n, \{\mathbf{w}_{i,n}\}_{i=1}^n \big]$.
\label{supp:assumption:variance_estimation}
\end{assumption}

\begin{assumption}[Policy effect sizes]\label{supp:Assumption:seperation}
For $\delta \in (0,\frac{1}{2})$, the asymptotic distance between the effects of policy $k$ and $j$ diverges as $n\rightarrow \infty$: 
\begin{align*}
n^{\delta} \cdot \min_{k\not\in \mathcal{H}_{(j)}}\big| \beta_{(j)} - \beta_k \big| \rightarrow \infty, \text{ as }n\rightarrow \infty, \quad j=1, \ldots, d. 
\end{align*}
\end{assumption}

\subsection{Proof of Theorem \ref{Theorem:resampling-consistency}}

In this section, we show the following theorem holds: 
\begin{thm}\label{supp:Theorem:resampling-consistency}
Under Assumptions \ref{supp:Assumption:Sampling}-\ref{supp:Assumption:seperation}, for any $t\in \mathbb{R}$, for the resampled statistics, the following holds precisely
\begin{align*}
\mathbb{P}\left(  \frac{ \sqrt{n}\big(\tilde{\beta}_{(j)}^* - \tilde{\beta}_{(j)}\big) }{ (\hat{\bm{e}}_j'\hat{\bm{\Omega}}^{\texttt{KJ}}_n \hat{\bm{e}}_j  )^{\frac{1}{2}} } \leq t\Big | \{ (y_{i,n}, \x_{i,n}', \w_{i,n}')' \}_{i=1}^n \right) = \Phi(t).
\end{align*}
For the original statistics, it holds that 
\begin{align*}
 \lim_{n\rightarrow \infty} \mathbb{P}\left(  \frac{ \sqrt{n}\big(\tilde{\beta}_{(j)} - {\beta}_{(j)}\big) }{ (\hat{\bm{e}}_j'\hat{\bm{\Omega}}^{\texttt{KJ}}_n \hat{\bm{e}}_j  )^{\frac{1}{2}} } \leq t\right) = \Phi(t). 
\end{align*}
In addition, we show that 
	\begin{align*}
\lim_{n\rightarrow \infty }\mathbb{P}\Big( 	\mathbb{P}\big( 	\tilde{\beta}_{(j)}^* \leq  \beta_{(j)}   |\{ \z_{i,n}\}_{i=1}^n   \big) \leq s \Big)  = s.
	\end{align*}
\end{thm}

\begin{proof}

Our proof of Theorem 1 entails the following steps: 
\begin{description}
\item[Step 1.] Under Assumptions \ref{supp:Assumption:Sampling}-\ref{supp:Assumption:Approximation}, \cite{jochmans2020heteroscedasticity} has shown the following holds 
\begin{align*}
(\hat{\bm{\Omega}}^{\texttt{KJ}}_n)^{-\frac{1}{2}}\sqrt{n}(\hat{ \bbeta} - \bbeta ) \leadsto N(0, \bm{I}_d),
\end{align*}
where $ \bm{I}_d$ is a $d$-dimensional identity matrix. 
Therefore, following the definition of $\bm{e}_j$ in Eq \eqref{eq:e-j-definition}, we have 
	\begin{align}
({\bm{e}}_j'\hat{\bm{\Omega}}^{\texttt{KJ}}_n {\bm{e}}_j  )^{-\frac{1}{2}}\cdot \sqrt{n}\big(  \frac{ \sum_{k\in  {\mathcal{H}}_{(j)}} \hat{\beta}_k }{ | {\mathcal{H}}_{(j)} |} - \frac{ \sum_{k\in  {\mathcal{H}}_{(j)}} {\beta}_k }{ | {\mathcal{H}}_{(j)} |}  \big)  \leadsto \mathcal{N}(0,1). \label{eq:proof-of-theorem-1-step-1}
	\end{align}
	
	\item[Step 2.] Because of Lemma \ref{Lemma:tie-set-consistency}, we have
	\[
	\lim_{n\rightarrow\infty} \mathbb{P}(\hat{\mathcal{H}}_{(j)} \neq \mathcal{H}_{(j)}) = 0.
	\]
	Combing this with \eqref{eq:proof-of-theorem-1-step-1}, we have 
		\begin{align*}
	    \Phi(t) &=  \lim_{n\rightarrow\infty}  \mathbb{P}\Big(({\bm{e}}_j'\hat{\Sigma}^{\texttt{KJ}}_n {\bm{e}}_j  )^{-\frac{1}{2}}\cdot \sqrt{n}\big(  \frac{ \sum_{k\in  {\mathcal{H}}_{(j)}} \hat{\beta}_k }{ | {\mathcal{H}}_{(j)} |} - \frac{ \sum_{k\in  {\mathcal{H}}_{(j)}} {\beta}_k }{ | {\mathcal{H}}_{(j)} |}  \big) \le t \Big) \\
	    &=  \lim_{n\rightarrow\infty} \Bigg[ 
	    \mathbb{P}\Big(({\bm{e}}_j'\hat{\Sigma}^{\texttt{KJ}}_n {\bm{e}}_j  )^{-\frac{1}{2}}\cdot \sqrt{n}\big(  \frac{ \sum_{k\in  {\mathcal{H}}_{(j)}} \hat{\beta}_k }{ | {\mathcal{H}}_{(j)} |} - \frac{ \sum_{k\in  {\mathcal{H}}_{(j)}} {\beta}_k }{ | {\mathcal{H}}_{(j)} |}  \big) \le t \Big| \hat{\mathcal{H}}_{(j)} = \mathcal{H}_{(j)}\Big) \mathbb{P}(\hat{\mathcal{H}}_{(j)} = \mathcal{H}_{(j)})  + o_\mathbb{P}(1) \Bigg]\\
	    &= \lim_{n\rightarrow\infty}  \mathbb{P}\Big(({\bm{e}}_j'\hat{\Sigma}^{\texttt{KJ}}_n {\bm{e}}_j  )^{-\frac{1}{2}}\cdot \sqrt{n}\big(  \frac{ \sum_{k\in  {\mathcal{H}}_{(j)}} \hat{\beta}_k }{ | {\mathcal{H}}_{(j)} |} - \frac{ \sum_{k\in  {\mathcal{H}}_{(j)}} {\beta}_k }{ | {\mathcal{H}}_{(j)} |}  \big) \le t \Big| \hat{\mathcal{H}}_{(j)} = \mathcal{H}_{(j)}\Big),
	\end{align*}
	in which $o_\mathbb{P}(1)$ is lower bounded by zero and upper bounded by $\mathbb{P}(\hat{\mathcal{H}}_{(j)} \neq \mathcal{H}_{(j)})$, which tends to zero when $n$ tends to infinity. We use this same $o_\mathbb{P}(1)$ notion throughout this proof. Now we have  
	\begin{align*}
\lim_{n\rightarrow \infty}\mathbb{P}\Big( (\hat{\bm{e}}_j'\hat{\Sigma}^{\texttt{KJ}}_n \hat{\bm{e}}_j  )^{-\frac{1}{2}}\cdot  \sqrt{n}\big(  \frac{ \sum_{k\in  \hat{\mathcal{H}}_{(j)}} \hat{\beta}_k }{ | \hat{\mathcal{H}}_{(j)} |} - \frac{ \sum_{k\in  \hat{\mathcal{H}}_{(j)}} {\beta}_k }{ | \hat{\mathcal{H}}_{(j)} |}  \big)   \leq t \Big| \hat{\mathcal{H}}_{(j)} = \mathcal{H}_{(j)}\Big)  = \Phi(t).
	\end{align*}
	Next, we have
	\begin{align*}
& \lim_{n\rightarrow \infty}\mathbb{P}\Big( (\hat{\bm{e}}_j'\hat{\Sigma}^{\texttt{KJ}}_n \hat{\bm{e}}_j  )^{-\frac{1}{2}}\cdot  \sqrt{n}\big(  \frac{ \sum_{k\in  \hat{\mathcal{H}}_{(j)}} \hat{\beta}_k }{ | \hat{\mathcal{H}}_{(j)} |} - \frac{ \sum_{k\in  \hat{\mathcal{H}}_{(j)}} {\beta}_k }{ | \hat{\mathcal{H}}_{(j)} |}  \big)   \leq t \Big)\\
=& \lim_{n\rightarrow \infty} \Bigg[ \mathbb{P}\Big( (\hat{\bm{e}}_j'\hat{\Sigma}^{\texttt{KJ}}_n \hat{\bm{e}}_j  )^{-\frac{1}{2}}\cdot  \sqrt{n}\big(  \frac{ \sum_{k\in  \hat{\mathcal{H}}_{(j)}} \hat{\beta}_k }{ | \hat{\mathcal{H}}_{(j)} |} - \frac{ \sum_{k\in  \hat{\mathcal{H}}_{(j)}} {\beta}_k }{ | \hat{\mathcal{H}}_{(j)} |}  \big)   \leq t \Big| \hat{\mathcal{H}}_{(j)} = \mathcal{H}_{(j)}\Big) \Big) \mathbb{P}(\hat{\mathcal{H}}_{(j)} = \mathcal{H}_{(j)}) + o_\mathbb{P}(1) \Bigg] \\
=& \lim_{n\rightarrow \infty}\mathbb{P}\Big( (\hat{\bm{e}}_j'\hat{\Sigma}^{\texttt{KJ}}_n \hat{\bm{e}}_j  )^{-\frac{1}{2}}\cdot  \sqrt{n}\big(  \frac{ \sum_{k\in  \hat{\mathcal{H}}_{(j)}} \hat{\beta}_k }{ | \hat{\mathcal{H}}_{(j)} |} - \frac{ \sum_{k\in  \hat{\mathcal{H}}_{(j)}} {\beta}_k }{ | \hat{\mathcal{H}}_{(j)} |}  \big)   \leq t \Big| \hat{\mathcal{H}}_{(j)} = \mathcal{H}_{(j)}\Big)\\
=&\,\Phi(t)
	\end{align*}
	The following holds precisely following the definition of the resampling procedure: 
	\begin{align}\label{proof-of-theorem-1}
\mathbb{P}\Big( ({\bm{e}}_j'\hat{\Sigma}^{\texttt{KJ}}_n {\bm{e}}_j  )^{-\frac{1}{2}}\cdot  \sqrt{n}\big(  \frac{ \sum_{k\in  {\mathcal{H}}_{(j)}} \hat{\beta}^*_k }{ | {\mathcal{H}}_{(j)} |} - \frac{ \sum_{k\in  {\mathcal{H}}_{(j)}} \hat{\beta}_k }{ | {\mathcal{H}}_{(j)} |}  \big)   \leq t \big| \{ \z_{i,n}\}_{i=1}^n  \Big)  = \Phi(t).
\end{align}	
We now show that
\begin{align*}
& \lim_{n\rightarrow \infty} \mathbb{P}\Big( (\hat{\bm{e}}_j'\hat{\Sigma}^{\texttt{KJ}}_n \hat{\bm{e}}_j  )^{-\frac{1}{2}}\cdot  \sqrt{n}\big(  \frac{ \sum_{k\in  \hat{\mathcal{H}}_{(j)}} \hat{\beta}^*_k }{ | \hat{\mathcal{H}}_{(j)} |} - \frac{ \sum_{k\in  \hat{\mathcal{H}}_{(j)}} \hat{\beta}_k }{ | \hat{\mathcal{H}}_{(j)} |}  \big)   \leq t \big| \{ \z_{i,n}\}_{i=1}^n  \Big)  \\ 
=& \lim_{n\rightarrow \infty} \Bigg[\mathbb{P}\Big( (\hat{\bm{e}}_j'\hat{\Sigma}^{\texttt{KJ}}_n \hat{\bm{e}}_j  )^{-\frac{1}{2}}\cdot  \sqrt{n}\big(  \frac{ \sum_{k\in  \hat{\mathcal{H}}_{(j)}} \hat{\beta}^*_k }{ | \hat{\mathcal{H}}_{(j)} |} - \frac{ \sum_{k\in  \hat{\mathcal{H}}_{(j)}} \hat{\beta}_k }{ | \hat{\mathcal{H}}_{(j)} |}  \big)   \leq t \big| \{ \z_{i,n}\}_{i=1}^n, \hat{\mathcal{H}}_{(j)} = \mathcal{H}_{(j)} \Big) \mathbb{P}(\hat{\mathcal{H}}_{(j)} = \mathcal{H}_{(j)}) + o_\mathbb{P}(1) \Bigg] \\ 
=& \lim_{n\rightarrow \infty} \Bigg[\mathbb{P}\Big( ({\bm{e}}_j'\hat{\Sigma}^{\texttt{KJ}}_n {\bm{e}}_j  )^{-\frac{1}{2}}\cdot  \sqrt{n}\big(  \frac{ \sum_{k\in  {\mathcal{H}}_{(j)}} \hat{\beta}^*_k }{ | {\mathcal{H}}_{(j)} |} - \frac{ \sum_{k\in  {\mathcal{H}}_{(j)}} \hat{\beta}_k }{ | {\mathcal{H}}_{(j)} |}  \big)   \leq t \big| \{ \z_{i,n}\}_{i=1}^n, \hat{\mathcal{H}}_{(j)} = \mathcal{H}_{(j)} \Big) \mathbb{P}(\hat{\mathcal{H}}_{(j)} = \mathcal{H}_{(j)}) + o_\mathbb{P}(1) \Bigg]\\
=&  \lim_{n\rightarrow \infty} \mathbb{P}\Big( ({\bm{e}}_j'\hat{\Sigma}^{\texttt{KJ}}_n {\bm{e}}_j  )^{-\frac{1}{2}}\cdot  \sqrt{n}\big(  \frac{ \sum_{k\in  {\mathcal{H}}_{(j)}} \hat{\beta}^*_k }{ | {\mathcal{H}}_{(j)} |} - \frac{ \sum_{k\in  {\mathcal{H}}_{(j)}} \hat{\beta}_k }{ | {\mathcal{H}}_{(j)} |}  \big)   \leq t \big| \{ \z_{i,n}\}_{i=1}^n \Big)\\
=&\, \Phi(t).
\end{align*}
Recall our definition in the main manuscript
\begin{align}
\tilde{\beta}_{(j)}^* = \frac{ \sum_{k\in  \widehat{\mathcal{H}}_{(j)}} \hat{\beta}^*_k }{ | \widehat{\mathcal{H}}_{(j)} |}, \text{ and } \tilde{\beta}_{(j)} = \frac{ \sum_{k\in  \widehat{\mathcal{H}}_{(j)}} \hat{\beta}_k }{ | \widehat{\mathcal{H}}_{(j)} |}, 
\end{align}  
we thus have reached the conclusion presented in the theorem:
	\begin{align*}
\lim_{n\rightarrow \infty} \mathbb{P}\Big( (\hat{\bm{e}}_j'\hat{\Sigma}^{\texttt{KJ}}_n \hat{\bm{e}}_j  )^{-\frac{1}{2}}\cdot  \sqrt{n}\big( \tilde{\beta}_{(j)}^*- \tilde{\beta}_{(j)}   \big)   \leq t \big| \{ \z_{i,n}\}_{i=1}^n  \Big)  = \Phi(t). 
\end{align*}
	\item[Step 3.] 
	Lastly, to prove the bootstrap consistency, we show that 
	\begin{align*}
\lim_{n\rightarrow \infty }\mathbb{P}\Big( 	\mathbb{P}\big( 	\tilde{\beta}_{(j)}^* \leq  \beta_{(j)}   |\{ \z_{i,n}\}_{i=1}^n   \big) \leq s \Big)  = s.
	\end{align*}
	Note that 
	\begin{align*}
	  & \lim_{n\rightarrow\infty}	\mathbb{P}\big( 	\tilde{\beta}_{(j)}^* \leq  \beta_{(j)}   |\{ \z_{i,n}\}_{i=1}^n   \big)  \\
	  =&\lim_{n\rightarrow\infty} \mathbb{P}\Big( ({\bm{e}}_j'\hat{\bm{\Omega}}^{\texttt{KJ}}_n {\bm{e}}_j  )^{-\frac{1}{2}}\cdot  \sqrt{n}\big(  \frac{ \sum_{k\in  {\mathcal{H}}_{(j)}} \hat{\beta}^*_k }{ | {\mathcal{H}}_{(j)} |} - \frac{ \sum_{k\in  {\mathcal{H}}_{(j)}} \hat{\beta}_k }{ | {\mathcal{H}}_{(j)} |}  \big) \\
	  \leq& ({\bm{e}}_j'\hat{\bm{\Omega}}^{\texttt{KJ}}_n {\bm{e}}_j  )^{-\frac{1}{2}}\cdot  \sqrt{n}\big(  \frac{ \sum_{k\in  {\mathcal{H}}_{(j)}} {\beta}_k }{ | {\mathcal{H}}_{(j)} |} - \frac{ \sum_{k\in  {\mathcal{H}}_{(j)}} \hat{\beta}_k }{ | {\mathcal{H}}_{(j)} |}  \big)  \big| \{ \z_{i,n}\}_{i=1}^n  \Big)  \\
	  =& \lim_{n\rightarrow\infty} \Phi(({\bm{e}}_j'\hat{\bm{\Omega}}^{\texttt{KJ}}_n {\bm{e}}_j  )^{-\frac{1}{2}}\cdot  \sqrt{n}\big(  \frac{ \sum_{k\in  {\mathcal{H}}_{(j)}} {\beta}_k }{ | {\mathcal{H}}_{(j)} |} - \frac{ \sum_{k\in  {\mathcal{H}}_{(j)}} \hat{\beta}_k }{ | {\mathcal{H}}_{(j)} |}  \big) ).
	\end{align*}
	Therefore, by the bounded convergence theorem, we have
	\begin{align*}
&\lim_{n\rightarrow \infty }\mathbb{P}\Big( 	\mathbb{P}\big( 	\tilde{\beta}_{(j)}^* \leq  \beta_{(j)}   |\{ \z_{i,n}\}_{i=1}^n   \big) \leq s \Big) \\
=&\ \mathbb{P}(\lim_{n\rightarrow \infty }\Phi(({\bm{e}}_j'\hat{\bm{\Omega}}^{\texttt{KJ}}_n {\bm{e}}_j  )^{-\frac{1}{2}}\cdot  \sqrt{n}\big(  \frac{ \sum_{k\in  {\mathcal{H}}_{(j)}} {\beta}_k }{ | {\mathcal{H}}_{(j)} |} - \frac{ \sum_{k\in  {\mathcal{H}}_{(j)}} \hat{\beta}_k }{ | {\mathcal{H}}_{(j)} |}  \big) ) \le s) \\
=&\ \mathbb{P}(\lim_{n\rightarrow \infty }({\bm{e}}_j'\hat{\bm{\Omega}}^{\texttt{KJ}}_n {\bm{e}}_j  )^{-\frac{1}{2}}\cdot  \sqrt{n}\big(  \frac{ \sum_{k\in  {\mathcal{H}}_{(j)}} {\beta}_k }{ | {\mathcal{H}}_{(j)} |} - \frac{ \sum_{k\in  {\mathcal{H}}_{(j)}} \hat{\beta}_k }{ | {\mathcal{H}}_{(j)} |}  \big) \le \Phi^{-1}(s))\\
=&\ \mathbb{P}(N(0,1)\le \Phi^{-1}(s))\\
=&\ s.
	\end{align*}
	
% 	\[
% 	\mathbb{P}\Big( ({\bm{e}}_j'\hat{\Sigma}^{\texttt{KJ}}_n {\bm{e}}_j  )^{-\frac{1}{2}}\cdot  \sqrt{n}\big(  \frac{ \sum_{k\in  {\mathcal{H}}_{(j)}} \hat{\beta}^*_k }{ | {\mathcal{H}}_{(j)} |} - \frac{ \sum_{k\in  {\mathcal{H}}_{(j)}} \hat{\beta}_k }{ | {\mathcal{H}}_{(j)} |}  \big)   \leq t \big| \{ \z_{i,n}\}_{i=1}^n  \Big)  = \Phi(t).
% 	\]
	
% 	\[
% 	({\bm{e}}_j'\hat{\Sigma}^{\texttt{KJ}}_n {\bm{e}}_j  )^{-\frac{1}{2}}\cdot \sqrt{n}\big(  \frac{ \sum_{k\in  {\mathcal{H}}_{(j)}} \hat{\beta}_k }{ | {\mathcal{H}}_{(j)} |} - \frac{ \sum_{k\in  {\mathcal{H}}_{(j)}} {\beta}_k }{ | {\mathcal{H}}_{(j)} |}  \big)  \leadsto \mathcal{N}(0,1)
% 	\]
\end{description}

\end{proof}

\section{Lemmas and corollary}
\subsection{Lemma \ref{Lemma:tie-set-consistency}}

\begin{lem}\label{Lemma:tie-set-consistency}
	Suppose $w_{k,(j)} = \mathds{1}( k \in \hat{\mathcal{H}}_{(j)} )$, for $j, k = 1, \ldots, d$, under Assumptions \ref{supp:Assumption:Sampling}-\ref{supp:Assumption:seperation}, we have the following argument holds $\forall \varepsilon >0$,
	\begin{align*}
		\lim_{n\rightarrow \infty}\mathbb{P}\big(  |w_{k,(j)} - \mathds{1}(k \in \mathcal{H}_{(j)} ) | > \varepsilon   \big) =0.
	\end{align*}

\end{lem}

\begin{proof}
%	For any $\varepsilon >0$, 
%	\begin{align*}
%	\mathbb{P}\big( \sqrt{n} |w_{k,(j)} -1 | \geq \varepsilon   \big) = 	\mathbb{E}\Big[ \mathbb{P}\big( \sqrt{n} |w_{k,(j)} -1 | \geq \varepsilon |   \z_{i,n}\}_{i=1}^n   \big)\Big]
%	\end{align*}
%	
We start with reviewing and introducing some notations to pave the way for a clear proof. Recall that $\tilde{\beta}_{(j)}^*$ is the $j$-th largest effect size for the resampled statistics $\hat{\beta}_{(1)}, \ldots, \hat{\beta}_{(d)}$, suppose $\tilde{\beta}_{(j)}^*$ is resampled statistics from the normal distribution centered at $ \hat{\beta}_{\widecheck{j}}$, that is 
	\begin{align*}
  	\tilde{\beta}_{(j)}^* \big|\{ \z_{i,n}\}_{i=1}^n \sim \mathcal{N}( \hat{\beta}_{\check{j}} , (\hat{\Sigma}^{\texttt{KJ}}_{n})_{\check{j}\check{j} }    ), \quad \check{j} = \sum_{k=1}^d k\cdot \mathds{1}( \hat{\beta}_k^* = 	\tilde{\beta}_{(j)}^* ),
	\end{align*}
	where $(\hat{\Sigma}^{\texttt{KJ}}_{n})_{\check{j}\check{j} }  $ is the $\check{j}$th component in the diagonal of the matrix $\hat{\Sigma}^{\texttt{KJ}}_{n}$. 
	
	Recall we define the near tie $	\mathcal{H}_{(j)} $ set as 
	\begin{align*}
	\mathcal{H}_{(j)} = \big\{ k: | \beta_k  - \beta_{(j)} | = o(n^{-\frac{1}{2}}),\ k=1, \ldots, d \big\}.
	\end{align*}
	We further define two sets of policies that have effect sizes lower/larger than the policies in the set $\mathcal{H}_{(j)} $: 
	\begin{align*}
&	\mathcal{H}_{(j)}^{L} =  \big\{ k: \  \beta_k < \min_{m\in   \mathcal{H}_{(j)} } \{\beta_m\} \ k=1, \ldots, d \big\}, \quad 	\mathcal{H}_{(j)}^{U} =  \big\{ k: \  \beta_k > \max_{m\in   \mathcal{H}_{(j)} } \{\beta_m\} \ k=1, \ldots, d \big\}. 
	\end{align*}
	As for the estimated near tie set, we have for any $j\in \hat{\mathcal{H}}_{(j)}$ that 
	\begin{align*}
	- b_L \leq \hat{\beta}^*_{k} -\tilde{\beta}^*_{ {{(j)}}} \leq b_R, \quad \text{with } |b_R- b_L | = O ( n^{-\delta} ),
	\end{align*}
	where $\delta\in(0,0.5)$. Thus, there exists a positive constant $C$ such that 
\begin{align*}
 \frac{| \hat{\beta}^*_{k} -\tilde{\beta}^*_{ {{(j)}}} | }{ n^{-\delta}} < 	C, \quad \forall j \in  \hat{\mathcal{H}}_{(j)}.
\end{align*}	
	
	Our proof is composed of the following three steps: 
	\begin{description}
		\item[Step 1.] We first show that the policy with the $j$th largest policy effect size in the resampled statistics falls into the set $ \mathcal{H}_{(j)} $ with high probability, that is 
		\begin{align}\label{eq:Lemma1-step1}
	\lim_{n\rightarrow \infty}	\mathbb{P}\Big(  \check{j} \in  \mathcal{H}_{(j)} \Big) =1. 
		\end{align}
	Because $  \hat{\beta}_{\check{j}}^*  \in [\min_{j\in \mathcal{H}_{(j)}}  \hat{\beta}_j^* ,  \max _{j\in \mathcal{H}_{(j)}}  \hat{\beta}_j^*]  $ by definition, coupled with the fact that 
			\begin{align*}
		\Big\{ \max_{k\in 	\mathcal{H}_{(j)}^{L} } \hat{\beta}_k^* < \min_{j\in \mathcal{H}_{(j)}}  \hat{\beta}_j^* \leq  \max _{j\in \mathcal{H}_{(j)}}  \hat{\beta}_j^* <   \min_{k\in 	\mathcal{H}_{(j)}^{U} } \hat{\beta}_k^* \Big\}  \subset \Big(  \check{j} \in  \mathcal{H}_{(j)} \Big) ,
		\end{align*}	
		it is suffice to show 
	\begin{align*}
\lim_{n\rightarrow \infty }	\mathbb{P}\left(	 \max_{k\in 	\mathcal{H}_{(j)}^{L} } \hat{\beta}_k^* < \min_{j\in \mathcal{H}_{(j)}}  \hat{\beta}_j^* \leq  \max _{j\in \mathcal{H}_{(j)}}  \hat{\beta}_j^* <   \min_{k\in 	\mathcal{H}_{(j)}^{U} } \hat{\beta}_k^*  \right) = 1. 
	\end{align*}
	
	Under Assumption \ref{supp:Assumption:seperation}, by Lemma \ref{lemma:bootstrap-convergence-rate}, for any $k\in \mathcal{H}_{(j)}^L$ and $m \in  \mathcal{H}_{(j)}$, we have 
	\begin{align*}
	\lim_{n\rightarrow \infty}\mathbb{P}\left(  \hat{\beta}_{k}^*  < \hat{\beta}_{m}^*  \right) =1 . 
	\end{align*}
	
	Similarly, for any $k\in \mathcal{H}_{(j)}^U$ and $m \in  \mathcal{H}_{(j)}$, we have 
	\begin{align*}
	\lim_{n\rightarrow \infty}\mathbb{P}\left(  \hat{\beta}_{m}^*  < \hat{\beta}_{k}^*  \right) =1 . 
	\end{align*}
	
	\item[Step 2.] We then show, for $k\not \in \mathcal{H}_{(j)} $
\begin{align*}
	\lim_{n\rightarrow \infty}\mathbb{P}\big(  w_{k,(j)}   > \varepsilon , k\not \in \mathcal{H}_{(j)}   \big) = 0 .
\end{align*}
For any $\varepsilon >0$ and $k\not \in \mathcal{H}_{(j)} $, we have the following holds
\begin{align*}
\mathbb{P}\big(  w_{k,(j)}   > \varepsilon   \big) = &\ \mathbb{P}\big( \mathds{1}( k \in \hat{\mathcal{H}}_{(j)} ) > \varepsilon  \big) \\
=&\  \mathbb{P}\big( \mathds{1}( k \in \hat{\mathcal{H}}_{(j)} ) > \varepsilon |   k \in \hat{\mathcal{H}}_{(j)}  \big) \cdot  \mathbb{P}\big(  k \in \hat{\mathcal{H}}_{(j)} \big)   \\
&\ + \mathbb{P}\big(\mathds{1}( k \in \hat{\mathcal{H}}_{(j)} ) > \varepsilon |   k \not\in \hat{\mathcal{H}}_{(j)}  \big) \cdot  \mathbb{P}\big(  k \not\in \hat{\mathcal{H}}_{(j)} \big)  \\
\leq & \ \mathbb{P}\big(  k \in \hat{\mathcal{H}}_{(j)} \big) \\
\overset{\text{Def}}{=}&\ \mathbb{P}\Big(  \frac{| \hat{\beta}^*_{k} -\hat{\beta}^*_{ {{(j)}}} | }{ n^{-\delta}} < 	C \Big) \\
= &\ \mathbb{P}\Big(  \frac{| \hat{\beta}^*_{k} -\hat{\beta}^*_{\check{j}} | }{ n^{-\delta}} < 	C \Big) \\
= &\ \mathbb{P}\Big( | \hat{\beta}^*_{k} -\hat{\beta}^*_{\check{j}}  | < 	n^{-\delta}\cdot C \Big) \\
=& \ \mathbb{P}\Big( | (\hat{\beta}^*_{k} - \beta_k) -  (\hat{\beta}^*_{\check{j}} - \beta_{\check{j}})+ (\beta_k - \beta_{\check{j}})  | < 	n^{-\delta}\cdot C \Big) \\
\leq & \ \mathbb{P}\Big( n^{\delta}|\beta_k - \beta_{\check{j}}| -  n^{\delta}| \hat{\beta}^*_{k} - \beta_k| -  n^{\delta} |\hat{\beta}^*_{\check{j}} - \beta_{\check{j}} | <  C \Big) \\
= & \ \mathbb{P}\Big(   n^{\delta}| \hat{\beta}^*_{k} - \beta_k| +  n^{\delta} |\hat{\beta}^*_{\check{j}} - \beta_{\check{j}} | > n^{\delta}|\beta_k - \beta_{\check{j}}|  + C \Big) 
\end{align*}

By definition, for $k\not\in \mathcal{H}_{(j)}$ 
\begin{align*}
\mathbb{P}\big(  n^{\delta} | \beta_k - \beta_{\check{j}} | < C \big) \leq & \  \mathbb{P}\big(  n^{\delta} | \beta_k - \beta_{\check{j}} | <  C , \check{j} \in  \mathcal{H}_{(j)}  \big) + \mathbb{P}\big(\check{j} \not\in  \mathcal{H}_{(j)}  \big)\\
\leq &\ \max_{j\in \mathcal{H}_j}\ \mathbb{P}\big(  n^{\delta} | \beta_k - \beta_{{j}} | <  C , {j} \in  \mathcal{H}_{(j)}  \big) + \mathbb{P}\big(\check{j} \not\in  \mathcal{H}_{(j)}  \big).
\end{align*}
Under Assumption \ref{supp:Assumption:seperation}, Lemma \ref{lemma:bootstrap-convergence-rate} and the conclusion in Eq \eqref{eq:Lemma1-step1} in Step 1 suggest that by letting $n\rightarrow \infty$ on both side, we have the above probability converges to zero. Based on above derivation, we have shown that $	\lim_{n\rightarrow \infty}\mathbb{P}\big(w_{k,(j)}   > \varepsilon , k\not \in \mathcal{H}_{(j)}   \big) = 0 $, for $k\not\in \mathcal{H}_{(j)}$. 

\item[Step 3.] We are left to prove that for all $k \in  \mathcal{H}_{(j)}$, the following holds $\forall \varepsilon>0$:
\begin{align*}
\lim_{n\rightarrow \infty}\mathbb{P}\big(  |w_{k,(j)} - 1  | > \varepsilon   \big) =0.
\end{align*}
Following similar arguments, for a positive constant $C$, we have for $k,j \in  \mathcal{H}_{(j)} $ the following statement holds
\begin{align*}
\mathbb{P}\big( |w_{k,(j)} - 1  | > \varepsilon  \big) = &\ \mathbb{P}\big(| \mathds{1}( k \in \hat{\mathcal{H}}_{(j)})-1 | > \varepsilon  \big) \\
=&\  \mathbb{P}\big( | \mathds{1}( k \in \hat{\mathcal{H}}_{(j)} )-1| > \varepsilon |   k \in \hat{\mathcal{H}}_{(j)}  \big) \cdot  \mathbb{P}\big(  k \in \hat{\mathcal{H}}_{(j)} \big)   \\
& \ +\mathbb{P}\big(| \mathds{1}( k \in \hat{\mathcal{H}}_{(j)} )-1| > \varepsilon |   k \not\in \hat{\mathcal{H}}_{(j)}  \big) \cdot  \mathbb{P}\big(  k \not\in \hat{\mathcal{H}}_{(j)} \big)   \\
\leq &\   \mathbb{P}\big(  k \not\in \hat{\mathcal{H}}_{(j)} \big) \\
\overset{\text{Def}}{=}&\ \mathbb{P}\Big(  \frac{| \hat{\beta}^*_{k} -\hat{\beta}^*_{ {{(j)}}} | }{ n^{-\delta}} \geq 	C  \Big) \\
= &\ \mathbb{P}\Big(  \frac{| \hat{\beta}^*_{k} -\hat{\beta}^*_{\check{j}} | }{ n^{-\delta}} \geq	C \Big) \\
=& \ \mathbb{P}\Big( | (\hat{\beta}^*_{k} - \beta_k) -  (\hat{\beta}^*_{\check{j}} - \beta_{\check{j}})+ (\beta_k - \beta_{\check{j}})  | \geq	n^{-\delta}\cdot C \Big) \\
\leq  & \ \mathbb{P}\Big( | \hat{\beta}^*_{k} - \beta_k| + |\hat{\beta}^*_{\check{j}} - \beta_{\check{j}}|+ |\beta_k - \beta_{\check{j}}|   \geq	n^{-\delta}\cdot C \Big) \\
\leq & \ \mathbb{P}\Big( n^{\frac{1}{2}} | \hat{\beta}^*_{k} - \beta_k| + n^{\frac{1}{2}}|\hat{\beta}^*_{\check{j}} - \beta_{\check{j}}|+ n^{\frac{1}{2}}|\beta_k - \beta_{\check{j}}|   \geq	n^{\frac{1}{2}-\delta}\cdot C \Big). 
\end{align*}

By definition of the near-tie set, for $k\in \mathcal{H}_{(j)}$, we have 
\begin{align*}
\mathbb{P}\big(  n^{\frac{1}{2}} | \beta_k - \beta_{\check{j}} | < C \big) \leq & \  \mathbb{P}\big(  n^{\delta} | \beta_k - \beta_{\check{j}} | <  C , \check{j} \in  \mathcal{H}_{(j)}  \big) + \mathbb{P}\big(\check{j} \not\in  \mathcal{H}_{(j)}  \big)\\
\leq &\ \max_{j\in \mathcal{H}_j}\ \mathbb{P}\big(  n^{\delta} | \beta_k - \beta_{{j}} | <  C , {j} \in  \mathcal{H}_{(j)}  \big) + \mathbb{P}\big(\check{j} \not\in  \mathcal{H}_{(j)}  \big).
\end{align*}
Again, under Assumption \ref{supp:Assumption:seperation}, Lemma \ref{lemma:bootstrap-convergence-rate} and the conclusion in Eq \eqref{eq:Lemma1-step1} we have derived in Step 1, by letting $n\rightarrow \infty$ on both side, we have the above probability converges to 1. 

\end{description}
	
\end{proof}

\subsection{Lemma \ref{lemma:bootstrap-convergence-rate}}

\begin{lem}\label{lemma:bootstrap-convergence-rate}
	Under Assumption \ref{supp:Assumption:seperation}, we show that for all $k\in \{1, \ldots, d\}$,  any positive constant $C$ and $\delta \in (0,\frac{1}{2})$, the following statement holds 
	\begin{align*}
	\lim_{n\rightarrow\infty}\mathbb{P}\big( | \hat{\beta}^*_{k} - \beta_k| \geq n^{-\delta}\cdot C \big) = 0. 
	\end{align*}
\end{lem}

\begin{proof}
    Note that 
    \begin{align}
        \sqrt{n}( \hat{\beta}^*_{k} - \beta_k) &= \sqrt{n}(\hat{\beta}_k-\beta_k) + \mathcal{N}\big(0,(\hat{\bm{\Omega}}^{\texttt{KJ}}_n)_{k,k} \big).
    \end{align}
    Because $\sqrt{n}(\hat{\beta}_k-\beta_k)$ converges in distribution to a finite-value random variable and $\hat{\bm{\Omega}}^{\texttt{KJ}}_n $ converges in probability to a finite-value matrix when $n$ tends to infinity \citep{jochmans2020heteroscedasticity}, for any given $\epsilon>0$, there exists an $M$, such that 
    \begin{align}
        \mathbb{P}(\sqrt{n}| \hat{\beta}^*_{k} - \beta_k| >M )<\epsilon.
    \end{align}
    Then, for any $n$ such that 
    \[
    n > \big(\frac{M}{C}\big)^{\frac{1}{\frac{1}{2}-\delta}}
    \]
    we have that
    \[
    \mathbb{P}\big( | \hat{\beta}^*_{k} - \beta_k| \geq n^{-\delta}\cdot C \big) <\epsilon
    \]
    and therefore
    \[
    \limsup_{n\rightarrow\infty}\mathbb{P}\big( | \hat{\beta}^*_{k} - \beta_k| \geq n^{-\delta}\cdot C \big) <\epsilon.
    \]
    Note that the above inequality holds for arbitrary $\epsilon>0$. Therefore, we have 
    \[
    \limsup_{n\rightarrow\infty}\mathbb{P}\big( | \hat{\beta}^*_{k} - \beta_k| \geq n^{-\delta}\cdot C \big)=0,
    \]
    completing the proof.
\end{proof}

 \begin{lem}
     Denote the selected policy as 
     \begin{align*}
        \hat{j} = \sum_{k=1}^d k\cdot \mathds{1}( \hat{\beta}_k = 	\hat{\beta}_{(j)} ),
     \end{align*}
     we show that 
     \begin{align*}
        \lim_{n\rightarrow \infty}\mathbb{P}\big( \hat{j}\in\mathcal{H}_{(j)} \big) =1. 
     \end{align*}
 \end{lem}
 \begin{proof}
     This is a direct result from Step 1 and Step 2 in the proof for Theorem \ref{Theorem:resampling-consistency}.
 \end{proof}

\subsection{Corollary \ref{corollary:observed-policy}}
\begin{coro}
\label{supp:corollary:observed-policy}
Under Assumptions \ref{supp:Assumption:Sampling}-\ref{supp:Assumption:seperation}, we have that  $\lim_{n\rightarrow \infty }\mathbb{P}\Big( 	\mathbb{P}\big( 	\tilde{\beta}_{(j)}^* \leq  \beta_{\hat{j}}   |\{ \z_{i,n}\}_{i=1}^n   \big) \leq s \Big)  = s$. 
\end{coro}
\begin{proof}
    Because of the consistency in Lemma 3, we have
    \[
    \lim_{n\rightarrow\infty}\mathbb{P} (\beta_{\hat{j}} = \beta_{(j)}) = 1.
    \]
    Therefore, 
    \begin{align*}
        \lim_{n\rightarrow\infty}|\mathbb{P}\big( 	\tilde{\beta}_{(j)}^* \leq  \beta_{(j)}   |\{ \z_{i,n}\}_{i=1}^n   \big) -	\mathbb{P}\big( 	\tilde{\beta}_{(j)}^* \leq  \beta_{(j)}   |\{ \z_{i,n}\}_{i=1}^n   \big) | \le \lim_{n\rightarrow\infty}\mathbb{P}(\beta_{(j)}\neq 	\hat{\beta}_{(j)}|\{ \z_{i,n}\}_{i=1}^n ) = 0. 
    \end{align*}
    The result then follows by applying Step 3 in the Proof of Theorem \ref{Theorem:resampling-consistency}. 
\end{proof}

\section{Additional simulation and empirical results}

\subsection{Practical implementation}\label{Sec:implementation}

In this section, we discuss the choice of tuning parameters (including $B$, $\delta$, $\bl$ and $\br$) of the proposed method in Section 2.2. For the number of repetitions for our resampling procedure, we recommend using $B=2,000$ as a good balance between computational load and statistical inference accuracy. 

For the tuning pair $(\br, \bl)$, from our theoretical analysis, we need to ensure that the distance between $\br$ and $\bl$ is of the order $n^{\delta}$ with $\delta\in (0,0.5)$ to guarantee the statistical validity of our proposed procedure. To achieve this goal, for the policy $\beta_{(j)}$, we adopt the tuning pair of the form 
\begin{align*}
    \bl^j= n^{-\delta}\cdot s_{\hat{j}}^{\delta} \cdot c_{\texttt{L}}^j, \quad  \br^j= n^{-\delta}\cdot s_{\hat{j}}^{\delta}\cdot c_{\texttt{R}}^j, 
\end{align*}
where $s_{\hat{j}}$ is the $\hat{j}$th element in the diagonal of the estimated covariance matrix $\hat{\bm{\Omega}}^{\texttt{KJ}}_n$, $c_{\texttt{L}}^j$ and $c_{\texttt{R}}^j$ are positive constants. 

The constants $c_{\texttt{L}}^j$ and $c_{\texttt{R}}^j$ can significantly impact the performance of the proposed approach in finite samples. In the extreme cases, on the one hand, if both $c_{\texttt{L}}^j$ and $c_{\texttt{R}}^j$ are overly large, the estimated near tie set might include more policies than necessary and our approach is only valid if all true policy effects are closely ties. On the other hand, if $c_{\texttt{L}}^j$ and $c_{\texttt{R}}^j$ are both closer to zero, our approach reduces to a standard parametric bootstrap approach, which is problematic in the presence of tied policy effects. To present a robust algorithm in finite samples, we thus adopt the following ``double-bootstrap" method as discussed in \cite{claggett2014meta} (note that double-bootstrapped statistics are labelled with double-star superscripts): 

%we decompose the tuning parameters in the following way to make sure they converge to 0 at a proper rate$$\bl=\tau_n\cdot  c_{\texttt{L}}^{j},\br=\tau_n\cdot  c_{\texttt{R}}$$where $\tau_n=O(n^{-\delta}),0<\delta<1/2$ and $ c_{\texttt{L}}^{j}, c_{\texttt{R}}=O(1)$ are positive. For simplicity, We can use the definition in \cite{claggett2014meta} that $\tau_n=(s^{(j)})^{2\delta}$, where $s^{(j)}$ is the estimated standard error of $\beta^{(j)}$. 

%The constants $( c_{\texttt{L}}^{j}, c_{\texttt{R}})$ need to be carefully chose. Take estimation of $\beta_{max}$ for example, too large values will involve more estimated subgroup effects in the mean result, which is inappropriate when $\lvert\mathcal{H}_{(j)}\rvert=1$. Similarly, too small values are not reasonable when $\lvert\mathcal{H}_{(j)}\rvert>1$. For a proper choice of $( c_{\texttt{L}}^{j}, c_{\texttt{R}})$, we use the ``double-bootstrap" method as shown in \cite{claggett2014meta}. In general, we assume the estimated subgroup effects as the ``true" values, and use the bootstrap values as new estimated subgroup effects. Then we test $( c_{\texttt{L}}^{j}, c_{\texttt{R}})$ with our proposed method based on the bootstrap  values to see how different the final estimation is from the ``true" value and pick the closest ones. And we use grid search for the best selected $( c_{\texttt{L}}^{j}, c_{\texttt{R}})$. 

\begin{enumerate}     
\item For $j=1,...,d$, set $\beta_j^*=\Delta\cdot\frac{\sum_{j=1}^{d}\hat\beta_j}{d}+(1-\Delta)\cdot\hat\beta_j$, where
$$\Delta={\min}\Big\{1,\frac{\sum_{j=1}^d s_{\hat{j}}}{n\sum_{j=1}^d(\hat\beta_j-\Bar{\hat\beta})^2}\times n^{0.05}\Big\}.$$

\item For every candidate pair $(c_{\texttt{L}}, c_{\texttt{R}})$ such that $c_{\texttt{L}}\in \mathcal{C}_{\texttt{L}}$ and $c_{\texttt{R}}\in \mathcal{C}_{\texttt{R}}$, do
        \begin{enumerate}[]
            \item For $t\gets 1$ to $T$, do
            \begin{enumerate}
			\item Generate $\hat{\bbeta}^*=(\hat\beta_1^*,...,\hat\beta_{d}^*)$ from $ \mathcal{N}(\bbeta^*,\hat \Omega_n^{\texttt{KJ}}/n)$, where $\bbeta^* = (\beta_1^*, \ldots, \beta^*_d)'$, and denote the ordered values in $\bbeta^*$ as $\beta_{(1)}^*\geq \ldots \geq \beta_{(d)}^*$.
			\item For $r\gets 1$ to $R$, do
			\begin{enumerate}
			    \item  Generate double bootstrap statistics $\hat{\bm{\beta}}^{**} \triangleq (\hat{\beta}_1^{**},...,\hat{\beta}_{d}^{**})'$ from $\mathcal{N}(\hat{\bbeta}^*,\hat \Omega_n^{\texttt{KJ}}/n)$, and denote the ordered values of $\hat{\bm{\beta}}^{**}$ as $\hat{\beta}_{(1)}^{**}\geq ...\geq \hat{\beta}_{(d)}^{**}$. 
			    \item Record 
			    $w_{k, (j)}^{**}=\mathds{1}\{- c_{\texttt{L}}\cdot n^{-\delta}\cdot s_{\hat{j}}^{\delta}\le(\hat{\beta}_k^{**}-\beta_{(j)}^{**})\le  c_{\texttt{R}}\cdot n^{-\delta}\cdot s_{\hat{j}}^{\delta}\}$
			    and
			    $\tilde{\beta}^{**}_{(j)}=\sum_{k=1}^{d}w_{k, (j)}^{**} \hat{\beta}_j^{**}/\sum_{k=1}^{d}w_{k, (j)}^{**}$, for $j=1,\ldots, d$.
			\end{enumerate}
            \item Calculate ${B}_{j,t}( c_{\texttt{L}}, c_{\texttt{R}})=\frac{1}{R} \sum_{r=1}^R\mathds{1}\big(\tilde{\beta}^{**, r}_{(j)}\le\beta^{*,r}_{(j)}\big) $,  for $j=1,\ldots, d$.
	    \end{enumerate}
        \end{enumerate}
 \item Record the loss function  
     \begin{align}\label{Eq:tuning-loss-function}
     {L}_j( c_{\texttt{L}}, c_{\texttt{R}})=\frac{1}{T}\sum_{t=1}^{T}\Big( {B}_{j,(t)}( c_{\texttt{L}}, c_{\texttt{R}})-\frac{t}{T+1}\Big)^2,
    \end{align}
 where ${B}_{j,(t)}( c_{\texttt{L}}, c_{\texttt{R}})$ is the $t$-th smallest statistics in ${B}_{j,1}( c_{\texttt{L}}, c_{\texttt{R}}), \ldots, {B}_{j,T}( c_{\texttt{L}}, c_{\texttt{R}})$. 
\item Choose the pair $( c_{\texttt{L}}^{j}, c_{\texttt{R}}^j)$ for inferring $\beta_{(j)}$ and $\beta_{\hat{j}}$ that minimizes ${L}_j( c_{\texttt{L}}, c_{\texttt{R}})$, that is
\begin{align*}
( c_{\texttt{L}}^{j}, c_{\texttt{R}}^j) = \underset{(c_{\texttt{L}}, c_{\texttt{R}})\in \mathcal{C}}{\min}{L}_j( c_{\texttt{L}}, c_{\texttt{R}}), \quad j=1, \ldots, d_0. 
\end{align*}
\end{enumerate}
 Note that we only use the above procedure to choose the tuning parameters $c_{\texttt{L}}^j$ and $c_{\texttt{R}}^j$, meaning that we do not use the resampled statistics in Step 1 to carry out inference on $\beta_{(j)}$. In Step 1, $\Delta$ is adopted to stabilize the performance of the tuning parameter selection in finite samples, and $\Delta$ only takes a close-to-zero value whenever limited variation is found between policy effect estimates.

 %To suppose the estimated policy effect as ``true" value, we need to shrink them towards their mean by their ratio of the within policy variation ($\sum_{j=1}^ds_j/n$) to the total variation in these policies ($\sum_{j=1}^d\hat\beta_j^2-\{\sum_{j=1}^d\hat\beta_j\}^2$) to reduce excess dispersion like the first step because $\hat\bbeta$ typically has wider spread than $\bbeta$. When the effect sizes are all the same, $\{\hat\beta_j\}$ are still different from each other so we need $\Delta$ to go to 1 as $n\to\infty$ and when policy effects are all different, $\Delta$ should go to 0. Other reasonable ways to set $\Delta$ can also be considered. 

Following Theorem \ref{Theorem:resampling-consistency}, we know that $\mathbb{P}\big( 	\tilde{\beta}_{(j)}^* \leq  \beta_{(j) }   |\{ \z_{i,n}\}_{i=1}^n   \big)$ roughly follows $\text{Unif}(0,1)$ when the sample size $n$ is large. Given a desirable tuning pair $( c_{\texttt{L}}, c_{\texttt{R}})$, we would thus expect that ${B}_{j,(1)}( c_{\texttt{L}}, c_{\texttt{R}}),...,{B}_{j,(T)}( c_{\texttt{L}}, c_{\texttt{R}})$ share a similar distribution with the ordered statistics of i.i.d. $\text{Unif}(0,1)$ random variables. The loss function defined in Eq \eqref{Eq:tuning-loss-function} measures the average of squared differences between ${B}_{j,(t)}( c_{\texttt{L}}, c_{\texttt{R}})$ and the expected value of the order statistics of the $\text{Unif}(0,1)$ random variables. Given the rational above, we would expect that the optimal tuning pair $( c_{\texttt{L}}^{j}, c_{\texttt{R}}^j)$ minimize such a loss.

 %We use ${B}_{j,t}( c_{\texttt{L}}, c_{\texttt{R}})$ as an estimate of $\mathbb{P}\big( 	\tilde{\beta}_{(j)}^* \leq  \beta_{(j)}   |\{ \z_{i,n}\}_{i=1}^n   \big)$. We know from Theorem \ref{Theorem:resampling-consistency} that $\mathbb{P}\big( 	\tilde{\beta}_{(j)}^* \leq  \beta_{(j) }   |\{ \z_{i,n}\}_{i=1}^n   \big)$ follows $\text{Unif}(0,1)$ asymptotically, so ${B}_{j,(1)}( c_{\texttt{L}}, c_{\texttt{R}}),...,{B}_{j,(T)}( c_{\texttt{L}}, c_{\texttt{R}})$ should perform like ordered statistics in $\text{Unif}(0,1)$ when we choose suitable $c_{\texttt{L}}$ and $c_{\texttt{R}}$. Thus we compare ${B}_{j,(t)}( c_{\texttt{L}}, c_{\texttt{R}})$ with $\frac{t}{T+1}$ to get the $c_L$, $c_R$ that minimizes the loss function ${L}_j( c_{\texttt{L}}, c_{\texttt{R}})$ as our final choice.

We further comment on several implementation details.  Our numerical results suggest using $R=200$ and $T=40$ can provide reasonable choice of the tuning parameters in finite samples. In addition, when the loss function ${L}_j( c_{\texttt{L}}, c_{\texttt{R}})$ do not fluctuate substantially over all considered pairs $( c_{\texttt{L}}, c_{\texttt{R}})$. In this case, let $\gamma$ denote the 97.5th percentile of ${L}_U=\frac{1}{T}\sum_{t=1}^T({U}_{(b)}-\frac{t}{T+1})^2$ and ${U}_{(1)},...,{U}_{(T)}$ are ordered observations from $\text{Unif}(0,1)$ distribution, we choose $(\bar{c}_L^j,\bar{c}_R^j)$ which is the mean of all plausible pairs such that ${L}_j( c_{\texttt{L}}, c_{\texttt{R}})<\gamma$. Lastly, as for the candidate region of $c_{\texttt{L}}$ and $c_{\texttt{R}}$, we first consider selecting $c_L$ from 0 to $\frac{2(\hat\beta_{(1)}-\hat\beta_{(j)})n^\delta}{s_{\hat{j}}^\delta}$ and $c_R$ from 0 to $\frac{2(\hat\beta_{(j)}-\hat\beta_{(d)})n^\delta}{s_{\hat{j}}\delta}$. Then based on the values of $L_j( c_{\texttt{L}}, c_{\texttt{R}})$ for different tuning pairs, we may choose to expand or shrink the candidate region to make our algorithm more efficient. 

\subsection{Robustness to different tuning choice}

We summarizing our simulation results with different choices of the tuning parameter $\delta\in\{ 0.05, 0.15, 0.25 \}$, $R\in\{200, 500\}$ and $T\in\{40, 100, 200\}$. To avoid redundancy, we showcase the results with $\bm{\beta}=0$ and $\beta_{j} = \Phi^{-1}\big( \frac{j}{d+1} \big),\ j = 1,\ldots, d$ while $q_n$ takes value $141$ or $561$. We report the coverage probabilities and asymptotic biases for estimating the top two policies (i.e., $d_0=2$) $\beta_{(1)}$ and $\beta_{(2)}$. Supplementary Materials Table \ref{table:delta} summarizes the simulation results  with different choice of $\delta$ and fixed $R = 200$ and $T=40$. There, we observe that the performance of our method is overall robust to the choice of different $\delta$ in a variety of settings. Though when $q_n$ is large and no policy is effective, smaller $\delta$ likely leads to under-covered confidence intervals. 
Supplementary Materials Table \ref{table:tuning} summarizes the simulation results under $R\in\{200, 500\}$ and $T\in\{40, 100, 200\}$, while fixing $\delta = 0.25$. Our results demonstrate that when $R$ or $T$ increases, the coverage probabilities are slightly increased and biases are marginally reduced. Overall, we observe that the proposed method is not very sensitive to the choice of various tuning parameters $\delta$, $T$, and $R$. To guide readers for the selection of tuning parameters to reach an optimal accuracy and computational efficiency trade-off, we further provide the computational time under various choices of $T$ and $R$ in the Supplementary Materials Table \ref{table:computational-time}. In practice, to reduce computational cost while maintaining valid statistical inference, we adopt the following tuning set in the rest of the numerical studies: $R=200$, $T=100$, and $\delta=0.25$.

%%%%%%%%%%%%%%% Tuning parameter investigation

	\begin{table}[h!]
\begin{center}
     \caption{Coverage probability and asymptotic bias with different choices of $\delta$ \label{table:delta}}
   \centering
\begin{adjustbox}{width=.9\textwidth,center}
\begin{tabular}{ccccccccccc}
     \hline\hline
     & \multicolumn{9}{c}{No policy is effective,  $\beta_{(1)} = \beta_{(2)}=0$ }\\\cline{3-10}
     & &\multicolumn{3}{c}{  $q_n=141$ }&  & \multicolumn{3}{c}{  $q_n=561$}\\
           \cline{3-5} \cline{7-10}
         
             &  &   $\delta$=0.05 & $\delta$=0.15 & $\delta$=0.25 & & $\delta=0.05$ & $\delta$=0.15 & $\delta$=0.25 & &   \\
       \cline{3-5} \cline{7-10}
    $\beta_{(1)}$ &Cover  & 0.97(0.00) & 0.95(0.01) & 0.96(0.01)  &   &  0.93(0.01)
    &0.92(0.01) & 0.96(0.01) \\

    & $\sqrt{n}$Bias & 0.01(0.02) &0.03(0.04) &
   -0.02(0.02) & & 0.03(0.01)
     & 0.03(0.01)& -0.01(0.01) \\[0.15cm]

   $\beta_{(2)}$ &Cover & 0.98(0.00)  & 0.94(0.01) & 0.97(0.01) & & 0.93 (0.01)& 0.93(0.01) & 0.96(0.01)\\

     & $\sqrt{n}$Bias & 0.00(0.01) & 0.02(0.01) &0.01(0.01) &  &0.03(0.01) &0.02(0.01) & 0.01(0.01) \\[0.5cm]
       & \multicolumn{9}{c}{Top two policies are effective,  $\beta_{(1)} = 0.97, \beta_{(2)} = 0.43$ }\\\cline{3-10}
     & &\multicolumn{3}{c}{  $q_n=141$ }&  & \multicolumn{3}{c}{  $q_n=561$}\\
           \cline{3-5} \cline{7-10}
         
             &  &   $\delta$=0.05 & $\delta$=0.15 & $\delta$=0.25 & & $\delta=0.05$ & $\delta$=0.15 & $\delta$=0.25 & &   \\
             $\beta_{(1)}$ &Cover  & 0.94(0.01) & 0.98(0.00) & 0.95(0.01)  &   &  0.93(0.01)
    &0.95(0.01) & 0.94(0.01) \\

    & $\sqrt{n}$Bias & 0.05(0.05) &-0.03(0.04) &
   -0.04(0.05) & & -0.07(0.06)
     & -0.06(0.07)& -0.07(0.07) \\[0.15cm]

   $\beta_{(2)}$ &Cover & 0.94(0.01)  & 0.97(0.01) & 0.95(0.01) & &  0.95(0.01)& 0.96(0.01) & 0.95(0.01)\\

     & $\sqrt{n}$Bias & 0.08(0.08) & 0.05(0.06) & 0.06(0.06) &  &0.08(0.08) &0.07(0.08) & 0.08(0.09) \\[0.15cm]
     \hline\hline
     \end{tabular}
     \end{adjustbox}
      \begin{tablenotes}\footnotesize
   \item Note: We fix $R=200$ and $T=40$. ``Cover" is the empirical coverage of the 95\% confidence interval for $\beta_{(j)}$ and `` $\sqrt{n}$Bias " captures the root-$n$ scaled Monte Carlo bias for estimating $\beta_{(j)}$. Monte Carlo standard errors are provided in the parenthesis.
     \end{tablenotes}
\end{center}
 \end{table}

 \subsection{Computational time with different tuning parameters}
 
 In this section, we summarize computational time with respect to various choices of $R$, $T$, and $n$ to make the computational costs transparent for readers. Here, we fix $\delta$ at $0.25$ and select $20$ candidate tuning parameters. Table \ref{table:computational-time} demonstrates that the computational costs are largely determined by $T$ and $R$. When both $T$ and $R$ reach $500$, running our method once takes approximately one hour. For simulation study with multiple iterations, we recommend setting $R=200$ and $T \leq 200$ to achieve a reasonable trade-off between accuracy and computational efficiency.

%%%%%%%%% T and R

	\begin{table}[h!]
\begin{center}
     \caption{Coverage probability and asymptotic bias with different choices of $T$ and $R$ \label{table:tuning}}
   \centering
\begin{adjustbox}{width=.9\textwidth,center}
\begin{tabular}{cccccccccccc}
     \hline\hline
     & \multicolumn{9}{c}{No policy is effective,  $\beta_{(1)} = \beta_{(2)}=0$ }\\\cline{4-10}
     & &\multicolumn{4}{c}{  $q_n=141$ }&  & \multicolumn{3}{c}{  $q_n=561$}\\
           \cline{4-6} \cline{8-11}
         
             &  &  & $T=40$ & $T=100$ & $T=200$ & & $T=40$ & $T=100$ & $T=200$ & &   \\
       \cline{4-6} \cline{8-11}
    $\beta_{(1)}$ & $R=200$ &Cover  & 0.96(0.01) & 0.96(0.01) & 0.96(0.01)  &   &  0.96(0.01)
    &0.96(0.02) & 0.96(0.02) \\

    &   & $\sqrt{n}$Bias & -0.02(0.02) & 0.03(0.06) &
   0.02(0.06) & & -0.01(0.01)
     & -0.01(0.12)& 0.01(0.12) \\[0.15cm]
     
        & $R=500$ &Cover  & 0.96(0.01) & 0.95(0.01) & 0.96(0.01)  &   &  0.97(0.01)
    &0.96(0.01) & 0.96(0.01) \\

    &   & $\sqrt{n}$Bias & 0.01(0.02) & -0.02(0.04) &
   0.01(0.02) & & 0.01(0.01)
     & 0.01(0.01)& 0.01(0.01) \\[0.15cm]

   $\beta_{(2)}$ & $R=200$ & Cover & 0.96(0.01)  & 0.96(0.01) & 0.97(0.01) & &  0.96(0.01)& 0.96(0.01) & 0.97(0.01)\\

     &  &$\sqrt{n}$Bias & 0.01(0.01) & -0.01(0.01) & 0.01(0.01) &  & 0.01(0.01) & 0.01(0.01) & -0.01(0.01)
     \\[0.15cm]

       & $R=500$ & Cover & 0.95(0.01)  & 0.95(0.01) & 0.96(0.01) & & 0.96(0.01)& 0.95(0.01) & 0.96(0.01)\\

      &  &$\sqrt{n}$Bias & 0.01(0.01) & 0.01(0.01) & 0.01(0.01) &  &0.01(0.01) &0.01(0.01) & 0.00(0.01)
     \\[0.5cm]

       & \multicolumn{9}{c}{Top two policies are effective,  $\beta_{(1)} = 0.97, \beta_{(2)} = 0.43$ }\\\cline{4-10}
     & &\multicolumn{4}{c}{  $q_n=141$ }&  & \multicolumn{4}{c}{  $q_n=561$}\\
           \cline{4-6} \cline{8-11}
         
             &  &  & $T=40$ & $T=100$ & $T=200$ & & $T=40$& $T=100$ & $T=200$ & &   \\
             
             $\beta_{(1)}$ & $R=200$ &Cover  & 0.95(0.01) & 0.95(0.01) & 0.96(0.01)  &   &  0.94(0.01)
    &0.95(0.01) & 0.95(0.01) \\

    & & $\sqrt{n}$Bias & -0.04(0.05) & 0.02(0.08) &
   0.02(0.08) & & -0.07(0.07)
     & 0.07(0.16)& 0.06(0.18) \\[0.15cm]
     
      & $R=500$ &Cover  & 0.96(0.01) & 0.96(0.01) & 0.96(0.01)  &   &  0.94(0.01)
    &0.95(0.01) & 0.95(0.01) \\  
    
      & & $\sqrt{n}$Bias & 0.02(0.04) & 0.01(0.04) &
   0.01(0.05) & & 0.05(0.07)
     & 0.05(0.07)& 0.04(0.07) \\[0.15cm]

   $\beta_{(2)}$ &$R=200$ &Cover & 0.95(0.01)  & 0.95(0.01) & 0.95(0.01) & &  0.95(0.01)& 0.95(0.01) & 0.96(0.01)\\

     &  & $\sqrt{n}$Bias & 0.06(0.06) & 0.05(0.06) & 0.05(0.06) &  &0.08(0.09) & 0.07(0.08) & 0.07(0.09) \\[0.15cm]

     & $R=500$ &Cover  & 0.96(0.01) & 0.96(0.01) & 0.96(0.01)  &   &  0.95(0.01)
    & 0.96(0.01) & 0.96(0.01) \\  
    
     &  & $\sqrt{n}$Bias & 0.06(0.08) & 0.05(0.06) & 0.05(0.06) &  & 0.07(0.08) & 0.07(0.08) & 0.06(0.09) \\[0.15cm]

     \hline\hline
     \end{tabular}
     \end{adjustbox}
      \begin{tablenotes}\footnotesize
   \item Note: We fix $\delta = 0.25$. ``Cover" is the empirical coverage of the 95\% confidence interval for $\beta_{(j)}$ and `` $\sqrt{n}$Bias " captures the root-$n$ scaled Monte Carlo bias for estimating $\beta_{(j)}$. Monte Carlo standard errors are provided in the parenthesis. 
     \end{tablenotes}
\end{center}
 \end{table}

 	\begin{table}[h!]
\begin{center}
   \centering
\begin{adjustbox}{width=.9\textwidth,center}
\begin{tabular}{cccccccccccc}
      & \multicolumn{9}{c}{Computational time with respect to various $n$, $q_n$, $T$, and $R$}\\
       \hline\hline

     & &\multicolumn{4}{c}{  $q_n=141$ ($s\times 10^{3}$) }&  & \multicolumn{3}{c}{  $q_n=561$ ($s\times 10^{3}$)}\\
           \cline{4-6} \cline{8-11}
         
             &  &  & $T=40$ & $T=200$ & $T=500$ & & $T=40$ & $T=200$ & $T=500$ & &   \\
       \cline{4-6} \cline{8-11}
    $n = 500$ & $R=200$ &  & 0.10 & 0.48  & 1.27  &  &0.11  & 0.48 & 1.35  \\

        & $R=500$ &  & 0.23 & 1.20 & 3.31  &   & 0.24 
    &1.38 & 3.50 \\  
    
    \cline{1-11}
   
   $n = 2000$ & $R=200$ &  & 0.10  & 0.50 & 1.29 & &0.12  & 0.51 & 1.40 \\

       & $R=500$ &  & 0.26  & 1.23  & 3.44 & &0.27 & 1.39 &3.98 \\
     
    \cline{1-11}
    
             $n = 5000$ & $R=200$ & & 0.11 & 0.52 & 1.32  &   &  0.14
    &0.55 & 1.51 \\

      & $R=500$ &  &0.27 &1.25 & 3.50  &   & 0.30 
    &1.40 & 4.05 \\  
    
    \cline{1-11}
     
%   $n = 10000$ &$R=200$ & & 6.40  &  &  & &  &  & \\

%      & $R=500$ & & ... &  &   &   &  
%     & &  \\  
    
     \hline\hline
     \end{tabular}
     \end{adjustbox}
   \caption{  \footnotesize The unit: 1,000 seconds.   We fix $\delta = 0.25$ and set $20$ candidate tuning pairs for $(c_{\texttt{L}}, c_{\texttt{R}})$.  The simulations are performed on a Lenovo NeXtScale nx360m5 node (24 cores per node) equipped with Intel Xeon Haswell processor. The core frequency is 2.3 Ghz and supports 16 floating-point operations per clock period.  \label{table:computational-time}}
\end{center}
 \end{table}

 \subsection{Simulation results: $\beta_{\hat{j}}$}
 
The simulation results presented in Table  \ref{table:sim-d5-hetero-1st-hat} and \ref{table:sim-d5-hetero-2nd-hat} help us confirm our theoretical analyses in Corollary \ref{corollary:observed-policy}, and we observe similar trends compared to the results in the main manuscript. 
 
  \begin{table}[h!]
%\begin{center}
 %\centering
\caption{ Simulation results ($d=5, \text{heterogeneity},  \beta_{\hat{1}}$)}\label{table:sim-d5-hetero-1st-hat}

%\resizebox{\columnwidth}{!}{%
\begin{adjustbox}{width=.8\textwidth,center}
\begin{tabular}{ccccccc}
     \hline\hline
           & &\multicolumn{4}{c}{  $\beta_{j} = \Phi^{-1}\big( \frac{j}{d+1} \big), \quad \bm{\gamma}_n=0, \quad j = 1,\ldots, d$} \\[0.15cm]
           \cline{3-7}
           &&\multicolumn{4}{c}{ $ \x_{i,n}\sim \mathcal{N}(0, \Sigma), \quad \w_{i,n} = \mathds{1}( \tilde{\w}_{i,n}\geq \Phi^{-1}(0.98) ) $}   \\[0.15cm]
            \cline{3-7} 
             & &  Proposed+KJ &  Proposed+HCK & Proposed+HC3 & Proposed+EW & No adjustment+KJ\\[0.15cm]
       \cline{3-7} 
   $q_n=1$ &Cover  & 0.96(0.01) & 0.96(0.01) & 0.95(0.01)  & 0.96(0.01) & 0.95(0.01)
    \\
    
 & $\sqrt{n}$Bias & 0.04(0.06)
     &0.04(0.05) &0.04(0.05) & -0.04(0.05) &0.04(0.06)
    \\[0.15cm]

   $q_n=141$ &Cover & 0.95(0.01) & 0.95(0.01) & 0.93(0.01) &  0.95(0.01)& 0.95(0.01)\\

     & $\sqrt{n}$Bias & 0.07(0.07) & 0.06(0.07) & 0.07(0.06) & 0.06(0.06) &0.05(0.07)\\[0.15cm]
     
       $q_n=281$ &Cover  & 0.94(0.01) & 0.95(0.01) & 0.84(0.01)  & 0.82(0.01) & 0.94(0.01)
    \\ 
    
    & $\sqrt{n}$Bias & -0.09(0.08)
     & -0.07(0.07)& -0.10(0.08) & -0.11(0.07) &-0.07(0.08)
    \\[0.15cm]

   $q_n=421$ &Cover & 0.94(0.01) & 0.91(0.01) &0.76(0.02) & 0.75(0.01) & 0.93(0.01)\\

     & $\sqrt{n}$Bias & -0.09(0.10) & -0.10(0.09)&-0.15(0.09) & -0.16(0.09) & -0.10(0.09)\\[0.15cm]
     
      $q_n=561$ &Cover &0.94(0.01) & 0.90(0.01) & 0.67(0.02)& 0.65(0.01) & 0.78(0.01)\\
    
     & $\sqrt{n}$Bias & -0.15(0.14) & -0.12(0.10) &-0.17(0.23) & -0.25(0.12) &0.15(0.11)\\[0.15cm]
     
      $q_n=631^*$ &Cover &0.92(0.01) &0.89(0.02) &0.45(0.02) & 0.42(0.01) & 0.68(0.01)\\
    
     & $\sqrt{n}$Bias & -0.19(0.18) & -0.22(0.13) &-0.35(0.29) & -0.54(0.22) &0.28(0.18)\\[0.15cm]
      \cline{3-7}

     & & \multicolumn{4}{c}{$\x_{i,n} = \mathds{1}(\tilde{\x}_{i,n} > 0), \quad \w_{i,n} \sim \mathcal{N}(0, I) $}    \\[0.15cm]
      \cline{3-7} 
      & &  Proposed+KJ &  Proposed+HCK  &  Proposed+HC3 & Proposed+EW & No adjustment+KJ\\[0.15cm]
       \cline{3-7} 
      $q_n=1$ &Cover  &0.97(0.01) &0.97(0.01) & 0.95(0.01) & 0.95(0.01) & 0.96(0.01)
    \\

       & $\sqrt{n}$Bias & -0.05(0.10) & -0.06(0.09)& -0.06(0.10) & -0.10(0.11) &-0.04(0.09)
       \\[0.15cm]
       
    $q_n=141$ &Cover &0.97(0.01) & 0.95(0.01)& 0.94(0.01)&  0.93(0.01) &0.95(0.01)\\

        & $\sqrt{n}$Bias &-0.06(0.11) & -0.08(0.12)& -0.07(0.11)& 0.13(0.12) &0.09(0.12) \\[0.15cm]
        
          $q_n=281$ &Cover  & 0.96(0.01) & 0.94(0.01) & 0.86(0.02) & 0.85(0.01) & 0.95(0.01)
    \\

    & $\sqrt{n}$Bias & -0.09(0.13)
     & -0.10(0.13) &-0.10(0.13) & -0.15(0.12) & -0.09(0.13)
    \\[0.15cm]

   $q_n=421$ &Cover &0.94(0.01) & 0.93(0.01) & 0.75(0.02) & 0.72(0.01) & 0.93(0.01)\\

     & $\sqrt{n}$Bias &0.11(0.17) &-0.12(0.13) & 0.18(0.17)& -0.20(0.17) &0.14(0.14)\\[0.15cm]
     
      $q_n=561$ &Cover & 0.94(0.01) &0.90(0.01) &0.51(0.02) &  0.48(0.01)& 0.92(0.01)\\

     & $\sqrt{n}$Bias &-0.15(0.22) & -0.21(0.20)&-0.25(0.23) & -0.46(0.26) &-0.21(0.20)\\[0.15cm]
     
     $q_n=631$ &Cover &0.91(0.01) &0.90(0.01) & 0.48(0.02) & 0.45(0.01)& 0.80(0.01)\\
    
     & $\sqrt{n}$Bias &-0.21(0.20) & -0.23(0.22)& 0.41(0.30)& -0.53(0.20) &0.35(0.22)\\[0.15cm]
      \cline{2-5}

     \hline\hline
     \end{tabular}
     \end{adjustbox}
      \begin{tablenotes}\footnotesize
  \item Note: ``Cover" is the empirical coverage of the 95\% confidence interval for $\beta_{\hat{1}}$ and `` $\sqrt{n}$Bias " captures the root-$n$ scaled Monte Carlo bias for estimating $\beta_{\hat{1}}$. `` * " indicates that $\hat{\bm{\Omega}}_n^{\texttt{KJ}}$ is not positive semi-definite in some Monte Carlo samples.
     \end{tablenotes}
%\end{center}
 \end{table}

  \begin{table}[h!]
%\begin{center}
 %\centering
\caption{ Simulation results ($d=5,\text{heterogeneity}, \beta_{\hat{2}}$)}\label{table:sim-d5-hetero-2nd-hat}

%\resizebox{\columnwidth}{!}{%
\begin{adjustbox}{width=.8\textwidth,center}
\begin{tabular}{ccccccc}
     \hline\hline
           & &\multicolumn{4}{c}{  $
           \bm{\beta}_j = \Phi^{-1}\big( \frac{j}{d+1} \big),  \quad \bm{\gamma}_n=0, \quad j = 1,\ldots,d$} \\[0.15cm]
           \cline{3-7}
           &&\multicolumn{4}{c}{ $ \x_{i,n}\sim \mathcal{N}(0, \Sigma), \quad \w_{i,n} = \mathds{1}( \tilde{\w}_{i,n}\geq \Phi^{-1}(0.98) ) $}   \\[0.15cm]
            \cline{3-7} 
             & &  Proposed+KJ & Proposed+HCK &Proposed+HC3  & Proposed+EW &No adjustment+KJ\\[0.15cm]
       \cline{3-7} 
   $q_n=1$ &Cover  & 0.96(0.01) & 0.97(0.01) & 0.97(0.01)  & 0.96(0.01)& 0.97(0.01)
    \\
    
 & $\sqrt{n}$Bias & -0.04(0.06)
     & -0.03(0.06) &-0.03(0.06) & -0.04(0.07)& -0.04(0.06)
    \\[0.15cm]

   $q_n=141$ &Cover & 0.96(0.01) & 0.96(0.02) & 0.88(0.02) & 0.90(0.01) &0.94(0.01) \\

     & $\sqrt{n}$Bias & -0.05(0.08) & -0.05(0.08) & -0.10(0.08) & -0.09(0.06) & -0.07(0.07)\\[0.15cm]
     
       $q_n=281$ &Cover  &0.95(0.01) & 0.94(0.02) & 0.86(0.02)  & 0.84(0.01)& 0.91(0.02)
    \\
    
    & $\sqrt{n}$Bias & 0.07(0.09)
     & 0.07(0.09)& 0.13(0.08) & -0.15(0.10) & 0.12(0.10)
    \\[0.15cm]

   $q_n=421$ &Cover & 0.94(0.01) & 0.93(0.02) &0.77(0.02) & 0.72(0.02) &0.71(0.02) \\

     & $\sqrt{n}$Bias & -0.10(0.13) & -0.12(0.13)&-0.15(0.11) & -0.17(0.13)& -0.19(0.17)\\[0.15cm]
     
      $q_n=561$ &Cover &0.94(0.01) & 0.92(0.02) & 0.65(0.02)& 0.60(0.01)& 0.69(0.02)\\
    
     & $\sqrt{n}$Bias & -0.16(0.17) & -0.18(0.16) &-0.18(0.15) & 0.20(0.13)& 0.35(0.22)\\[0.15cm]
     
      $q_n=631^*$ &Cover &0.93(0.01) &0.92(0.02) &0.44(0.02) & 0.42(0.01)& 0.50(0.02) \\
    
     & $\sqrt{n}$Bias & -0.18(0.17) & -0.23(0.21) &-0.45(0.19) & -0.49(0.17) & 0.48(0.30)\\[0.15cm]
      \cline{3-7}

     & & \multicolumn{4}{c}{$\x_{i,n} = \mathds{1}(\tilde{\x}_{i,n} > 0), \quad \w_{i,n} \sim \mathcal{N}(0, I) $}    \\[0.15cm]
      \cline{3-7} 
      & &  Proposed+KJ &  Proposed+HCK  & Proposed+HC3 & Proposed+EW & No adjustment+KJ\\[0.15cm]
       \cline{3-7} 
     $q_n=1$ &Cover  & 0.96(0.01) &0.98(0.01) & 0.96(0.01) & 0.96(0.01)& 0.97(0.01)
    \\

       & $\sqrt{n}$Bias & -0.08(0.13) & -0.07(0.13)& -0.10(0.12) & 0.09(0.12) & -0.10(0.10)
       \\[0.15cm]
       
    $q_n=141$ &Cover &0.96(0.01) & 0.96(0.01)& 0.94(0.01)& 0.95(0.01) &0.95(0.01) \\

        & $\sqrt{n}$Bias &0.09(0.14) & 0.10(0.14)& 0.13(0.14)& 0.12(0.13) & 0.10(0.14) \\[0.15cm]
        
          $q_n=281$ &Cover  & 0.95(0.01) & 0.95(0.01) & 0.92(0.01) &0.90(0.01) & 0.95(0.01)
    \\

    & $\sqrt{n}$Bias & -0.11(0.16)
     & -0.11(0.15) &0.17(0.15) & 0.19(0.14)& 0.13(0.17)
    \\[0.15cm]

   $q_n=421$ &Cover &0.95(0.01) & 0.92(0.02) & 0.82(0.02) & 0.78(0.01) &0.94(0.01) \\

     & $\sqrt{n}$Bias &0.18(0.20) &-0.20(0.20) & -0.25(0.20)& -0.32(0.17) & -0.20(0.20)\\[0.15cm]
     
      $q_n=561$ &Cover & 0.94(0.01) &0.92(0.02) &0.67(0.02) & 0.64(0.01)& 0.73(0.01) \\

     & $\sqrt{n}$Bias &-0.26(0.24) & 0.22(0.24)&-0.41(0.23) & -0.48(0.23)& 0.27(0.25)\\[0.15cm]
     
     $q_n=631$ &Cover &0.92(0.01) & 0.87(0.02) & 0.52(0.02) &0.48(0.01) & 0.50(0.01)\\
    
     & $\sqrt{n}$Bias &0.30(0.27) & 0.41(0.29)& 0.50(0.28)&  -0.65(0.13)& 0.55(0.28)\\[0.15cm]
      \cline{2-5}

     \hline\hline
     \end{tabular}
     \end{adjustbox}
      \begin{tablenotes}\footnotesize
  \item Note: ``Cover" is the empirical coverage of the 95\% confidence interval for $\beta_{\hat{2}}$ and `` $\sqrt{n}$Bias " captures the root-$n$ scaled Monte Carlo bias for estimating $\beta_{\hat{2}}$. `` * " indicates that $\hat{\bm{\Omega}}_n^{\texttt{KJ}}$ is not positive semi-definite in some Monte Carlo samples.
     \end{tablenotes}
%\end{center}
 \end{table}

\subsection{Simulation results: $d=10$}
This section provides an additional set of simulation results when $d=10$, which is larger than the setting ($d=5$) adopted in the main manuscript.  We investigate the performance of the five methods for estimating $\beta_{(2)}$, $\beta_{(5)}$, and $\beta_{(10)}$. Table \ref{table:sim-d10} - \ref{table:sim-d10-10th} demonstrate that 
without adjustment, the coverage probabilities for $\beta_{(5)}$ fall below $80\%$ when $q_n\geq 281$, while our proposed method reaches nominal level coverage regardless the ranking of $\beta_j$. ``Proposed + EW", ``Proposed + HCK", and ``Proposed + HC3" show similar trends compared to those when $d = 5$. 

 \begin{table}[h!]
\caption{ Simulation results ($d=10, \text{heterogeneity}, \beta_{(2)}$)}\label{table:sim-d10}

%\resizebox{\columnwidth}{!}{%
\begin{adjustbox}{width=.8\textwidth,center}
\begin{tabular}{ccccccc}
     \hline\hline
           & &\multicolumn{4}{c}{  $\beta_{j} = \Phi^{-1}\big( \frac{j}{d+1} \big), \quad \bm{\gamma}_n=0$, \quad j = 1, \ldots, d} \\[0.15cm]
           \cline{3-7}
           &&\multicolumn{4}{c}{ $ \x_{i,n}\sim \mathcal{N}(0, \Sigma), \quad \w_{i,n} = \mathds{1}( \tilde{\w}_{i,n}\geq \Phi^{-1}(0.98) ) $}   \\[0.15cm]
            \cline{3-7} 
             & &  Proposed+KJ &  Proposed+HCK  &  Proposed+HC3 & Proposed+EW &No adjustment+KJ\\[0.15cm]
       \cline{3-7} 
   $q_n=1$ &Cover  & 0.96(0.01) & 0.96(0.01) & 0.96(0.01) & 0.97(0.01)  & 0.94(0.01)
    \\
    
 & $\sqrt{n}$Bias & -0.05(0.07)
     &-0.04(0.07) & -0.05(0.06) &-0.06(0.07) &0.06(0.06)
    \\[0.15cm]

   $q_n=141$ &Cover & 0.96(0.01) & 0.95(0.01) &  0.96(0.01)& 0.96(0.01) & 0.94(0.01) \\

     & $\sqrt{n}$Bias & -0.06(0.07) & -0.07(0.07) & 0.07(0.08) & -0.07(0.07) &-0.06(0.06)\\[0.15cm]
     
       $q_n=281$ &Cover  &0.95(0.01) & 0.95(0.01) & 0.92(0.01)& 0.89(0.01)  & 0.92(0.01)
    \\
    
    & $\sqrt{n}$Bias & -0.07(0.09)
     & -0.08(0.09)& 0.15(0.10) & -0.18(0.08) &0.20(0.08)
    \\[0.15cm]

   $q_n=421$ &Cover & 0.96(0.01) & 0.94(0.01) & 0.88(0.01) &0.80(0.02) &0.90(0.01)\\

     & $\sqrt{n}$Bias & -0.08(0.11) & -0.10(0.11)& 0.20(0.11) &-0.25(0.10) &-0.25(0.09)\\[0.15cm]
     
      $q_n=561$ &Cover &0.95(0.01) & 0.94(0.01) & 0.75(0.01) & 0.67(0.02)&0.89(0.01) \\
    
     & $\sqrt{n}$Bias & 0.14(0.15) & 0.15(0.15) & -0.27(0.16) &-0.31(0.13) &-0.28(0.12)\\[0.15cm]
     
      $q_n=631^*$ &Cover &0.93(0.01) & 0.91(0.01) &  0.65(0.01)& 0.57(0.02) &0.88(0.01) \\
    
     & $\sqrt{n}$Bias & 0.22(0.20) & 0.27(0.19) & -0.33(0.20) & 0.35(0.17) &0.38(0.16)\\[0.15cm]
      \cline{3-7}

     & & \multicolumn{4}{c}{$\x_{i,n} = \mathds{1}(\tilde{\x}_{i,n} > 0), \quad \w_{i,n} \sim \mathcal{N}(0, I) $}    \\[0.15cm]
      \cline{3-7} 
      & &  Proposed+KJ & Proposed+HCK  & Proposed+HC3 &   Proposed+EW & No adjustment+KJ\\[0.15cm]
       \cline{3-7} 
       
      $q_n=1$ &Cover  &0.97(0.01) &0.96(0.01) & 0.96(0.01) & 0.96(0.01) &0.95(0.01)
    \\

       & $\sqrt{n}$Bias & -0.05(0.11) & -0.07(0.11)&  0.06(0.12) & -0.08(0.11) &-0.09(0.09)
       \\[0.15cm]
       
    $q_n=141$ &Cover &0.96(0.01) & 0.95(0.01)&  0.94(0.01)& 0.93(0.01)&0.94(0.01) \\

        & $\sqrt{n}$Bias &0.07(0.12) & 0.10(0.13)&  0.13(0.13) & 0.15(0.13)&0.10(0.11)\\[0.15cm]
        
          $q_n=281$ &Cover  & 0.95(0.01) & 0.95(0.01) & 0.91(0.01)& 0.88(0.01) & 0.93(0.01)
    \\

    & $\sqrt{n}$Bias & 0.14(0.15)
     & 0.16(0.16) & 0.20(0.13) & -0.28(0.15) &-0.29(0.13)
    \\[0.15cm]

   $q_n=421$ &Cover &0.94(0.01) & 0.93 (0.01) &  0.85(0.01) & 0.76(0.01) &0.92(0.01)\\

     & $\sqrt{n}$Bias & 0.20(0.19) & 0.22(0.18) & 0.25(0.17) & 0.37(0.18)&0.32(0.16)\\[0.15cm]
     
      $q_n=561$ &Cover & 0.93(0.01) &0.93(0.01) & 0.77(0.01) &0.67(0.02) &0.90(0.01)\\

     & $\sqrt{n}$Bias & 0.22(0.23) & 0.25(0.23)& -0.30(0.19) &0.41(0.22) &0.44(0.19)\\[0.15cm]
     
     $q_n=631$ &Cover &0.90(0.01) & 0.88(0.01) &   0.67(0.01) & 0.60(0.02) &0.88(0.01)\\
    
     & $\sqrt{n}$Bias &0.31(0.25) & 0.45(0.26)& -0.48(0.13) & 0.55(0.25)&0.61(0.23)\\[0.15cm]
      \cline{2-5}

     \hline\hline
     \end{tabular}
     \end{adjustbox}
      \begin{tablenotes}\footnotesize
  \item Note: ``Cover" is the empirical coverage of the 95\% confidence interval for $\beta_{(2)}$ and `` $\sqrt{n}$Bias " captures the root-$n$ scaled Monte Carlo bias for estimating $\beta_{(2)}$. ``$^*$" indicates that $\hat{\bm{\Omega}}_n^{\texttt{KJ}}$ is not positive semi-definite in some Monte Carlo samples.
     \end{tablenotes}
%\end{center}
 \end{table}

 %% 5th - d10
 
  \begin{table}[h!]
\caption{ Simulation results ($d=10, \text{homogeneity}, \beta_{(5)}$)}\label{table:sim-d10-5th-homo}

%\resizebox{\columnwidth}{!}{%
\begin{adjustbox}{width=.8\textwidth,center}
\begin{tabular}{ccccccc}
     \hline\hline
           & &\multicolumn{4}{c}{  $\bm{\beta}=0, \quad \bm{\gamma}_j=1/j$, \quad $j = 1,\ldots, q_n$} \\[0.15cm]
           \cline{3-7}
           &&\multicolumn{4}{c}{ $ \x_{i,n}\sim \mathcal{N}(0, \Sigma), \quad \w_{i,n} = \mathds{1}( \tilde{\w}_{i,n}\geq \Phi^{-1}(0.98) ) $}   \\[0.15cm]
            \cline{3-7} 
             & &  Proposed+KJ  & Proposed+HCK  & Proposed+HC3 &Proposed+EW & No adjustment+KJ\\[0.15cm]
       \cline{3-7} 
    $q_n=1$ &Cover  & 0.95(0.01) & 0.95(0.01) & 0.95(0.01) & 0.93(0.01)  &0.92(0.01)
    \\
    
 & $\sqrt{n}$Bias & 0.01(0.01)
     &0.01(0.01) & 0.01(0.01) &0.01(0.01) & 0.04(0.02)
    \\[0.15cm]

   $q_n=141$ &Cover & 0.94(0.01) & 0.94(0.01) & 0.94(0.01) & 0.93(0.01) &0.82(0.01) \\

     & $\sqrt{n}$Bias & -0.01(0.01) & -0.01(0.01) & 0.02(0.02) & 0.02(0.01) &-0.17(0.02)\\[0.15cm]
     
       $q_n=281$ &Cover  &0.94(0.01) & 0.94(0.01) & 0.94(0.01) & 0.92(0.01)  &0.80(0.01)
    \\  
    
    & $\sqrt{n}$Bias & 0.02(0.02)
     & 0.02(0.02)& 0.02(0.02) & 0.05(0.02) &-0.18(0.02)
    \\[0.15cm]

   $q_n=421$ &Cover & 0.93(0.02) & 0.93(0.02) & 0.93(0.01) &0.80(0.02) &0.78(0.01)\\

     & $\sqrt{n}$Bias & -0.02(0.02) & -0.02(0.02)& 0.05(0.04) &0.17(0.03) &0.21(0.03)\\[0.15cm]
     
      $q_n=561$ &Cover &0.93(0.02) & 0.92(0.02) & 0.93(0.01) & 0.75(0.02)&0.76(0.01) \\
    
     & $\sqrt{n}$Bias & 0.02(0.02) & -0.03(0.02) & 0.05(0.04) & 0.32(0.04) &-0.28(0.04)\\[0.15cm]
     
      $q_n=631^*$ &Cover &0.91(0.02) & 0.90(0.01) & 0.90(0.01) & 0.73(0.02) &0.75(0.01) \\
    
     & $\sqrt{n}$Bias & -0.12(0.05) & -0.30(0.05) & 0.14(0.09) & 0.36(0.06) &-0.34(0.05)\\[0.15cm]
      \cline{3-7}

     & & \multicolumn{4}{c}{$\x_{i,n} = \mathds{1}(\tilde{\x}_{i,n} > 0), \quad \w_{i,n} \sim \mathcal{N}(0, I) $}    \\[0.15cm]
      \cline{3-7} 
      & &  Proposed+KJ & Proposed+HCK  &  Proposed+HC3 & Proposed + EW &No adjustment+KJ\\[0.15cm]
       \cline{3-7} 
     $q_n=1$ &Cover  &0.95(0.01) &0.95(0.02) & 0.94(0.01) & 0.93(0.02) &0.91(0.01)
    \\

       & $\sqrt{n}$Bias & 0.01(0.01) & 0.01(0.01)& -0.02(0.02)
       & 0.01(0.01) &0.06(0.03)\\[0.15cm]
       
    $q_n=141$ &Cover &0.94(0.01) & 0.94(0.01)& 0.94(0.01) & 0.93(0.01)&0.81(0.01) \\

        & $\sqrt{n}$Bias & -0.01(0.01) & -0.01(0.01)&  -0.02(0.02) & 0.01(0.00)&0.19(0.03)\\[0.15cm]
        
          $q_n=281$ &Cover  & 0.94(0.01) & 0.94(0.01) & 0.94(0.01) & 0.92(0.01) &0.78(0.01)
    \\

    & $\sqrt{n}$Bias & 0.02(0.02)
     & 0.02(0.02) & -0.02(0.03) & 0.06(0.02) &-0.21(0.04)
    \\[0.15cm]

   $q_n=421$ &Cover &0.93(0.01) & 0.93 (0.01) & 0.92(0.02) & 0.80(0.01) &0.75(0.01)\\

     & $\sqrt{n}$Bias & -0.01(0.00) & -0.01(0.00) & -0.04(0.03) & 0.17(0.03)&0.37(0.05)\\[0.15cm]
     
      $q_n=561$ &Cover & 0.93(0.01) &0.93(0.01) & 0.92(0.01) &0.75(0.02) &0.70(0.01)\\

     & $\sqrt{n}$Bias & 0.03(0.02) & 0.03(0.02)& -0.08(0.06) & 0.18(0.04) &-0.48(0.05)\\[0.15cm]
     
     $q_n=631$ &Cover &0.92(0.01) & 0.92(0.01) &  0.91(0.01) & 0.75(0.02) &0.66(0.01)\\
    
     & $\sqrt{n}$Bias & -0.12(0.05) & -0.30(0.05)& -0.32(0.11) & 0.36(0.06)&-0.65(0.09)\\[0.15cm]
      \cline{2-5}

     \hline\hline
     \end{tabular}
     \end{adjustbox}
      \begin{tablenotes}\footnotesize
  \item Note: ``Cover" is the empirical coverage of the 95\% confidence interval for $\beta_{(2)}$ and `` $\sqrt{n}$Bias " captures the root-$n$ scaled Monte Carlo bias for estimating $\beta_{(2)}$. ``$^*$" indicates that $\hat{\bm{\Omega}}_n^{\texttt{KJ}}$ is not positive semi-definite in some Monte Carlo samples.
     \end{tablenotes}
%\end{center}
 \end{table}

  \begin{table}[h!]
\caption{ Simulation results ($d=10, \text{heterogeneity}, \beta_{(5)}$)}\label{table:sim-d10-5th-hetero}

%\resizebox{\columnwidth}{!}{%
\begin{adjustbox}{width=.8\textwidth,center}
\begin{tabular}{ccccccc}
     \hline\hline
           & &\multicolumn{4}{c}{  $\beta_{j} = \Phi^{-1}\big( \frac{j}{d+1} \big), \quad \bm{\gamma}_n=0$, \quad j = 1, \ldots, d} \\[0.15cm]
           \cline{3-7}
           &&\multicolumn{4}{c}{ $ \x_{i,n}\sim \mathcal{N}(0, \Sigma), \quad \w_{i,n} = \mathds{1}( \tilde{\w}_{i,n}\geq \Phi^{-1}(0.98) ) $}   \\[0.15cm]
            \cline{3-7} 
             & &  Proposed+KJ  & Proposed+HCK  & Proposed+HC3  & Proposed+EW &No adjustment+KJ\\[0.15cm]
       \cline{3-7} 
    $q_n=1$ &Cover  & 0.96(0.01) & 0.95(0.01) & 0.92(0.02) & 0.96(0.01)  &0.94(0.01)
    \\
    
 & $\sqrt{n}$Bias & 0.06(0.07)
     &0.07(0.07) & -0.12(0.09) &0.07(0.07) & 0.06(0.06)
    \\[0.15cm]

   $q_n=141$ &Cover & 0.95(0.01) & 0.95(0.01) & 0.92(0.01) & 0.92(0.01) & 0.94(0.01) \\

     & $\sqrt{n}$Bias & 0.08(0.08) & 0.08(0.08) & -0.12(0.08) & 0.12(0.08) & -0.07(0.07)\\[0.15cm]
     
       $q_n=281$ &Cover  &0.95(0.01) & 0.95(0.01) & 0.91(0.01) & 0.91(0.01)  & 0.93(0.01)
    \\  
    
    & $\sqrt{n}$Bias & 0.08(0.08)
     & 0.08(0.08)& -0.13(0.09) & 0.12(0.08) & 0.10(0.07)
    \\[0.15cm]

   $q_n=421$ &Cover & 0.95(0.01) & 0.95(0.01) & 0.84(0.02) & 0.82(0.02) & 0.93(0.01)\\

     & $\sqrt{n}$Bias & 0.10(0.11) & 0.11(0.11)& -0.13(0.07)  & 0.14(0.11) & -0.12(0.09)\\[0.15cm]
     
      $q_n=561$ &Cover & 0.94(0.01) & 0.94(0.01) & 0.80(0.01) & 0.79(0.02)& 0.92(0.01) \\
    
     & $\sqrt{n}$Bias & 0.12(0.13) & 0.13(0.13) & -0.18(0.06) & 0.20(0.12) & -0.15(0.10)\\[0.15cm]
     
      $q_n=631^*$ &Cover & 0.93(0.01) & 0.92(0.01) & 0.78(0.01) & 0.76(0.01) & 0.90(0.01) \\
    
     & $\sqrt{n}$Bias & 0.16(0.14) & 0.18(0.14) & -0.28(0.26)  & 0.24(0.13) & 0.18(0.12)\\[0.15cm]

      \cline{3-7}

     & & \multicolumn{4}{c}{$\x_{i,n} = \mathds{1}(\tilde{\x}_{i,n} > 0), \quad \w_{i,n} \sim \mathcal{N}(0, I) $}    \\[0.15cm]
      \cline{3-7} 
      & &  Proposed+KJ & Proposed+HCK  & Proposed+HC3 & Proposed+EW & No adjustment+KJ\\[0.15cm]
       \cline{3-7} 
     $q_n=1$ &Cover  & 0.96(0.01) & 0.96(0.01) & 0.95(0.01) & 0.96(0.01) & 0.96(0.01)
    \\

       & $\sqrt{n}$Bias & 0.05(0.11) & 0.06(0.11)& -0.08(0.11)  & 0.08(0.11) & -0.06(0.09)
       \\[0.15cm]
       
    $q_n=141$ &Cover & 0.96(0.01) & 0.96(0.01)& 0.95(0.01) & 0.95(0.01)& 0.95(0.01)  \\

        & $\sqrt{n}$Bias & 0.06(0.11) & 0.08(0.11)& -0.11(0.12)  & 0.10(0.11)& 0.09(0.10) \\[0.15cm]
        
          $q_n=281$ &Cover  & 0.95(0.01) & 0.95(0.01) & 0.94(0.01) & 0.93(0.01) & 0.95(0.01)
    \\

    & $\sqrt{n}$Bias & 0.12(0.13)
     & 0.12(0.13) &  0.10(0.10) & 0.14(0.12) & -0.10(0.11)
    \\[0.15cm]

   $q_n=421$ &Cover & 0.95(0.01) &  0.95(0.01) & 0.91(0.01) & 0.90(0.01) & 0.95(0.01)\\

     & $\sqrt{n}$Bias & 0.11(0.14) & 0.15(0.15) & -0.20(0.08) & 0.25(0.14)& -0.12(0.12)\\[0.15cm]
     
      $q_n=561$ &Cover & 0.94(0.01) & 0.94(0.01) & 0.85(0.01) & 0.83(0.02) & 0.94(0.01)\\

     & $\sqrt{n}$Bias & -0.15(0.18) & -0.16(0.17)& -0.26(0.20)  & 0.28(0.18) & -0.15(0.15)\\[0.15cm]
     
     $q_n=631$ &Cover & 0.93(0.01) & 0.92(0.01) & 0.75(0.01)  & 0.72(0.02) & 0.93(0.01)\\
    
     & $\sqrt{n}$Bias & 0.25(0.20) & 0.30(0.19)&  -0.70(0.34) & 0.68(0.20)& 0.28(0.18)\\[0.15cm]
      \cline{2-5}

     \hline\hline
     \end{tabular}
     \end{adjustbox}
      \begin{tablenotes}\footnotesize
  \item Note: ``Cover" is the empirical coverage of the 95\% confidence interval for $\beta_{(2)}$ and `` $\sqrt{n}$Bias " captures the root-$n$ scaled Monte Carlo bias for estimating $\beta_{(2)}$. ``$^*$" indicates that $\hat{\bm{\Omega}}_n^{\texttt{KJ}}$ is not positive semi-definite in some Monte Carlo samples.
     \end{tablenotes}
%\end{center}
 \end{table}

 \begin{table}[h!]
\caption{ Simulation results ($d=10, \text{heterogeneity}, \beta_{(10)}$)}\label{table:sim-d10-10th}

%\resizebox{\columnwidth}{!}{%
\begin{adjustbox}{width=.8\textwidth,center}
\begin{tabular}{ccccccc}
     \hline\hline
           & &\multicolumn{4}{c}{  $\beta_{j} = \Phi^{-1}\big( \frac{j}{d+1} \big),\quad \gamma_n =0$, \quad j = 1, \ldots, d} \\[0.15cm]
           \cline{3-7}
           &&\multicolumn{4}{c}{ $\quad \x_{i,n}\sim \mathcal{N}(0, \Sigma), \quad \w_{i,n} = \mathds{1}( \tilde{\w}_{i,n}\geq \Phi^{-1}(0.98) )$}   \\[0.15cm]
            \cline{3-7} 
             & &  Proposed+KJ  & Proposed+HCK  & Proposed+HC3 & Proposed+EW &No adjustment+KJ\\[0.15cm]
       \cline{3-7} 
    $q_n=1$ &Cover  & 0.94(0.01) & 0.95(0.01) & 0.94(0.01) & 0.94(0.01)  & 0.94(0.01)
    \\
    
 & $\sqrt{n}$Bias & 0.06(0.07)
     &0.03(0.06) & 0.06(0.06) &0.07(0.07) & -0.04(0.05)
    \\[0.15cm]

   $q_n=141$ &Cover & 0.93(0.02) & 0.93(0.02) & 0.92(0.02) & 0.92(0.02) & 0.94(0.01) \\

     & $\sqrt{n}$Bias & 0.06(0.06) & 0.06(0.06) & 0.08(0.07) & 0.08(0.06) & 0.06(0.06)\\[0.15cm]
     
       $q_n=281$ &Cover  &0.94(0.01) & 0.94(0.01) & 0.90(0.01) & 0.89(0.01)  & 0.94(0.01)
    \\  
    
    & $\sqrt{n}$Bias & 0.07(0.08)
     & 0.08(0.08)& 0.11(0.08)  & 0.10(0.07) & 0.07(0.07)
    \\[0.15cm]

   $q_n=421$ &Cover & 0.94(0.01) & 0.94(0.01) & 0.84(0.01) & 0.82(0.02) & 0.93(0.01)\\

     & $\sqrt{n}$Bias & 0.08(0.09) & 0.09(0.09)& 0.16(0.10) & 0.17(0.09) & -0.11(0.08)\\[0.15cm]
     
      $q_n=561$ &Cover & 0.93(0.01) & 0.91(0.01) & 0.68(0.02) & 0.61(0.02)& 0.91(0.01) \\
    
     & $\sqrt{n}$Bias & 0.16(0.14) & 0.18(0.14) & 0.20(0.10)  & 0.24(0.13) & 0.15(0.12)\\[0.15cm]
     
      $q_n=631^*$ &Cover & 0.92(0.01) & 0.90(0.01) & 0.53(0.02) & 0.50(0.02) & 0.82(0.01) \\
    
     & $\sqrt{n}$Bias & 0.20(0.18) & 0.23(0.18) & -0.50(0.12) & 0.55(0.18) & -0.30(0.17)\\[0.15cm]
      \cline{3-7}

     & & \multicolumn{4}{c}{$\x_{i,n} = \mathds{1}(\tilde{\x}_{i,n} > 0), \quad \w_{i,n} \sim \mathcal{N}(0, I)$}    \\[0.15cm]
      \cline{3-7} 
      & &  Proposed+KJ & Proposed+HCK  &   Proposed+HC3 & Proposed+EW &No adjustment+KJ\\[0.15cm]
       \cline{3-7} 
     $q_n=1$ &Cover  & 0.96(0.01) & 0.96(0.01) &  0.95(0.01) & 0.96(0.01) & 0.96(0.01)
    \\

       & $\sqrt{n}$Bias & 0.08(0.11) & 0.08(0.11)& 0.11(0.12) & 0.09(0.11) & -0.08(0.09)
       \\[0.15cm]
       
    $q_n=141$ &Cover & 0.94(0.01) & 0.94(0.01)& 0.93(0.01) & 0.92(0.01)& 0.94(0.01) \\

        & $\sqrt{n}$Bias & 0.10(0.12) & 0.11(0.12)& 0.14(0.13) & 0.15(0.12)&  0.11(0.11)\\[0.15cm]
        
          $q_n=281$ &Cover  & 0.94(0.01) & 0.94(0.01) & 0.88(0.01) & 0.84(0.01) & 0.94(0.01)
    \\

    & $\sqrt{n}$Bias & 0.11(0.15)
     & 0.12(0.15) & 0.16(0.14) & 0.19(0.15) & 0.12(0.13)
    \\[0.15cm]

   $q_n=421$ &Cover &0.94(0.01) & 0.94(0.01) & 0.85(0.01)  & 0.82(0.01) & 0.94(0.01)\\

     & $\sqrt{n}$Bias & 0.14(0.16) & 0.16(0.16) & 0.20(0.16) & 0.25(0.16)& 0.15(0.15)\\[0.15cm]
     
      $q_n=561$ &Cover & 0.93(0.02) & 0.93(0.01) & 0.65(0.01) & 0.62(0.02) & 0.92(0.01)\\

     & $\sqrt{n}$Bias & 0.19(0.20) & 0.23(0.21)& -0.24(0.15) & 0.30(0.21) &  -0.24(0.20)\\[0.15cm]
     
     $q_n=631$ &Cover & 0.94(0.01) & 0.92(0.01) & 0.57(0.01)  & 0.50(0.02) & 0.91(0.01)\\
    
     & $\sqrt{n}$Bias & 0.26(0.29) & 0.30(0.28)&  -0.85(0.23) & 0.88(0.31)& -0.33(0.29)\\[0.15cm]
      \cline{2-5}

     \hline\hline
     \end{tabular}
     \end{adjustbox}
      \begin{tablenotes}\footnotesize
  \item Note: ``Cover" is the empirical coverage of the 95\% confidence interval for $\beta_{(10)}$ and `` $\sqrt{n}$Bias " captures the root-$n$ scaled Monte Carlo bias for estimating $\beta_{(10)}$. ``$^*$" indicates that $\hat{\bm{\Omega}}_n^{\texttt{KJ}}$ is not positive semi-definite in some Monte Carlo samples.
     \end{tablenotes}
%\end{center}
 \end{table}

 %%%%%% heteroscedasticity error
 %%%% case2. case3.
  \begin{table}[h!]
\caption{ Simulation results ($d=5, \text{ heteroscedasticity}, \text{ heterogeneity},  \beta_{(1)}$)}\label{table:sim-hetero-error-1}

%\resizebox{\columnwidth}{!}{%
\begin{adjustbox}{width=.8\textwidth,center}
\begin{tabular}{ccccccc}
     \hline\hline
           & &\multicolumn{5}{c}{  $\beta_{j} = \Phi^{-1}\big( \frac{j}{d+1} \big), \quad \x_{i,n}\sim \mathcal{N}(0, \Sigma), \quad \w_{i,n} = \mathds{1}( \tilde{\w}_{i,n}\geq \Phi^{-1}(0.98) )$, \quad j = 1, \ldots, d} \\[0.15cm]
           \cline{3-7}
           &&\multicolumn{5}{c}{ $ \bm{\gamma}_n=0 $}   \\[0.15cm]
            \cline{3-7} 
             & &  Proposed+KJ  & Proposed+HCK  & Proposed+HC3 & Proposed+EW &No adjustment+KJ\\[0.15cm]
       \cline{3-7} 
    $q_n=1$ &Cover  & 0.96(0.01) & 0.96(0.01) & 0.96(0.01) & 0.95(0.01)  & 0.95(0.01)
    \\
    
 & $\sqrt{n}$Bias & 0.06(0.07)
     &0.07(0.07) & 0.07(0.07) &0.07(0.07) & 0.07(0.07)
    \\[0.15cm]

   $q_n=141$ &Cover & 0.95(0.01) & 0.95(0.01) & 0.94(0.01) & 0.88(0.02) & 0.94(0.01) \\

     & $\sqrt{n}$Bias & -0.12(0.12) & -0.12(0.12) & -0.12(0.12) & -0.16(0.07) & -0.12(0.12)\\[0.15cm]
     
       $q_n=281$ &Cover  &0.95(0.01) & 0.93(0.02) & 0.92(0.01) & 0.79(0.01)  & 0.94(0.01)
    \\  
    
    & $\sqrt{n}$Bias & -0.11(0.13)
     & -0.13(0.13)& -0.33(0.13)  & -0.35(0.12) & -0.12(0.12)
    \\[0.15cm]

   $q_n=421$ &Cover & 0.94(0.01) & 0.93(0.01) & 0.89(0.01) & 0.73(0.02) & 0.93(0.01)\\

     & $\sqrt{n}$Bias & -0.14(0.15) & -0.18(0.15)& -0.44(0.15) & -0.52(0.15) & -0.19(0.15)\\[0.15cm]
     
      $q_n=561$ &Cover & 0.93(0.01) & 0.91(0.01) & 0.82(0.02) & 0.66(0.02)& 0.91(0.01) \\
    
     & $\sqrt{n}$Bias & -0.22(0.19) & -0.32(0.19) & -0.67(0.18)  & -0.71(0.19) & -0.33(0.18)\\[0.15cm]
     
      $q_n=631^*$ &Cover & 0.92(0.01) & 0.90(0.01) & 0.76(0.02) & 0.56(0.02) & 0.90(0.01) \\
    
     & $\sqrt{n}$Bias & -0.34(0.24) & -0.50(0.26) & -0.73(0.30) & -0.80(0.24) & -0.46(0.24)\\[0.15cm]
      \cline{3-7}

     & & \multicolumn{5}{c}{$\gamma_k = 1/k, \quad k= 1, \ldots, q_n$}    \\[0.15cm]
      \cline{3-7} 
      & &  Proposed+KJ & Proposed+HCK  &   Proposed+HC3 & Proposed+EW &No adjustment+KJ\\[0.15cm]
       \cline{3-7} 
     $q_n=1$ &Cover  & 0.95(0.01) & 0.95(0.01) &  0.95(0.01) & 0.94(0.01) & 0.95(0.01)
    \\

       & $\sqrt{n}$Bias & 0.08(0.09) & 0.08(0.09)& 0.09(0.09) & 0.09(0.09) & 0.09(0.09)
       \\[0.15cm]
       
    $q_n=141$ &Cover & 0.95(0.01) & 0.94(0.01)& 0.92(0.01) & 0.87(0.01)& 0.94(0.01) \\

        & $\sqrt{n}$Bias & 0.12(0.16) & 0.14(0.16)& 0.17(0.14) & 0.36(0.14)&  0.14(0.14)\\[0.15cm]
        
          $q_n=281$ &Cover  & 0.94(0.01) & 0.94(0.01) & 0.88(0.01) & 0.82(0.01) & 0.93(0.01)
    \\

    & $\sqrt{n}$Bias & -0.13(0.14)
     & 0.14(0.14) & -0.33(0.11) & -0.45(0.13) & 0.15(0.12)
    \\[0.15cm]

   $q_n=421$ &Cover &0.93(0.01) & 0.92(0.01) & 0.80(0.01)  & 0.73(0.01) & 0.92(0.01)\\

     & $\sqrt{n}$Bias & -0.21(0.18) & -0.24(0.18) & -0.40(0.11) & -0.56(0.11)& -0.24(0.13)\\[0.15cm]
     
      $q_n=561^*$ &Cover & 0.92(0.01) & 0.91(0.01) & 0.67(0.01) & 0.53(0.02) & 0.91(0.01)\\

     & $\sqrt{n}$Bias & -0.28(0.22) & -0.35(0.21)& -0.47(0.17) & -0.51(0.19) &  -0.37(0.20)\\[0.15cm]
     
     $q_n=631^*$ &Cover & 0.91(0.01) & 0.89(0.01) & 0.59(0.01)  & 0.51(0.02) & 0.88(0.01)\\
    
     & $\sqrt{n}$Bias & -0.29(0.25) & -0.33(0.26)&  -0.55(0.21) & -0.61(0.18)& -0.42 (0.22)\\[0.15cm]
      \cline{2-5}

     \hline\hline
     \end{tabular}
     \end{adjustbox}
      \begin{tablenotes}\footnotesize
  \item Note: ``Cover" is the empirical coverage of the 95\% confidence interval for $\beta_{(1)}$ and `` $\sqrt{n}$Bias " captures the root-$n$ scaled Monte Carlo bias for estimating $\beta_{(1)}$. In this heteroscedastic design, $\tilde{\w}_{i,n} \sim \mathcal{N}(0, \bm{I}_{q_n})$, $\varepsilon_{i,n}\sim \mathcal{N}(0,1)$,  $\mathbb{V}[\varepsilon_{i,n}\lvert \x_{i,n},\w_{i,n}]=c_{\varepsilon}(1+(t(x_{1,i,n})+\bm{l}'\w_{i,n})^2/4)$, and $\mathbb{V}[x_{k,i,n}\lvert\w_{i,n}]=c_{x_k}(1+(\bm{l}'\w_{i,n})^2/4)$, where $x_{k,i,n}$ denotes the $k$th component of the vector $\x_{i,n}$. The constants $c_{\varepsilon}$ and $c_{x_k}$ are chosen so that $\mathbb{V}[\varepsilon_{i,n}]=\mathbb{V}[x_{k,i,n}]=1$ and $t(a)=a\mathds{1}(-1\le a\le 1)+sgn(a)(1-\mathds{1}(-1\le a\le 1))$. $\bm{l}$ is the conformable vector of ones. ``$^*$" indicates that $\hat{\bm{\Omega}}_n^{\texttt{KJ}}$ is not positive semi-definite in some Monte Carlo samples.
     \end{tablenotes}
%\end{center}
 \end{table}
 
\subsection{Simulation results: realistic error terms}

In this section, we consider two DGPs of generating more practical errors beyond simple i.i.d. Gaussian noises. For the first DGP, we generate covariates from $\x_{i,n}\sim \mathcal{N}(0, \Sigma) $ and $\w_{i,n} \sim \mathcal{N}(0, \bm{I}_{q_n})$, and then generate random noise from (1) an asymmetric distribution with the density function $ 0.5\phi(\varepsilon\lvert -0.5,0.25)+0.5\phi(\varepsilon\lvert 0.5,1)$; (2) a bimodal distribution  with the density function  $0.5\phi(\varepsilon\lvert -1.5,0.25)+0.5\phi(\varepsilon\lvert 1.5,1)$,
where $\phi(\varepsilon\lvert \mu,\sigma^2)$ denotes the density function of a normal random variable with mean $\mu$ and variance $\sigma^2$. The simulation results are summarized in Supplementary Materials Table \ref{table:newsim-d5-1st-hetero-case34}. We further study this setting with a larger sample size, $n=2,000$. This sample size is closer to the sample size adopted in our case study I. The simulation results under $n=2,000$ are summarized in Supplementary Materials Table \ref{table:sim-hetero-error-n2000}. We also consider the design with both $\beta\neq 0$ and $\gamma \neq 0$. The simulation results are summarized in  Supplementary Materials Table \ref{table:sim-hetero-error-1}--\ref{table:sim-hetero-error-n2000}.

For the second DGP, we consider heteroscedastic errors following the setup in \cite{cattaneo2018inference} with: $\x_{i,n}\sim \mathcal{N}(0, \Sigma) $, $\w_{i,n} = \mathds{1}( \tilde{\w}_{i,n}\geq \Phi^{-1}(0.98) )$ with $\tilde{\w}_{i,n} \sim \mathcal{N}(0, \bm{I}_{q_n})$, and $\varepsilon_{i,n}\sim \mathcal{N}(0,1)$ with $\mathbb{V}[\varepsilon_{i,n}\lvert \x_{i,n},\w_{i,n}]=c_{\varepsilon}(1+(t(x_{1,i,n})+\bm{l}'\w_{i,n})^2/4)$ and $\mathbb{V}[x_{k,i,n}\lvert\w_{i,n}]=c_{x_k}(1+(\bm{l}'\w_{i,n})^2/4)$, where $x_{k,i,n}$ denotes the $k$th component of the vector $\x_{i,n}$. The constants $c_{\varepsilon}$ and $c_{x_k}$ are chosen so that $\mathbb{V}[\varepsilon_{i,n}]=\mathbb{V}[x_{k,i,n}]=1$ and $t(a)=a\mathds{1}(-1\le a\le 1)+\text{sgn}(a)(1-\mathds{1}(-1\le a\le 1))$. $\bm{l}$ is the conformable vector of ones. The simulation results are summarized in Supplementary Materials Table \ref{table:sim-hetero-error-1}.

Table \ref{table:sim-hetero-error-1} shows that, under the second DGP, our proposed method has slightly compromised performance, but still reaches nominal level coverage when $q_n\leq 421$. Table \ref{table:newsim-d5-1st-hetero-case34} and \ref{table:newsim-d5-1st} demonstrate that the performance of our method is robust even when both $\beta\neq 0$ and $\gamma \neq 0$, and $q_n\leq 561$. Table \ref{table:sim-hetero-error-n2000} suggests that when sample size increases, our proposed method has smaller bias and improved coverage probabilities when $q_n=631$. The other considered methods show similar trends to the settings under homoscedastic errors.

%%%%%%%%% another heteroscedastic design
%%%%%%%%% asymmetric and bimodal
%  case 32, case 42
 
 \begin{table}[h!]
\caption{ Simulation results ($d=5, \text{heteroscedasticity}, \text{heterogeneity}, \beta_{(1)}$)}\label{table:newsim-d5-1st-hetero-case34}

%\resizebox{\columnwidth}{!}{%
\begin{adjustbox}{width=.9\textwidth,center}
\begin{tabular}{ccccccc}
     \hline\hline
           & &\multicolumn{5}{c}{  $\beta_{j} = \Phi^{-1}\big( \frac{j}{d+1} \big), \quad j = 1,\ldots, d,\quad \bm{\gamma}_k=1/k, \quad k = 1, \ldots, q_n$} \\[0.15cm]
           \cline{3-7}
           &&\multicolumn{5}{c}{ $ \x_{i,n}\sim N(0, \Sigma), \quad \w_{i,n} \sim N(0, I_{q_n}),\quad      \varepsilon_i \sim f(\varepsilon)=0.5\phi(\varepsilon|-0.5, 0.25)+0.5\phi(\varepsilon|0.5, 1)  $}   \\[0.15cm]
            \cline{3-7} 
             & &  Proposed+KJ & Proposed+HCK &  Proposed + HC3 & Proposed+EW &No adjustment+KJ\\[0.15cm]
       \cline{3-7} 
    $q_n=1$ &Cover  & 0.94(0.01) & 0.94(0.01) &0.94(0.01) & 0.94(0.01)  & 0.94(0.01)
    \\  
    
 & $\sqrt{n}$Bias & -0.02(0.05)
     &-0.02(0.05) & -0.02(0.05) &-0.02(0.05) & 0.05(0.05)
    \\[0.15cm]

   $q_n=141$ &Cover & 0.94(0.01) & 0.94(0.01) & 0.94(0.01) & 0.92(0.01) & 0.94(0.01) \\

     & $\sqrt{n}$Bias & -0.05(0.06) & -0.05(0.06) & -0.05(0.06) & -0.07(0.06) & 0.06(0.06)\\[0.15cm]
     
       $q_n=281$ &Cover  &0.94(0.01) & 0.94(0.01) & 0.93(0.01) & 0.89(0.02)  & 0.94(0.01)
    \\  
    
    & $\sqrt{n}$Bias & 0.06(0.06)
     & 0.06(0.06)& -0.10(0.07) & 0.14(0.06) & 0.07(0.07)
    \\[0.15cm]

   $q_n=421$ &Cover & 0.94(0.01) & 0.93(0.01) & 0.92(0.01) & 0.73(0.02) &0.92(0.01)\\

     & $\sqrt{n}$Bias & -0.09(0.09) & -0.12(0.09)& -0.14(0.09)&0.17(0.09) & -0.10(0.08)\\[0.15cm]
     
      $q_n=561^*$ &Cover & 0.93(0.02) & 0.92(0.01) & 0.92(0.01) &0.58(0.02)&0.92(0.01) \\
    
     & $\sqrt{n}$Bias & -0.12(0.13) & -0.15(0.13) & -0.15(0.12) & -0.20(0.13) &0.14(0.12)\\[0.15cm]
     
      $q_n=631^*$ &Cover & 0.94(0.01) & 0.91(0.01) & 0.90(0.01)  & 0.44(0.02) & 0.91(0.01)\\
    
     & $\sqrt{n}$Bias & 0.17(0.17) & 0.18(0.16) & 0.17(0.13) & 0.25(0.17) & 0.19(0.17)\\[0.15cm]
      \cline{3-7}

     & & \multicolumn{5}{c}{$\x_{i,n}\sim N(0, \Sigma), \quad \w_{i,n} \sim N(0, I_{q_n}),\quad      \varepsilon_i \sim f(\varepsilon)=0.5\phi(\varepsilon|-1.5, 0.25)+0.5\phi(\varepsilon|1.5, 1) $}    \\[0.15cm]
      \cline{3-7} 
      & &  Proposed+KJ & Proposed+HCK  & Proposed + HC3 &  Proposed+EW & No adjustment+KJ\\[0.15cm]
       \cline{3-7} 
       
     $q_n=1$ &Cover  & 0.96(0.01) & 0.95(0.01) & 0.95(0.01) & 0.95(0.01) &0.95(0.01)
    \\

       & $\sqrt{n}$Bias & -0.07(0.10) & -0.08(0.10)& -0.09(0.10) & -0.08(0.10) & -0.08(0.09)
       \\[0.15cm]
       
    $q_n=141$ &Cover & 0.95(0.01) & 0.95(0.01)& 0.95(0.01) & 0.93(0.01)&0.94(0.01) \\

        & $\sqrt{n}$Bias & -0.10(0.12) & -0.10(0.12)&  -0.12(0.12) & -0.13(0.12)&0.12(0.10)\\[0.15cm]
        
          $q_n=281$ &Cover  & 0.95(0.01) & 0.95(0.01) & 0.95(0.01) & 0.92(0.01) &0.94(0.01)
    \\

    & $\sqrt{n}$Bias & -0.11(0.12)
     & -0.11(0.12) & -0.12(0.12) & -0.14(0.12) &-0.13(0.13)
    \\[0.15cm]

   $q_n=421$ &Cover &0.94(0.01) & 0.94(0.01) & 0.92(0.02) &0.78(0.02) &0.93(0.01)\\

     & $\sqrt{n}$Bias & -0.14(0.15) & -0.14(0.15) & -0.16(0.15) & -0.20(0.15)&-0.17(0.16) \\[0.15cm]
     
      $q_n=561^*$ &Cover & 0.94(0.01) &0.92(0.01) & 0.91(0.01) &0.64(0.02) & 0.92(0.01)\\

     & $\sqrt{n}$Bias & -0.20(0.21) & -0.23(0.21)& -0.26(0.24) & 0.42(0.22) &-0.24(0.19)\\[0.15cm]
     
     $q_n=631^*$ &Cover &0.93(0.01) & 0.92(0.01) &  0.90(0.01) & 0.47(0.02) &0.87(0.01)\\
    
     & $\sqrt{n}$Bias & -0.29(0.28) & -0.31(0.28)& -0.33(0.17) & 0.85(0.29)& 0.41(0.29)\\[0.15cm]
      \cline{2-5}

     \hline\hline
     \end{tabular}
     \end{adjustbox}
      \begin{tablenotes}\footnotesize
  \item Note: ``Cover" is the empirical coverage of the 95\% confidence interval for $\beta_{(1)}$ and `` $\sqrt{n}$Bias " captures the root-$n$ scaled Monte Carlo bias for estimating $\beta_{(1)}$. ``$^*$" indicates that $\hat{\bm{\Omega}}_n^{\texttt{KJ}}$ is not positive semi-definite in some Monte Carlo samples.
     \end{tablenotes}
%\end{center}
 \end{table}

 %%%%%%%%% new simulation setup (3) (4)	
 %%%%%%%%% Both beta, gamma not equal to 0
	% case 13, 23
 \begin{table}[h!]
\caption{ Simulation results ($d=5, \text{heterogeneity}, \beta_{(1)}$)}\label{table:newsim-d5-1st}

%\resizebox{\columnwidth}{!}{%
\begin{adjustbox}{width=.8\textwidth,center}
\begin{tabular}{ccccccc}
     \hline\hline
           & &\multicolumn{5}{c}{  $\beta_{j} = \Phi^{-1}\big( \frac{j}{d+1} \big),  \quad j = 1, \ldots, d, \quad \bm{\gamma}_k=1/k, \quad k = 1, \ldots, q_n$} \\[0.15cm]
           \cline{3-7}
           &&\multicolumn{5}{c}{ $ \x_{i,n}\sim \mathcal{N}(0, \Sigma), \quad \w_{i,n} = \mathds{1}( \tilde{\w}_{i,n}\geq \Phi^{-1}(0.98) ) $}   \\[0.15cm]
            \cline{3-7} 
             & &  Proposed+KJ  & Proposed+HCK &  Proposed + HC3 &  Proposed+EW &No adjustment+KJ\\[0.15cm]
       \cline{3-7} 
   $q_n=1$ &Cover  & 0.95(0.01) & 0.96(0.01) &0.98(0.01) & 0.96(0.01)  & 0.96(0.01)
    \\
    
 & $\sqrt{n}$Bias & 0.05(0.05)
     &0.05(0.05) & 0.04(0.05) &0.05(0.05) & 0.05(0.06)
    \\[0.15cm]

   $q_n=141$ &Cover & 0.94(0.01) & 0.93(0.01) & 0.95(0.01) & 0.91(0.01) & 0.94(0.01) \\

     & $\sqrt{n}$Bias & -0.06(0.06) & -0.10(0.06) & -0.06(0.06) & -0.12(0.06) & 0.06(0.06)\\[0.15cm]
     
       $q_n=281$ &Cover  &0.94(0.01) & 0.93(0.01) & 0.94(0.01) & 0.88(0.02)  & 0.94(0.01)
    \\  
    
    & $\sqrt{n}$Bias & 0.07(0.07)
     & 0.10(0.07)& 0.07(0.07) & 0.13(0.07) & 0.07(0.08)
    \\[0.15cm]

   $q_n=421$ &Cover & 0.94(0.01) & 0.93(0.01) & 0.94(0.01) & 0.78(0.02) &0.91(0.01)\\

     & $\sqrt{n}$Bias & 0.09(0.09) & 0.12(0.09)& 0.09(0.09)&0.14(0.09) & 0.10(0.09)\\[0.15cm]
     
      $q_n=561^*$ &Cover & 0.93(0.02) & 0.89(0.01) & 0.92(0.01) &0.59(0.02)&0.90(0.01) \\
    
     & $\sqrt{n}$Bias & -0.10(0.13) & -0.15(0.13) & 0.13(0.12) & -0.19(0.13) &0.14(0.12)\\[0.15cm]
     
      $q_n=631^*$ &Cover & 0.93(0.01) & 0.92(0.01) & 0.92(0.01)  & 0.43(0.02) &0.82(0.01)\\
    
     & $\sqrt{n}$Bias & 0.19(0.17) & 0.17(0.16) & 0.16(0.13) & -0.30(0.18) & 0.21(0.16)\\[0.15cm]
      \cline{3-7}

     & & \multicolumn{5}{c}{$\x_{i,n} = \mathds{1}(\tilde{\x}_{i,n} > 0), \quad \w_{i,n} \sim \mathcal{N}(0, I) $}    \\[0.15cm]
      \cline{3-7} 
      & &  Proposed+KJ & Proposed+HCK  & Proposed+HC3 & Proposed+EW &No adjustment+KJ\\[0.15cm]
       \cline{3-7} 
    $q_n=1$ &Cover  & 0.96(0.01) & 0.95(0.01) & 0.95(0.01) & 0.96(0.01) &0.95(0.01)
    \\

       & $\sqrt{n}$Bias & 0.08(0.09) & 0.09(0.09)& 0.09(0.09) & 0.08(0.09) &-0.09(0.09)
       \\[0.15cm]
       
    $q_n=141$ &Cover & 0.95(0.01) & 0.95(0.01)& 0.95(0.01) & 0.94(0.01)&0.95(0.01) \\

        & $\sqrt{n}$Bias & -0.10(0.11) & -0.10(0.11)&  -0.11(0.11) & -0.11(0.11)&0.11(0.11)\\[0.15cm]
        
          $q_n=281$ &Cover  & 0.95(0.01) & 0.95(0.01) & 0.95(0.01) & 0.92(0.01) &0.95(0.01)
    \\

    & $\sqrt{n}$Bias & 0.11(0.12)
     & 0.11(0.12) & 0.12(0.11) & 0.14(0.12) &0.12(0.12)
    \\[0.15cm]

   $q_n=421$ &Cover &0.94(0.01) & 0.94(0.01) & 0.93(0.02) &0.73(0.02) &0.94(0.01)\\

     & $\sqrt{n}$Bias & -0.14(0.15) & -0.14(0.15) & -0.14(0.14) & -0.22(0.15)&-0.15(0.15) \\[0.15cm]
     
      $q_n=561$ &Cover & 0.94(0.01) &0.93(0.01) & 0.92(0.01) &0.61(0.02) & 0.93(0.01)\\

     & $\sqrt{n}$Bias & 0.19(0.19) & 0.21(0.19)& 0.20(0.16) & 0.44(0.20) &-0.20(0.19)\\[0.15cm]
     
     $q_n=631^*$ &Cover &0.94(0.01) & 0.93(0.01) &  0.90(0.01) & 0.50(0.02) &0.91(0.01)\\
    
     & $\sqrt{n}$Bias & 0.24(0.26) & 0.30(0.28)& 0.35(0.21) & 0.72(0.28)&0.32(0.28)\\[0.15cm]
      \cline{2-5}

     \hline\hline
     \end{tabular}
     \end{adjustbox}
      \begin{tablenotes}\footnotesize
  \item Note: ``Cover" is the empirical coverage of the 95\% confidence interval for $\beta_{(1)}$ and `` $\sqrt{n}$Bias " captures the root-$n$ scaled Monte Carlo bias for estimating $\beta_{(1)}$. ``$^*$" indicates that $\hat{\bm{\Omega}}_n^{\texttt{KJ}}$ is not positive semi-definite in some Monte Carlo samples.
     \end{tablenotes}
%\end{center}
 \end{table}

%%%%% n = 2000
%%%%%%%%% heteroscedastic
   \begin{table}[h!]
\caption{ Simulation results ($d=5, \text{heteroscedasticity}, \text{heterogeneity},  \beta_{(1)}, n = 2000$)}\label{table:sim-hetero-error-n2000}

%\resizebox{\columnwidth}{!}{%
\begin{adjustbox}{width=.8\textwidth,center}
\begin{tabular}{ccccccc}
     \hline\hline
           & &\multicolumn{5}{c}{  $\beta_{j} = \Phi^{-1}\big( \frac{j}{d+1} \big), \quad \x_{i,n}\sim \mathcal{N}(0, \Sigma), \quad \w_{i,n} = \mathds{1}( \tilde{\w}_{i,n}\geq \Phi^{-1}(0.98) )$, \quad j = 1,\ldots, d} \\[0.15cm]
           \cline{3-7}
           &&\multicolumn{5}{c}{ $ \bm{\gamma}_n=0 $}   \\[0.15cm]
            \cline{3-7} 
             & &  Proposed+KJ  & Proposed+HCK  & Proposed+HC3 & Proposed+EW &No adjustment+KJ\\[0.15cm]
       \cline{3-7} 
    $q_n=1$ &Cover  & 0.95(0.01) & 0.95(0.01) & 0.95(0.01) & 0.94(0.01)  & 0.95(0.01)
    \\
    
 & $\sqrt{n}$Bias & 0.02(0.03)
     & 0.03(0.03) & 0.03(0.03) & 0.03(0.03) & 0.03(0.03)
    \\[0.15cm]

   $q_n=141$ &Cover & 0.95(0.01) & 0.95(0.01) & 0.94(0.01) & 0.91(0.02) & 0.94(0.01) \\

     & $\sqrt{n}$Bias & -0.04(0.04) & -0.04(0.04) & -0.04(0.04) & -0.06(0.04) & -0.04(0.04)\\[0.15cm]
     
       $q_n=281$ &Cover  &0.94(0.01) & 0.94(0.01) & 0.93(0.01) & 0.88(0.01)  & 0.93(0.01)
    \\  
    
    & $\sqrt{n}$Bias & 0.04(0.05)
     &0.05(0.05)& 0.07(0.05)  &-0.12(0.05) & 0.07(0.05)
    \\[0.15cm]

   $q_n=421$ &Cover & 0.94(0.01) & 0.94(0.01) & 0.91(0.01) & 0.85(0.01) & 0.93(0.01)\\

     & $\sqrt{n}$Bias & -0.05(0.05) & -0.05(0.05)& -0.09(0.05) & -0.15(0.05) & -0.07(0.05)\\[0.15cm]
     
      $q_n=561$ &Cover & 0.94(0.01) & 0.93(0.01) & 0.91(0.01) & 0.82(0.01)& 0.91(0.01) \\
    
     & $\sqrt{n}$Bias & -0.05(0.05) & -0.06(0.05) & -0.10(0.05)  & -0.18(0.05) & -0.11(0.05)\\[0.15cm]
     
      $q_n=631^*$ &Cover & 0.93(0.01) & 0.91(0.01) & 0.90(0.01) & 0.78(0.01) & 0.90(0.01) \\
    
     & $\sqrt{n}$Bias & -0.07(0.05) & -0.09(0.05) & -0.12(0.05) & -0.22(0.05) & -0.14(0.05)\\[0.15cm]
      \cline{3-7}

     & & \multicolumn{5}{c}{$\gamma_k = 1/k,\quad k = 1,\ldots, q_n$}    \\[0.15cm]
      \cline{3-7} 
      & &  Proposed+KJ & Proposed+HCK  &   Proposed+HC3 & Proposed+EW &No adjustment+KJ\\[0.15cm]
       \cline{3-7} 
     $q_n=1$ &Cover  & 0.96(0.01) & 0.96(0.01) &  0.95(0.01) & 0.94(0.01) & 0.95(0.01)
    \\

       & $\sqrt{n}$Bias & 0.04(0.05) & 0.05(0.05)& 0.05(0.05) & 0.05(0.05) & 0.05(0.05)
       \\[0.15cm]
       
    $q_n=141$ &Cover & 0.95(0.01) & 0.95(0.01)& 0.94(0.01) & 0.93(0.01)& 0.94(0.01) \\

        & $\sqrt{n}$Bias & 0.07(0.08) & 0.08(0.08)& 0.08(0.08) & 0.11(0.08)&  0.08(0.08)\\[0.15cm]
        
          $q_n=281$ &Cover  & 0.94(0.01) & 0.94(0.01) & 0.94(0.01) & 0.92(0.01) & 0.94(0.01)
    \\

    & $\sqrt{n}$Bias & 0.08(0.08)
     & 0.10(0.08) & 0.13(0.08) & 0.14(0.08) & 0.08(0.08)
    \\[0.15cm]

   $q_n=421$ &Cover & 0.94(0.01) & 0.93(0.01) & 0.92(0.01)  & 0.88(0.01) & 0.92(0.01)\\

     & $\sqrt{n}$Bias & 0.08(0.08) & 0.12(0.08) & 0.14(0.08) & 0.18(0.08)& 0.15(0.08)\\[0.15cm]
     
      $q_n=561^*$ &Cover & 0.94(0.01) & 0.91(0.01) & 0.90(0.01) & 0.83(0.02) & 0.90(0.01)\\

     & $\sqrt{n}$Bias & 0.09(0.08) & 0.15(0.08)& 0.17(0.07) & 0.23(0.08) &  0.18(0.08)\\[0.15cm]
     
     $q_n=631^*$ &Cover & 0.92(0.01) & 0.90(0.01) & 0.86(0.01)  & 0.76(0.01) & 0.85(0.01)\\
    
     & $\sqrt{n}$Bias & 0.12(0.08) & 0.17(0.08)&  0.20(0.05) & 0.28(0.08)&  0.23(0.08)\\[0.15cm]
      \cline{2-5}

     \hline\hline
     \end{tabular}
     \end{adjustbox}
      \begin{tablenotes}\footnotesize
  \item Note: ``Cover" is the empirical coverage of the 95\% confidence interval for $\beta_{(1)}$ and `` $\sqrt{n}$Bias " captures the root-$n$ scaled Monte Carlo bias for estimating $\beta_{(1)}$. In this heteroscedastic design, $\tilde{\w}_{i,n} \sim \mathcal{N}(0, \bm{I}_{q_n})$, $\varepsilon_{i,n}\sim \mathcal{N}(0,1)$,  $\mathbb{V}[\varepsilon_{i,n}\lvert \x_{i,n},\w_{i,n}]=c_{\varepsilon}(1+(t(x_{1,i,n})+\bm{l}'\w_{i,n})^2/4)$, and $\mathbb{V}[x_{k,i,n}\lvert\w_{i,n}]=c_{x_k}(1+(\bm{l}'\w_{i,n})^2/4)$, where $x_{k,i,n}$ denotes the $k$th component of the vector $\x_{i,n}$. The constants $c_{\varepsilon}$ and $c_{x_k}$ are chosen so that $\mathbb{V}[\varepsilon_{i,n}]=\mathbb{V}[x_{k,i,n}]=1$ and $t(a)=a\mathds{1}(-1\le a\le 1)+sgn(a)(1-\mathds{1}(-1\le a\le 1))$. $\bm{l}$ is the conformable vector of ones. ``$^*$" indicates that $\hat{\bm{\Omega}}_n^{\texttt{KJ}}$ is not positive semi-definite in some Monte Carlo samples.
     \end{tablenotes}
%\end{center}
 \end{table}

 \subsection{Additional analysis for case study I}
 In this section, we revisit case study I with a much smaller model that only includes the main effects. The results are summarized in Table \ref{table:ask-amount-main}. Table \ref{table:ask-amount-main} shows that, overall, the results under a smaller model do not change substantively. But ``asking 25\% more" no longer has a significant impact on donation amount even without calibration. 
 
 \begin{table}[h!]
    \centering
 \begin{tabular}{ccccc}
    \hline
     \hline\\[-2ex] 
      Method & Policies(Ask amount)  & Est (95\% CI)  & $p$-value  &   \\
  \\[-2ex] 
    \hline
      \\[-2ex] 
 Uncalibrated & Same &  0.70(0.10,  1.29) & 0.023*  &  \\
    \\[-2ex] 
      %\hline
  & 25\% more &  0.67(-0.04, 1.37)  &  0.065 &\\
  \\[-2ex] 
  & 50\% more &  0.38(-0.21, 0.96)  & 0.205  &  \\
  \\[-2ex]
  \hline
  \\[-2ex]
Calibrated  &  Same &  0.66(0.07, 1.24) & 0.026*   &  \\
  \\[-2ex] 
     \hline\hline
    \end{tabular}
    \caption{Uncalibrated and calibrated results under a smaller model with main effects only ($n=7,938, p = 53$). Estimated treatment effects (Est), 95\% confidence intervals (95\% CI), and two-sided $p$-values for the three  ``ask amount" policies. ``Uncalibrated" refers to the study results obtained without any adjustment, and the confidence intervals are constructed based on normal approximation with the estimated covariance matrix $ \hat{\bm{\Omega}}^{\texttt{KJ}}_n$.   ``Calibrated" refers to our proposed methodology. The computational time is $533$ seconds on a Lenovo NeXtScale nx360m5 node (24 cores per node) equipped with Intel Xeon Haswell processor. \label{table:ask-amount-main}}
\end{table}

\section{Extension to regression models with fixed effects}
 
As stated in the main manuscript, our approach extends to linear panel data models with fixed effects. We shall briefly discuss this connection below. Because it is a common practice to include the subscript $t$ to denote time in panel data analyses, to avoid using triple subscript, we drop subscript $n$ in all considered random variables in the discussion below. 

Suppose we have access to one panel data with cross-sectional observations denoted by $i \in\mathcal{N}=\{ 1, \ldots, N\}$ and time periods $t\in\mathcal{T}=\{1, \ldots, T\}$. 
Consider the following fixed effects panel data model
\begin{align*}
    y_{it} = \bbeta'\x_{it} + c_i + e_{d_{it}} + u_{it},\quad i=1, \ldots,N, \quad t = 1, \ldots, T, 
\end{align*}
where $c_i$ is an unobserved effect that varies across sections but is assumed to be constant over time, $y_{it}\in\mathbb{R}$ is the observed outcome, $\x_{it}\in\mathbb{R}^{d\times 1}$ contains the policy variables of interest, and error terms $u_{it}$'s are uncorrelated conditional on $\x_{it}$ and $d_{it}$.  $e_{d_{it}}$ is an unobserved effect indexed by an observed indexing variable $d_{it}\in \{ 1, \ldots, G\}$, and is assumed to be constant across all observations that share the same value of $d_{it}$. When $e_{d_{it}}=0$, this model reduces to the one-way fixed effects model studied in \cite{stock2008heteroskedasticity}, otherwise the above model coincides with the one studied in \cite{verdier2020estimation}.

To concretely introduce the connection of the above model and our model setup, consider the case when  $e_{d_{it}}\neq 0 $, we stack the data over cross-sectional observations and time periods. Define 
\begin{align*}
    \mathbf{y} = & ( y_{11}, \ldots, y_{1T}, y_{21}, \ldots, y_{2T}, \ldots, y_{N1}, \ldots, y_{NT} )'\in\mathbb{R}^{NT\times 1}, \\
    \x = & ( \x_{11}, \ldots, \x_{1T}, \x_{21}, \ldots, \x_{2T}, \ldots, \x_{N1}, \ldots, \x_{NT} )'\in\mathbb{R}^{NT\times d},\\
    \w = & \Big( \mathbf{g}_{1}, \mathbf{g}_{2}\Big)\in\mathbb{R}^{NT \times (N +G)} ,   \\
    \mathbf{g}_{1} =  & \big(\mathbf{1}_{(i=j)} \big)^{j\in\mathcal{N}}_{(i,t)\in\mathcal{N}\times \mathcal{T} }, \quad \mathbf{g}_{2} =  \big(\mathbf{1}_{(d_{it}=d)} \big)^{d\in\{1,\ldots, G\} }_{(i,t)\in\mathcal{N}\times \mathcal{T} } ,\\
    \bgamma_n =& ( c_1, \ldots, c_N, e_1, \ldots, e_{G} )'\in\mathbb{R}^{(N+G) \times 1}, \\
    \mathbf{u}=& ( u_{11}, \ldots, u_{1T}, u_{21}, \ldots, u_{2T}, \ldots, u_{N1}, \ldots, u_{NT} )'\in\mathbb{R}^{NT\times 1} . 
\end{align*}
With the above notations, the fixed effects panel data model can be written as the following  
\begin{align}
\mathbf{y} = \x\bbeta +\w\bgamma_n + \mathbf{u}.
\end{align}
This indicates that our approach also goes through in linear panel data models, as long as we can construct an estimator of $\bbeta$ that converges to a Gaussian distribution with its covariance matrix being consistently estimated.

\cite{jochmans2020heteroscedasticity} has shown that the covariance matrix estimator $\hat{\bm{\Omega}}^{\texttt{KJ}}_n$ remains consistent in one-way fixed effect panel data regression models when $e_{d_{it}}=0$. This suggests that our approach can be naturally extended to make inference on multiple best policies in one-way fixed effect models. In addition, \cite{cattaneo2018inference} have shown that the covariance matrix estimator $ \hat{\bm{\Omega}}^{\texttt{HCK}}_n$ is consistent in both one-way and two-way fixed effect panel data regression models. Since our resampling based approach only requires a consistent covariance matrix estimator to calibrate multiple best policy effects, this suggests that in two-way fixed effect models, what our approach can be adopted when using $\hat{\bm{\Omega}}^{\texttt{HCK}}_n$ to estimate the covariance matrix of $\hat{\bbeta}$. Lastly, we note that our Assumption 1 in the main manuscript requires error terms to be conditionally uncorrelated within each observation $i$. This condition does rule out dynamic models as those discussed in \cite{verdier2020estimation}.

\end{document}